\documentclass[11pt]{article}

\usepackage{geometry,showlabels, fullpage, authblk}
\usepackage{xr-hyper}
\usepackage{times,enumitem}
\usepackage{hyperref}[]
\hypersetup{
    colorlinks=true,
    linkcolor=blue,
    filecolor=magenta,      
    urlcolor=cyan,
      citecolor=blue,
    }

\usepackage{url}

\usepackage{bigints}
\usepackage{amsfonts,amsmath,amssymb,amsthm, bbm}
\usepackage{wrapfig}
\usepackage{verbatim,float,url}
\usepackage{graphicx}
\usepackage{subcaption}
\usepackage[round]{natbib}   % omit 'round' option if you prefer square brackets
\usepackage[english]{babel}
\usepackage{cancel}
\usepackage{color}
\usepackage[dvipsnames]{xcolor}
\usepackage{cleveref}

\usepackage{mathtools}

\newtheorem{theorem}{Theorem}
\newtheorem{lemma}[theorem]{Lemma}

\usepackage{mdwlist}	
\theoremstyle{remark}
\newtheorem{remark}{Remark}
\theoremstyle{definition}
\newtheorem{definition}{Definition}

\usepackage{tikz}  
\usetikzlibrary{trees}
\usetikzlibrary{arrows.meta}
\usetikzlibrary{decorations.pathreplacing}

\graphicspath{{figs/}}

\usepackage{algorithm}% http://ctan.org/pkg/algorithms
\usepackage{algpseudocode}% http://ctan.org/pkg/algorithmicx

\algnewcommand\algorithmicinput{\textbf{INPUT:}}
\algnewcommand\INPUT{\item[\algorithmicinput]}
\algnewcommand\algorithmicoutput{\textbf{OUTPUT:}}
\algnewcommand\OUTPUT{\item[\algorithmicoutput]}

\usepackage{mathtools}
\DeclarePairedDelimiter\ceil{\lceil}{\rceil}
\DeclarePairedDelimiter\floor{\lfloor}{\rfloor}

\DeclareMathOperator*{\argmax}{arg\,max}

\theoremstyle{plain}
\newtheorem{assumption}{Assumption}

\title{Optimal nonparametric change point detection and localization}
\author[1]{Oscar Hernan Madrid Padilla}
\author[2]{Yi Yu}
\author[3]{Daren Wang}
\author[4]{Alessandro Rinaldo}

\affil[1]{\small Department of Statistics, UC Berkeley}
\affil[2]{\small School of Mathematics, University of Bristol}
\affil[3]{\small Department of Statistics, University of Chicago}
\affil[4]{\small Department of Statistics and Data Science, Carnegie Mellon University}

\begin{document}

\date{\today}

\maketitle

\begin{abstract}
We study change point detection and localization for univariate data in fully nonparametric settings in which, at each time point, we acquire an i.i.d. sample from an unknown distribution.  We quantify the magnitude of  the distributional changes at the change points using the  Kolmogorov--Smirnov distance. We allow all the relevant parameters -- the minimal spacing between two consecutive change points, the minimal magnitude of the changes in the Kolmogorov--Smirnov distance, and the number of sample points collected at each time point -- to change with the length of time series. We generalize the renowned binary segmentation \citep[e.g.][]{ScottKnott1974} algorithm and its variant, the wild binary segmentation of \cite{fryzlewicz2014wild}, both originally designed for univariate mean change point detection problems, to our nonparametric settings and exhibit rates of consistency for both of them. In particular, we prove that the procedure based on wild binary segmentation is nearly minimax rate-optimal. We further demonstrate a phase transition in the space of model parameters that separates parameter combinations for which consistent localization is possible from the ones for which this task is statistical unfeasible. Finally, we provide extensive numerical experiments to support our theory. R code  is available at \url{https://github.com/hernanmp/NWBS}.
%We show phase transition and minimax optimality in the nonparametric change point detection and localisation problems by allowing all parameters -- the minimal spacing between two consecutive change points, the minimal jump size in terms of the Kolmogorov--Smirnov distance, and the order of sample points collected at every time point -- to change with the total number of time points $T$.  We do not impose further assumptions on the distribution functions in order to gain maximum flexibility.  

\end{abstract}

\section{Introduction}
\label{sec:introduction}

Change point analysis is a well-established topic in statistics that is concerned with detecting and localizing abrupt changes in the data generating distribution in time series data. Initiate during World War II \citep[see, e.g.,][]{Wald1945}, the field of change point analysis  has produced a large literature as well as hosted well-established methods for statistical inference available to practitioners. These techniques are now widely used to address important real life problems in a variety of disciplines, including, for example,  biology \citep{fan2015identifying,jewell2018fast}, speech  recognition  \citep{fox2011sticky,chowdhury2012bayesian}, social networks \citep{liu2013change},  climate  \citep{itoh2010change}, and   financial  data \citep{preuss2015detection,russell2018breaks}.

The theoretical understanding of the statistical challenges associated to change point problems has also progressed considerably. The initial groundbreaking results of \cite{Yao1988}, \cite{YaoAu1989},  \cite{YaoDavis1986} from the 1980s, aimed at studying change point detection for a univariate piecewise constant signal, have now been extended in several ways. For instance, \cite{fryzlewicz2014wild} and \cite{FrickEtal2014}, among others, proposed computationally-efficient methods dealing with the situations with potentially multiple mean change points \citep[see also,][]{wang2018univariate}. More recently, \cite{pein2017heterogeneous} constructed a method  that  can handle  mean and variance  changes simultaneously. \cite{Cho2015}, \cite{ChoFryzlewicz2015} and \cite{wang2016high}  studied high-dimensional mean change point detection problems. A different line of work,  including efforts by 
 \cite{AueEtal2009}, \cite{avanesov2016change} and \cite{wang2017optimal}, has investigated scenarios where covariance matrices change.  \cite{CribbenYu2017}, \cite{LiuEtal2018} and \cite{wang2018optimal}, among others, inspected dynamic network change point detection problems.

Most of the existing theoretical frameworks for statistical analysis of change point problems, however, largely rely on strong modeling assumptions of parametric nature that may be inadequate to capture the inherent complexity of modern, high-dimensional datasets. Indeed, 
the statistical literature on nonparametric change point analysis is surprisingly limited compared to its parametric counterpart.  Among the few nonparametric results, \cite{carlstein1988nonparametric} considered the scenario where there is at most one change point;  \cite{hawkins2010nonparametric} proposed a Mann--Whitney-type statistics to conduct online change point detection; \cite{matteson2014nonparametric} established the consistency of change point estimators based on statistics originally introduced in \cite{rizzo2010disco};  \cite{zou2014nonparametric} proposed a nonparametric multiple change point detection method which came with some consistency guarantees; more recently,  \cite{padilla2018sequential} proposed  an algorithm  for  nonparametric change point detection based on the  Kolmogorov--Smirnov  statistic; \cite{fearnhead2018changepoint}  a  mean change point detection robust to outliers; and \cite{vanegas2019multiscale} constructed a multiscale method for detecting changes in fixed quantiles of the distributions.%, for instance, for situations where the median is piecewise constant.

In this paper we advance both the theory and methodology of nonparametric change point analysis by presenting two computationally-efficient procedures for univariate change point localization that are provably consistent and, in one case, nearly minimax optimal. Our analysis builds upon various recent contributions in the literature on parametric and high-dimensional change point analysis but allows for a fully nonparametric change point model. The pioneering work of \cite{zou2014nonparametric} is, to the best of our knowledge, one of very few examples yielding a procedure for nonparametric change point with provable guarantees.  A detailed comparison between our results and the ones of \cite{zou2014nonparametric} will be given in Section~\ref{sec-comp}.

\subsection*{Problem formulation}	

We describe the change point model we are going to consider next. Our settings and notation are fairly standard, with one crucial difference from most of the contributions in the field: the changes in the underlying distribution at the change points are not parametrically specified but are instead quantified through a nonparametric measure of distance between distributions. This feature renders our methods and analysis applicable to a very broad range of change point problems. 

\begin{assumption}[Model]
\label{as1}
	Let $\{Y_{t, i},\, t = 1, \ldots, T,\, i = 1, \ldots, n_t\} \subset \mathbb{R}$ be a collection of independent random variables such that
		\begin{equation}
			\label{eq-model}
				Y_{t, i} \sim F_t, %\quad t = 1, \ldots, T,
		\end{equation}
		where $F_1, \ldots, F_T$ are cumulative distribution functions (CDFs).% with its density function $f_t$.
	
	Let $\{\eta_k\}_{k = 0}^{K+1} \subset \{0, 1, \ldots, T\}$ be a collection of change points with $0 = \eta_0 < \eta_1 < \ldots < \eta_K < \eta_{K+1} = T$ such that
	\[
F_{\eta_0} = F_{\eta_1}   \quad \text{and} \quad  F_{\eta_k} \neq F_{\eta_k + 1}, \quad k = 1,\ldots,K,
	\]
and
		\[
		F_{\eta_k + 1} = \ldots = F_{\eta_{k+1}}, \quad k = 0, \ldots, K.
		\]
		
Then we	assume that the minimal spacing $\delta$ and the jump size $\kappa$ satisfy	
		\[
		\min_{k = 1, \ldots, K+1} \{\eta_k - \eta_{k-1}\} \geq \delta > 0,
		\]
		and
		\begin{equation}\label{eq-as1-kappa}
		\min_{k = 1, \ldots, K+1} \sup_{z \in \mathbb{R}} \bigl|F_{\eta_{k}}(z) - F_{\eta_{k-1}}(z)\bigr| = \min_{k = 1, \ldots, K+1} \kappa_k = \kappa > 0.
		\end{equation}
		%using the convention  that $F_0 = F_1$.
Furthermore, we set
	\[
	n_{\min} = \min_{t = 1, \ldots, T} n_t, \quad \mbox{and} \quad n_{\max} = \max_{t = 1, \ldots, T} n_t.
	\]
	%and assume that $n_{\max} \asymp n_{\min} \asymp n$.  %\textcolor{Orchid}{(this is needed in proving the last bit of Step 1 in Lemma~\ref{lem-21}.)}
		
\end{assumption}

{\bf Remark.} According to our notation, which is consistent with the literature on change point analysis, the $k$th change occurs at time $\eta_k + 1$.\\

In \eqref{eq-model},  we allow for multiple observations $n_t$ to be collected at each time $t$. This  generalizes the classical  change point detection framework where $n_t = 1$  for all $t$; see, for instance,  \cite{zou2014nonparametric}.  This flexibility is inspired by the recent interest in anomaly detection problems  where 
multiple  observations can be measured  in a fixed time; see the work in \cite{chan2014distribution}, \cite{reinhart2014spatially} and  \cite{padilla2018sequential}. Our results remain valid even if $n_t = 1$ for all $t$.

We quantify the magnitude of the distributional changes using the Kolmogorov--Smirnov (KS) distance between distribution functions at consecutive change points; see  \eqref{eq-as1-kappa}.  The KS distance is a natural and widely used metric for univariate probability distributions. It is well-known that convergence in the KS distance is stronger than weak convergence but weaker than convergence in the total variation distance and also, provided that the distributions admit bounded Lebesgue densities, in the $L_1$-Wasserstein distance. Reliance on the KS metric buys a great deal of flexibility in our results, which hold without virtually any assumptions on the underlying distribution functions $\{F_t\}$. In particular, they can be continuous, discrete or of mixed type.

 % It is a \textcolor{red}{slightly stronger ATUALLY, KS  IS WEAKER THAN BOTH TV AND WASSERSTEIN} distance than . Thus, for any probability measures $\mu, \nu: \mathcal{F} \to [0, 1]$, where $\mathcal{F}$ is the Borel $\sigma$-algebra, it holds that
%	\[
%		\sup_{z \in \mathbb{R}} \bigl|\mu((-\infty, z]) - \nu((-\infty, z])\bigr| \leq \sup_{A \in \mathcal{F}} \bigl|\mu(A) - \nu(A)\bigr|
%	\]
%	and
%	\[
%		\sup_{z \in \mathbb{R}} \bigl|\mu((-\infty, z]) - \nu((-\infty, z])\bigr| \leq C\sqrt{\sup\left\{\left|\int f\, d\mu - \int f\, d\nu\right|: \, f \%mbox{ is 1-Lipschitz}\right\}},
%	\]
%	both of which can be derived from the definitions.  Despite the aforementioned, change point detection coupled with the KS distance provides with the generality and robustness which are beyond the reach of those parametric change point detection methods.
	
%\end{remark}

%\begin{remark}
%Furthermore, we emphasize that throughout the paper, we do not make any additional assumptions on the continuity or the support	 of the distribution functions $\{F_t\}$.  This provides us with great generality that is achieved by chaining arguments used in Lemma~\ref{lem:concetration}.
%\end{remark}

%It can be seen in Assumption~\ref{as1} that the model is completely characterised by the total number of time points $T$, the minimal spacing between two consecutive change points $\delta$, the minimal jump size in terms of the Kolmogorov--Smirnov distance $\kappa$, and the order of data points at every time point $n$.  

The nonparametric change point model defined above in Assumption~\ref{as1} is specified by few key parameters: the minimal spacing between two consecutive change points $\delta$, the minimal jump size in terms of the Kolmogorov--Smirnov distance $\kappa$, and the number $n_t$ of data points acquired at each time $t$. We adopt a high-dimensional framework in which all these quantities are allowed to change as functions of the total number of time points $T$. %For ease of readability we do not make 
For technical reasons, we further assume that, in such asymptotic regime, $n_{\max} \asymp n_{\min} \asymp n$

We consider the change point localization  problem of establishing consistent change point estimators $\{\hat{\eta}_k\}_{k = 1}^{\widehat{K}}$ of the true change points. These are  measurable functions of the date and return an increasing sequence of random time points  $\hat{\eta}_1 < \ldots < \hat{\eta}_{\widehat{K}}$, such that, as $T \to \infty$,  the following event holds with probability tending to $1$:
	\[
		\widehat{K} = K \quad \text{and} \quad \max_{k = 1, \ldots, K}\bigl|\hat{\eta}_k - \eta_k \bigr| \leq \epsilon, 
	\]
where $\epsilon = \epsilon(T)$ is such that $\lim_{T \to \infty} \epsilon/T = 0$
	
	Throughout the rest of the paper, we  refer to the quantity $\epsilon/T$ as the {\it localization rate.} Our goal is to obtain change point estimators that, under the weakest possible conditions, are guaranteed to yield the smallest possible $\epsilon_T$. In fact, the change point estimators we produce satisfy a stronger property that $\lim_{T \to \infty} \epsilon/\delta = 0$.

%\subsection{Notation}	
	
\subsection*{Summary of results}	
	
We will show that under Assumption~\ref{as1}, the hardness of the change point localization task is fully captured  by the quantity 
\begin{equation}\label{eq:snr}
\kappa\sqrt{\delta n},
\end{equation}
 which can be thought of as some form of signal-to-noise ratio. We may interpret this quantity by drawing a parallel with the localization task in classical univariate mean  change point problem involving a signal with piecewise constant mean corrupted  by additive noise. In that problem \cite{wang2018univariate} show that the relevant signal-to-noise ratio is $\kappa\sqrt{\delta}/\sigma$, where $\sigma$ is an upper bound on the variance of the noise and $\kappa$ and $\delta$ are the smallest size of the jump and the minimal spacing between change points. Since, under Assumption~\ref{as1}, we have $\Theta(n)$ data points at every time point, and if we take the sample mean at every time point, then we obtain a bound on the variance of order $n^{-1/2}$, yielding \eqref{eq:snr}.

We list our contributions as follows.
\begin{itemize}
	\item We demonstrate the existence of a phase transition for the localization task in terms of the signal-to-noise ratio $\kappa\sqrt{\delta n}$. Specifically,  We show that in the low signal-to-noise ratio regime $\kappa\sqrt{\delta n} \lesssim 1$, no algorithm is guaranteed to be consistent.  We also show that if  $\kappa\sqrt{\delta n} \geq \zeta_T$, where $\zeta_T$ is any sequence such that $\lim_{T \to \infty}\zeta_T = \infty$, then a minimax lower bound of the localization error rate is $(n\kappa^2 T)^{-1}$.
	
	\item We develop Kolmogorov--Smirnov versions of binary segmentation (BS)  from \cite{ScottKnott1974}  (see  Algorithm~\ref{alg:BS}), and of wild binary segmentation  (WBS)  from  \cite{fryzlewicz2014wild} (see Algorithm~\ref{algorithm:WBS}). We show that under suitable conditions, both methods are consistent.  In addition, Algorithm~\ref{algorithm:WBS} is proved to be nearly minimax rate-optimal, save for  logarithm factors, in terms of both the required signal-to-noise ratio (see Assumption~\ref{assumption-3}) and the localization error rate (see Theorem~\ref{thm-wbs}). Interestingly,  for the lower bounds  on the signal-to-noise ratio and the localization error rate, our rates match those derived for the mean change point localization problem under sub-Gaussian noise; see, e.g., \cite{wang2018univariate}. 
	%\citep[e.g.][]{ScottKnott1974} (see Algorithm~\ref{alg:BS}) and wild binary segmentation \citep{fryzlewicz2014wild} (see Algorithm~\ref{algorithm:WBS}).  We show that under suitable conditions, both methods are consistent.  In addition, Algorithm~\ref{algorithm:WBS} is proved to be optimal off by a logarithm factor in terms of both the required signal-to-noise ratio (see Assumption~\ref{assumption-3}) and the localisation error rate (see Theorem~\ref{thm-wbs}).

	\item  We  provide extensive comparisons of our algorithms and theoretical guarantees with several competing methods and results from the literature. See \Cref{sec-comp} and \Cref{sec-num}. In particular, our simulations indicate that our procedures perform very well across a variety of scenarios.  
	% We achieve the same rates, for the lower bound of required signal-to-noise ratio and the localisation error rate, as those in the mean change point detection problem assuming that the noise variables are sub-Gaussian.  This means our methods are not only robust but also optimal.
	
	%We achieve the same rates for the lower bound of required signal-to-noise ratio and the localisation error rate as those in the mean change point detection problem assuming the noise variables being sub-Gaussian.  This means our methods are not only robust but also optimal.
	%% Assuming that the noise variables being sub-Gaussian
\end{itemize}

We point out that, although in deriving the theoretical guarantees for our methodologies we rely on techniques proposed in existing work, namely \cite{venkatraman1992consistency} and \cite{fryzlewicz2014wild}, our results deliver significant improvements in two aspects. First, the extension of the analyses of both BS and WBS to the nonparametric 
settings -- in which the magnitude of the the distributional changes is measures using the KS metric -- requires novel and non-trivial arguments, especially to quantify the order of the stochastic fluctuations of the associated CUSUM statistics. Secondly, the analysis of BS and WBS given in \cite{fryzlewicz2014wild} does not extend to our problem because the proofs of those results  do not track  all the relevant model parameters and, more importantly, suffer from critical errors. Thus, we have derived our results largely from scratch.

\subsection*{Outline }

The  the paper is organized as follows.  In Section~\ref{sec-met}, we detail the two nonparametric change point detection methods, the theoretical results of which are exhibited in Theorems~\ref{thm-bs} and \ref{thm-wbs}, respectively.  The phase transition and the lower bound of the localisation error rate are demonstrated in Section~\ref{sec-phase}.  A further comparison of the two proposed methods, and with other nonparametric and parametric change point detection algorithms, is discussed in Section~\ref{sec-comp}.  Extensive numerical validation of  our methods is presented in Section~\ref{sec-num}.  %This paper is concluded with some discussions collected in Section~\ref{sec-discussion}.

\section{Methodology}\label{sec-met}

%In this section, we detail the two nonparametric change point detection methods, both of which are based on the CUSUM Kolmogorov--Smirnov statistic (CUSUM KS).

In this section, we detail our two nonparametric change point detection procedures, both of which are based on the cumulative sum Kolmogorov--Smirnov statistic (CUSUM KS) defined  next.

\begin{definition}[CUSUM KS]
	\label{def:cusum}
	For any $1\leq s < t < e \leq T$, define the CUSUM Kolmogorov--Smirnov statistic
	\begin{equation}
	\label{eqn:ks}
		D_{s,e}^t = \sup_{z \in \mathbb{R}} \bigl|D^t_{s, e}(z)\bigr|,
	\end{equation}
	with
	\begin{equation}
	\label{eq:ks_z}
		D_{s,e}^t(z) = \sqrt{\frac{n_{s:t}\,  n_{(t+1):e} }{n_{s:e}} }\,\left\{ \widehat{F}_{s:t}(z)  -  \widehat{F}_{(t+1):e}(z)    \right\},
	\end{equation}
where	for all $s^{\prime} < e^{\prime}$, and  $z\in  \mathbb{R}$, we write
	\[
		\widehat{F}_{ s^{\prime} : e^{\prime}   }(z) = \frac{1}{n_{s^{\prime}:e^{\prime}}} \sum_{t=s^{\prime}}^{e^{\prime}}  \sum_{i=1}^{n_t} \mathbbm{1}_{\{Y_{t,i} \leq  z\}},  
	\]
	and
	\[
		n_{s^{\prime}  :e^{\prime} } = \sum_{t=s^{\prime}}^{e^{\prime}} n_t.
	\]

\end{definition}

\begin{algorithm}[h!]
	%[htbp]
	\begin{algorithmic}
		\INPUT $\{Y_{t, i}, \, t = s, \ldots, e, \, i = 1, \ldots, n_t \} \subset \mathbb{R}$, $\tau > 0$.
		\State {\bf Initial} $\mathrm{FLAG} \leftarrow 0$,
		\While{$e-s > 2$ and $\mathrm{FLAG} = 0$}
		\State $a \leftarrow \max_{s < t < e} D^t_{s,e}$
		\If{ $a \leq \tau$}
		\State $\mathrm{FLAG} \leftarrow 1$
		\Else
		\State $b \leftarrow \argmax_{s < t < e}  D^t_{s,e}$
		\State add $b$ to the collection of estimated change points
		\State $\mathrm{NBS} ((s,b-1),\tau)$
		\State $\mathrm{NBS} ((b,e),\tau)$
		\EndIf
		\EndWhile
		\OUTPUT The collection of estimated change points.
		\caption{Nonparametric Binary Segmentation. $\mathrm{NBS}((s,e), \tau)$ }
		\label{alg:BS}
	\end{algorithmic}
\end{algorithm}

\begin{algorithm}[h!]
	\begin{algorithmic}
		\INPUT $\{Y_{t, i}, \, t = s, \ldots, e, \, i = 1, \ldots, n_t \} \subset \mathbb{R}$, $\tau$.
		\For{$m = 1, \ldots, M$}   
		\State $(s_m, e_m) \leftarrow [s,e]\cap [\alpha_m,\beta_m]$	
		\If{$e_m - s_m \geq 2$}
		\State $a_m \leftarrow \max_{s_m < t < e_m} D^t_{s_m, e_m}$
		\State $b_{m} \leftarrow \argmax_{s_m < t < e_m} D^t_{s_m, e_m}$
		\Else 
		\State $a_m \leftarrow -1$	
		\EndIf
		\EndFor
		\State $m^* \leftarrow \argmax_{m = 1, \ldots, M} a_{m}$
		\If{$a_{m^*} > \tau$}
		\State add $b_{m^*}$ to the set of estimated change points
		\State NWBS$((s, b_{m*}),\{ (\alpha_m,\beta_m)\}_{m=1}^M, \tau)$
		\State NWBS $((b_{m*}+1,e),\{ (\alpha_m,\beta_m)\}_{m=1}^M,\tau)$
		
		\EndIf  
		\OUTPUT The set of estimated change points.
		\caption{Nonparametric Wild Binary Segmentation. NWBS$((s, e),$ $\{ (\alpha_m,\beta_m)\}_{m=1}^M, \tau$)}
		\label{algorithm:WBS}
	\end{algorithmic}
\end{algorithm}

Notice that in Definition  \ref{def:cusum}, $n_{s^{\prime}:e^{\prime}} $ is the total number of observations   collected  in the interval $[s^{\prime},e^{\prime}]$, and $\widehat{F}_{s^{\prime}:e^{\prime}}$  is  the empirical  CDF  estimated  using the data collected in that time. Moreover,	$D_{s,e}^t $ is the KS statistic for testing  the null  hypothesis that the data collected over the time course  $[s,e]$ have the same distribution,  against the alternative that $t$ is the only change point in $[s,e]$.

Based on the CUSUM KS statistic, we propose nonparametric versions of BS and WBS in Algorithms~\ref{alg:BS} and \ref{algorithm:WBS}, respectively.  These two algorithms share the same rationale behind  the BS and WBS procedures for  univariate mean change point detection problems (see, for instance, \citealt{fryzlewicz2014wild}).  Specifically, for the nonparametric  binary segmentation  procedure (NBS)  described in Algorithm \ref{alg:BS}, we iteratively search for the time point which maximizes the CUSUM KS statistic.  If the corresponding CUSUM KS value exceeds a pre-specified threshold, then the time point is added to the estimated change point collection. The process stops when all the  statistics  are below the threshold.  As for the nonparametric wild binary segmentation method (NWBS) in Algorithm (\ref{algorithm:WBS}), instead of starting off from the whole time course, we draw a collection of random intervals and conduct NBS within each interval.  The maximum is chosen to be the maximum CUSUM KS over all the intervals, and the comparisons with the pre-specified threshold are conducted thereafter.  

In both Algorithms~\ref{alg:BS} and \ref{algorithm:WBS}, we set the input $\tau$ as the pre-specified threshold.  This  tunning parameter  directly controls the number of estimated change points,  with larger  values of $\tau$  producing   smaller  number of estimated change points. Our theory  in the next sections will shed some lights on how to choose $\tau$. 

Furthermore, for   $s < t <e$,   the computational cost for  calculating  the  statistic $D_{s,e}^t$  is  $O(  (e-s)n_{s:e} \log n_{s:e}  )$. This  can be seen with a naive calculation based on the merge  sort algorithm, and the fact that  supremum  in  \eqref{eqn:ks}  only needs to be taken over  $z \in  \{  y_{u,i}  \,:\, \,\,\,  s  \leq  u \leq  e,  \,  i = 1,\ldots, n_u    \} $. Hence, for a  choice of $\tau$  that produces  $\tilde{K}$ estimated change points, Algorithm \ref{alg:BS}  has worst case running time  $O(\tilde{K} T n_{1:T} \log n_{1:T}    )$.
%algorith$$
%for a given choice of $\tau$,  the computational  complexity 
%The number of random intervals $M$ needs to be large enough and more discussion will be accompanying Theorem~\ref{thm-wbs}. 

%We conclude  this section by noticing that Algorithm \ref{alg:BS} has a typical computational complexity in the order $O(  n_{1:T} \log^2  n_{1:T}   )$. The latter becomes  $O(T \log^2 T)$  in the case where $n_t = 1$ for  all $t$. To put this in perspective,  the original Binary Segmentation algorithm has a typical complexity of $O(T \log T)$. Thus, our procedure has an extra logarithmic factor  that comes from the fact that we compute the empirical cdfs rather than averaging, and computing the cdfs requires sorting the data  within every subinterval.   

\section{Theory}
\label{sec-theory}

In this section, we first prove the consistency of the NBS, and then the  consistency and optimality of NWBS.  To back our optimality claim, we provide information-theoretic lower bounds in Section~\ref{sec-phase}.  Finally, we compare our results with existing methods, in particular with \cite{zou2014nonparametric}.  The proofs, as well as any auxiliary technical results, can be found in the Appendices.

\subsection{Nonparametric  binary segmentation and nonparametric wild binary segmentation}\label{sec-bswbs}
%\label{sec-bs}
%Assumption~\ref{as2} is used to obtain the consistency of the estimators generated by Algorithm~\ref{alg:BS}.  Despite the nonparametric essence of our problem, we reach the same assumptions as those in \cite{wang2017optimal} and \cite{fryzlewicz2014wild}, which are designed for high-dimensional covariance and univariate mean change point detection in sub-Gaussian data problems. 

We start by showing consistency for the NBS estimator as defined in Algorithm \ref{alg:BS}. Before arriving at our first result, we state an assumption involving the minimum spacing condition, the minimum  jump size, and the  time horizon $T$.

\begin{assumption}
	\label{as2}
	There exists a constant  $C_{\alpha} >0$  such that  $\kappa \delta \,\geq \, C_{\alpha} T^{\Theta}  $ for some $\Theta \in (9/10,1]$. 
\end{assumption}

Note that  Assumption \ref{as2}  can be thought as a nonparametric version of  Assumption 2  in \cite{fryzlewicz2014wild},  which was used  there to  prove the consistency of BS  in the univariate mean change point detection problem with  sub-Gaussian errors.

The constant $C_{\alpha}$ in Assumption~\ref{as2} is an absolute constant, in the sense that it is not a function of $T$.  Technically speaking, when it is used in Theorem~\ref{thm-bs} below, it is required to be large enough to create enough room for the other constants in Theorem~\ref{thm-bs}.  We do not claim the optimality of the constant $C_{\alpha}$, but it can be tracked in the proofs, which we provide in the Appendices for readability.

We are now  ready to state  our theoretical  guarantee for the NBS  estimator. This proves that we can consistently  estimate  the number  of change points and their locations. 

\begin{theorem}[NBS]\label{thm-bs}
	Suppose  that Assumptions \ref{as1}-\ref{as2}  hold  and
	\begin{equation}
	\label{eqn:tuning}
	2\sqrt{n_{\max}\epsilon}     \,+\, C \sqrt{\log n_{1:T}}   \,<\, \tau\,<\,\frac{ \kappa  \delta n_{\min}  }{8   \sqrt{(e-s)n_{\max}} },
	\end{equation}
	with
	\begin{equation}\label{eq:epsilo.nbs}
	\epsilon \,=\, C_2 \, \log( n_{1:T} )  \,\delta^{-3} \kappa^{-1} (e-s)^{7/2}\frac{n_{\max}^{5/2} }{ n_{\min}^3 }, 
	%\sqrt{\log n_{1:T}} \,\delta^{-2} \kappa^{-1} (e-s)^{5/2}\frac{n_{\max}^{9/2} }{ n_{\min}^5 },
	\end{equation}
	for some positive constant $C_2.$ Then, with probability approaching one  as $T \rightarrow \infty$, the output of  NBS($(0,T),\tau$)  is such that
	\[
	\hat{K} \,=\,K,    \,\,\,\,\,\,\,\text{and}  \,\,\,\,\,\,\,        \underset{k = 1, \ldots, \hat{K} }{\max} \,\vert \hat{\eta}_k - \eta_{k}  \vert \,\leq \, \epsilon.
	\]
	
\end{theorem}

From Theorem~\ref{thm-bs} we can see that
	\[
		\frac{\epsilon}{T} \lesssim \log( n_{1:T} )  \,\delta^{-3} \kappa^{-1} (e-s)^{7/2}\frac{1}{n^{1/2} T} \lesssim \frac{\log(n_{1:T})}{T^{1/8} n^{1/2}} \to 0,
	\]
	where the second inequality holds under Assumption~\ref{as2}.  This implies consistency of Algorithm~\ref{alg:BS}.  Similar calculations, based on Assumption~\ref{as2}, also show that the range of values of $\tau$  defined by \eqref{eqn:tuning} is not empty.

%\subsection{Wild binary segmentation}

Though NBS is computationally fast and provably consistent, the localization rate implied by \eqref{eq:epsilo.nbs} is not optimal, as we will show below. In order to achieve near minimax-optimality we will turn instead to  the more complicated NWBS estimator. By considering random intervals, 
the NWBS methodology will be consistent in settings where the NBS procedure might not and, furthermore, will achieve faster localization rate. %Despite the superior performance of NWBS, we consider NBS a good and fast method.

\begin{assumption}\label{assumption-3}
Assume that there exists a constant $C_{\mathrm{SNR}} > 0$ such that
	\[
		\sqrt{\log(n_{1:T})}  < C_{\mathrm{SNR}} \kappa \sqrt{\delta} \sqrt{n}.
	\]
\end{assumption}

The constant $C_{\mathrm{SNR}} > 0$ plays a similar role as $C_{\alpha}$ used in Assumption~\ref{as2} and inherits the same discussion.  Assumption~\ref{assumption-3} is essentially a requirement on the rate of the signal-to-noise ratio.

Recall that in the univariate mean change point detection case, it is shown that if 
	\[
		\kappa \sqrt{\delta}/\sigma < \sqrt{\log(T)},
	\]
	then no algorithm is guaranteed to produce consistent change point estimators \citep[e.g.][]{wang2018univariate}.  Note that in this paper, essentially, the data in use are the empirical distribution estimators, and intuitively their  fluctuations are in the order of $\sigma \asymp n^{-1/2}$. Hence,  our next theorem leads to   an  intuitive argument that our method is optimal.
	% As a consequence, it leads to an intuitive argument that our method is optimal.  

\begin{theorem}[NWBS]\label{thm-wbs}
	Assume the inputs of the NWBS algorithm are as follows:
	\begin{itemize}
	\item the sequence ${\{Y_{t, i}\}_{i = 1}^{n_t}}_{t = 1}^T$ satisfies Assumptions~\ref{as1} and \ref{assumption-3};
	\item the collection of intervals $\{(\alpha_s, \beta_s)\}_{s = 1}^S \subset \{1, \ldots, T\}$, with endpoints drawn independently and uniformly from $\{1, \ldots, T\}$, satisfy
	\begin{equation}
		\label{eqn:intervals_condtion}
		\max_{s = 1, \ldots, S}(\beta_s - \alpha_s) \leq C_S\delta,
	\end{equation}
	 almost surely, for an absolute constant $C_S > 1$; and 
	\item the tuning parameter $\tau$ satisfies
		\begin{equation}\label{eq-thm4-tau}
			c_{\tau, 1}\sqrt{\log(n_{1:T})} \leq \tau \leq c_{\tau, 2}  \kappa \delta^{1/2} \frac{n^{3/2}_{\min}}{n_{\max}},
		\end{equation}
		where $c_{\tau, 1}, c_{\tau, 2} > 0$ are constants.
	\end{itemize}

Let $\{\hat{\eta}_k\}_{k = 1}^{\widehat{K}}$ be the corresponding output of the NWBS algorithm.  Then
	\begin{align}
		& \mathbb{P}\left\{\widehat{K} = K \quad \,\mbox{and} \,\quad  \epsilon_k \leq    \epsilon:= C_{\epsilon}\kappa_k^{-2}\log(n_{1:T}) n^{9}_{\max}n^{-10}_{\min}, \,\,\, \forall k = 1, \ldots, K\right\} \nonumber\\
		\geq & 1 - \frac{24 \log(n_{1:T})}{T^3 n_{1:T}} - \frac{48 T}{n_{1:T} \log(n_{1:T}) \delta} - \exp\left\{\log\left(\frac{T}{\delta}\right) - \frac{S\delta^2}{16T^2}\right\}, \label{eq-thm4-result}
	\end{align}
	where $\epsilon_k  \,=\,\vert \hat{\eta}_k - \eta_k \vert$ for $k= 1,\ldots,K$, and $C_{\epsilon}$ is  a positive  constant.
\end{theorem}

%\begin{remark}
Few remarks are in order.
\begin{itemize}
\item Based on Assumption~\ref{assumption-3}, the range of tuning parameters $\tau$ defined in  \eqref{eq-thm4-tau} is not empty, and the upper bound of the localisation error rate satisfies $\max_{k = 1, \ldots, K}\epsilon_k/T \to 0$ as $T$ grows unbounded.   In addition, as long as we choose 
\[
S \gtrsim \log\left(\frac{T}{\delta}\right)\frac{T^2}{\delta^2},
\]
the probability in \eqref{eq-thm4-result} tends to 1 as $T \to \infty$, which shows that NWBS is consistent.

%Recall that in the univariate mean change point detection problem, it is shown in Lemma~2 in \cite{wang2018univariate} that the lower bound of the localisation rate is of the order $\sigma^2 \kappa^{-2}$.  Intuitively, the upper bound of $\epsilon$ in Theorem~\ref{thm-wbs} is optimal off by a polynomial logarithm factor, if we further assume $n_{\min} \asymp n_{\max} \asymp n$.
%\end{remark}

\item The condition  \eqref{eqn:intervals_condtion} is somewhat unsatisfactory, as it assumes some knowledge of the constant $C_s$. This may not be available in practice, even though in many cases an educated guess on the minimal spacing is not too unreasonable to assume. In our proofs we do need condition \eqref{eqn:intervals_condtion}. We remark that such condition does not appear among the assumptions of Theorem 3.2 in \cite{fryzlewicz2014wild} due to a flaw in the proof of that result. On the other hand, in their alternative analysis of WBS, \cite{wang2016high} do not impose \eqref{eqn:intervals_condtion} but formulate instead a different condition that, like in our case, requires knowledge of a lower bound for $\delta$.
Luckily, condition  \eqref{eqn:intervals_condtion} is needed  to guarantee minimax rate optimality of NWBS, but not its consistency. 
  For instance, assuming that 
\begin{equation}
\label{eqn:new_condition}
\kappa \sqrt{\delta n}(\delta/T) \gtrsim \log^{1/2+\xi}(n_{1:T}),
\end{equation}
a slightly more stringent setting than in Assumption~\ref{assumption-3}, we have  that, with probability  tending to one,
\[
\max_{k = 1, \ldots, K}\epsilon_k   =  \epsilon  \lesssim  \frac{\log (n_{1:T})}{\kappa^2}\frac{n^9_{\max}}{n^{10}_{\min}}\left(\frac{T}{\delta}\right)^2.
\]
Thus,  under (\ref{eqn:new_condition}), instead  of Assumption \ref{assumption-3}, we have that 
NWBS is still consistent in the sense that $\epsilon/T \to 0$ as $T \to \infty$. This conclusion can be obtained  by  following the proof of Theorem \ref{thm-wbs}, replacing  $C_S$  with $T/\delta$
which explains the extra $T/\delta$ factor.

%which is still consistent in the sense that $\epsilon/T \to 0$ as $T \to \infty$. This conclusion can be obtained  by noticing that  one can always take   $C_S  \,=\,   T/\delta$ in (\ref{eqn:interval}), which explains the extra $T/\delta$ factor  

%the constant $C_S$, in practice, it holds automatically that $C_S = O(1)$.  In theory, without the additional condition  defined in \eqref{eqn:intervals_condtion}, one needs a stronger condition than Assumption~\ref{assumption-3}.   In fact, without assuming \eqref{eqn:intervals_condtion},  and larger localisation error rate, both of which are sub-optimal, off by a factor $T/\delta$ -- the upper bound of $C_S$.  If we do not impose this constant $C_S$, then under the condition that 
%	\[
%		\kappa \sqrt{\delta n}(\delta/T) \gtrsim \log^{1/2+\xi}(n_{1:T}),
%	\] 
%	we have
%	\[
%		\epsilon \lesssim \frac{\log (n_{1:T})}{\kappa^2}\frac{n^9_{\max}}{n^{10}_{\min}}\left(\frac{T}{\delta}\right)^2,
%	\]
%	which is still consistent in the sense that $\epsilon/T \to 0$ as $T \to \infty$.

%	\subsection{Outline of proofs of Theorems  \ref{thm-bs}--\ref{thm-wbs}  }
	
\item An important aspect of the \Cref{thm-wbs}, which strictly improves upon the guarantees claimed in \cite{fryzlewicz2014wild},% and \cite{wang2016high}
	is that it yields localization rates that are {\it local,} in the sense that each change point $\eta_k$ is associated its own localization rate,  depending on $\kappa_k$, the magnitude of the corresponding distributional change.

\item To better  understand Theorem \ref{thm-wbs},   recall  that in the univariate  mean change point detection problem, Lemma~2 in \cite{wang2018univariate} showed that the lower bound of the localisation rate is of the order $\sigma^2 \kappa^{-2}$. Moreover, if  $n_{\min} \asymp n_{\max} \asymp n$, then 
 $\epsilon$ in Theorem \ref{thm-wbs} is of the order $ n^{-1} \kappa^{-2}   \log (n_{1:T}) $. Hence, intuitively, the upper  bound $\epsilon$ in  Theorem  \ref{thm-wbs} is optimal off by a  logarithm factor.

\end{itemize}

The detailed proofs of Theorems~\ref{thm-bs} and \ref{thm-wbs} can be found in Appendix \ref{sec-appa}.  Here we sketch both results' roadmap, which proceeds by induction due to the nature of the NBS and NWBS algorithms.  We first control the deviances between the  statistics   $\{D_{s,e}^t \,:\,   1\leq s <t \leq e \leq T  \}$ and  their population versions, by exploiting concentration inequalities.  The rest of the proofs are conducted within the so-called good events, the probabilities of which tend to 1 as $T$ grows unbounded, that such fluctuations remain within an appropriate range. Next, we show that the population CUSUM KS statistics (Definition \ref{def:pop_cusum} in Appendix \ref{sec-appa}) achieve their maxima at the true change points.  Consequently, within the good events, our algorithms will correctly detect or reject the existence of change points.  Finally, we show that the CUSUM KS statistics decrease fast enough around their maxima, which leads to the upper bounds on the localisation error rates.  %In addition, in Algorithm~\ref{algorithm:WBS} and Theorem~\ref{thm-wbs}, we refine the localisation errors from Theorem~\ref{thm-bs} by adopting an ANOVA type of argument.

\subsection{Phase transition and minimax optimality}\label{sec-phase}

In this subsection,  we  prove that the NWBS algorithm is optimal, in the sense of nearly achieving minimax localization rates across a range of models for which the localization task is possible.  Towards that end, recall that in Theorem~\ref{thm-wbs}, we have shown that Algorithm~\ref{algorithm:WBS} provides consistent change point estimators under the assumption that 
	\begin{equation}\label{eq-high-snr}
		\kappa \sqrt{\delta} \sqrt{n} \gtrsim \sqrt{\log(n_{1:T})}.
	\end{equation}
	In Lemma~\ref{lemma-low-snr} below, we will show that if
	\begin{equation}\label{eq-low-snr}
		\kappa \sqrt{\delta} \sqrt{n} \lesssim 1,
	\end{equation}
	then no algorithm is guaranteed to be consistent uniformly over all possible change point problems. More precisely, we will construct a change point setting in which, under any choice of the parameters obeying the scaling \eqref{eq-low-snr}, consistent change point localization is impossible.  In light of \eqref{eq-high-snr}, \eqref{eq-low-snr} and \Cref{thm-wbs}, we have found a phase transition over the space of model parameters that separate scaling for which the localization task is impossible, from combinations of model parameters for which there exists an algorithm -- namely NWBS -- that is consistent. The separation between the these two regions is sharp, saving for a $\sqrt{\log(n_{1:T})}$ term.

\begin{lemma}\label{lemma-low-snr}
Let ${\{Y_{t, i}\}_{i=1}^n}_{t = 1}^T$ be a time series satisfying Assumption~\ref{as1} with one and only one change point.  Let $P^T_{\kappa, n, \delta}$ denote the corresponding joint distribution.  For any $0 < \zeta < 1/\sqrt{2}$, denote
	\[
		\mathcal{P}_{\zeta}^T = \left\{P^T_{\kappa, n, \delta}: \, \delta = \min\left\{\left\lfloor \frac{\zeta^2}{\kappa^2 n} \right\rfloor, \, \left\lfloor \frac{T}{3} \right\rfloor \right\}\right\}.
	\]
	Let $\hat{\eta}$ and $\eta(P)$ be an estimator and the true change point, respectively.  It holds that
	\[
		\inf_{\hat{\eta}} \sup_{P \in \mathcal{P}_{\zeta}^T} \mathbb{E}_P\bigl(\bigl|\hat{\eta} - \eta(P)\bigr|\bigr) \geq \frac{1 - 2\zeta^2}{3} T,
	\]
	where the infimum is over all possible estimators of the change point location.
\end{lemma}

In our next result, we derive a minimax lower bound on the localization task, which applies to nearly all combinations of model parameters outside the impossibility regions found in \Cref{lemma-low-snr}.

\begin{lemma}\label{lemma-error-opt}
Let $\{Y_{t, i}\}_{i=1, t = 1}^{n, T}$ be a time series satisfying Assumption~\ref{as1} with one and only one change point.  Let $P^T_{\kappa, n, \delta}$ denote the corresponding joint distribution.  Consider the class of distributions
	\[
		\mathcal{Q} = \left\{P^T_{\kappa, n, \delta}: \, \delta < T/2,\, \kappa < 1/2, \, \kappa\sqrt{\delta}\sqrt{n} \geq \zeta_T \right\},
	\]
	for any sequence $\{ \zeta_T \}$ such that $\lim_{T \rightarrow \infty} \zeta_T = \infty $. Let $\hat{\eta}$ and $\eta(P)$ be an estimator and the true change point, respectively.  Then, for all $T$ large enough, it holds that 
	\[
		\inf_{\hat{\eta}} \sup_{P \in \mathcal{Q}} \mathbb{E}_P\bigl(\bigl|\hat{\eta} - \eta(P)\bigr|\bigr) \geq \max \left\{ 1, \frac{1}{2} \Big\lceil\frac{1}{n\kappa^2} \Big\rceil e^{-2} \right\},
	\]
	where the infimum is over all possible estimators of the change point locations.
	
\end{lemma}

The above lower bound matches, saving for a logarithm factor, the localization rate for  NWBS we have established in  \Cref{thm-wbs}, thus showing that NWBS is nearly minimax rate-optimal.

\subsection{Comparisons}\label{sec-comp}

The comparisons between Algorithms~\ref{alg:BS} and \ref{algorithm:WBS} follow  the same lines as those in other change point detection problems.  Both  algorithms can be conducted in polynomial time. Also, under suitable conditions, both algorithms provide consistent change point estimators.  On one hand, Algorithm~\ref{alg:BS} is computationally cheaper than Algorithm~\ref{algorithm:WBS}, and contains fewer tuning parameters; but on the other hand, Algorithm~\ref{algorithm:WBS} requires a weaker signal-to-noise ratio and achieves smaller localisation error rates.  In fact, as we have explained in Section \ref{sec-phase}, Algorithm~\ref{algorithm:WBS} is optimal in both senses, off by a logarithm factor.

We can also compare our rates with those in the univariate mean change point detection problem, which assumes sub-Gaussian data \citep[e.g.][]{wang2018univariate}.  On one hand, this comparison inherits the main arguments when comparing parametric and nonparametric modeling methods in general.  Especially with the general model assumption we impose on the underlying distribution functions, we enjoy  risk-free from model mis-specification.  On the other hand, seemingly surprisingly, we achieve the same rates of those in the univariate mean change point detection case, even though sub-Gaussianity is assumed thereof.  In fact, this is expected.  Note that we are using the empirical distribution function in our CUSUM KS statistic, which is essentially a weighted Bernoulli random variable at every $z \in \mathbb{R}$.  Due to the fact that Bernoulli random variables are sub-Gaussian, and that the empirical distribution functions are step functions with knots only at the sample points, we are indeed to expect the same rates.

Furthermore,  the H-SMUCE procedure from \cite{pein2017heterogeneous}   can also be compared  to NWBS. Assuming Gaussian errors, $\delta \gtrsim T $, and  $K = O(1)$, Theorems 3.7 and 3.9  from \cite{pein2017heterogeneous}  proved that H-SMUCE  can
 consistently estimate  the number of change points,  and that $\epsilon  \lesssim  r(T)$,  for  any $r(T)$ sequence such that $r(T)/\log (T)  \,\rightarrow \,   \infty$.  This is weaker than our upper bound  that guarantees  $\epsilon    \lesssim\log T$. In addition, NWBS can handle  changes in variance  when  the mean remains constant, a setting where it is unknown if H-SMUCE  is consistent. 
 
 Another interesting contrast can be made between NWBS and the multiscale quantile segmentation (MQS) method recently introduced by \cite{vanegas2019multiscale}. Both of these algorithms make no assumption about the distributional form of the CDFs $F_t$. However, MQS is designed to identify changes in known quantiles. This is not a requirement  for NWBS which can detect any type of changes without previous knowledge of their  nature. As for statistical  guarantees, translating to our notation,  provided that  $\delta \gtrsim \log  (T) $, MQS can consistently estimate  the number of change points and  have  $\epsilon  \lesssim \log  (T)$. This matches our theoretical guarantees in Theorem \ref{thm-wbs}.%
 
To conclude this subsection, we would like to provide a thorough comparison between the guarantees of the NWBS algorithms, described in Theorem~\ref{thm-wbs}, and the properties of the methodology of \cite{zou2014nonparametric} for nonparametric change point detection. %To the best of our knowledge, the work of \cite{zou2014nonparametric} is the  
\begin{itemize}
	\item The approach from \cite{zou2014nonparametric} is based  on a Binomial likelihood function integrated over the whole support with a properly chosen weight function. In contrast,  our  algorithms exploit the CUSUM-KS statistics and does not require specifying a weight function nor integration calculations.  
	\item The conditions assumed in \cite{zou2014nonparametric} are, in general,  more stringent  than ours.  For instance, they required the CDF functions to be continuous (see Assumption (A1) thereof), while our results hold with arbitrary distribution  function.  Furthermore,  the empirical CDF over the whole data is assume to converge, almost surely, converge to the true CDF, uniformly on the support (see Assumption (A3) thereof).  This condition holds automatically if the number of change points is fixed by the renown Glivenko-Canteli theorem.  We do not need this condition, and we allow for the number of change points to diverge as the number of total time points grows unbounded.  Finally, in \cite{zou2014nonparametric}, the size of the distributional change at the change points is measured in a sophisticated way involving an integration over an appropriately chosen Kullback--Leibler divergence (see Assumption (A4) thereof). In contrast, we take this quantity to be the KS distance between the corresponding CDFs, see (\ref{eq-as1-kappa}).  Arguably, our characterization is more interpretable.
	\item NWSB is shown to enjoy stronger properties  under weaker conditions.  In detail:
	\begin{itemize}
		\item Algorithm  \ref{algorithm:WBS}  can handle the case when $\delta \gtrsim \log(T)$ and $K \lesssim T/\log(T)$ with a guaranteed localization rate  $\epsilon/T \lesssim \log(T)/T$.  On the contrary, this case does not satisfy the conditions in the Theorem~1 in \cite{zou2014nonparametric}. 
		\item If we let the number of change points to be of order $O(1)$ in  Theorem~2  of \cite{zou2014nonparametric}, then, again translating into our notation, their method yields  $\epsilon \gtrsim \log^{2+c}(T)$, for $c > 0$; while we have $\epsilon \lesssim \log(T)$.
		\item Our results and conditions involve the magnitude of the jump sizes, allowing for  the minimum jump size to decay to $0$. In contrast, \cite{zou2014nonparametric}  constraints  the jump sizes to be bounded away from $0$, and the result localization rate does not explicitly involve the jump sizes.  
		\item There appears to be somewhat of a conflict between the Corollary~1 and Theorem~1 in \cite{zou2014nonparametric}.  To be specific, translating their notation to ours, if $\epsilon = O_p(1)$ as guaranteed by Corollary~1, then the number of change points has to be of order $o(1)$ to ensure that the conditions in Theorem~1 are met.  The latter implies that their results only hold when there is no change point. 
	\end{itemize}
%	\item \textcolor{red}{I am not sure I understand this point} Lastly
%	the price we pay for the generality in model assumption is that under the conditions assumed in \cite{zou2014nonparametric}, we expect that due to the difference between the Kullback--Leibler divergence and the KS distance, their methods may have more power than ours.  It is difficult to directly compare these two measures, but we will provide with more insights through numerical experiments in Section~\ref{sec-num}.
%	\textcolor{red}{ Lastly, under the conditions in  \cite{zou2014nonparametric}, it is possible that  the method in \cite{zou2014nonparametric} would have more power than NWBS. In that sense, we might  pay a price for our more general assumptions than those in \cite{zou2014nonparametric}. }
	
\end{itemize}

\subsection{Choice of tuning parameter}
\label{sec:tuning_parameter}

The key tuning parameter in Algorithms  \ref{alg:BS} and \ref{algorithm:WBS} is $\tau$, a parameter that in essence determines whether a candidate change point should be included in their output.  In this subsection, we focus on how to pick $\tau $   in Algorithm \ref{algorithm:WBS}, although the discussion below remains  valid for  NBS by replacing  Algorithm \ref{algorithm:WBS} with Algorithm \ref{alg:BS}.

Notice that, in NWBS, if we vary $\tau$ along the real line from $\infty$ to $0$, then we start collecting more and more change points. In particular, if all the other inputs - namely, $\{Y_{t,i}\}, \, \{(\alpha_m, \beta_m)\}$ - are kept fixed, then $\mathcal{B}(\tau_1) \subseteq \mathcal{B}(\tau_2)$, for  $\tau_1 \geq \tau_2$, where $\mathcal{B}(\tau)$ is the output of Algorithm~\ref{algorithm:WBS} with input $\tau$.

%Recall the proof of Theorem~\ref{thm-wbs}.  In Step~1, we use the upper and lower bounds of $\tau$ to determine if a new change point estimator is included and if further splitting is conducted.  In Step~2, we show that if there exists an undetected change point in a certain interval and the algorithm is conducted, then we will be able to locate this undetected change point with the error rate upper bounded by $\epsilon$.  It is worth noting that the localisation error rate is derived irrelevant of $\tau$.

%In the following, we state two algorithms.  Algorithm~\ref{algorithm:compare} is a generic algorithm used to compare two different collections of change point estimators, with $\mathcal{B}_1 \subset \mathcal{B}_2$.  Algorithm~\ref{alg:full} is the full change point detection procedure with tuning parameter selection.

%Algorithm~\ref{algorithm:compare} aims to add in extra change point estimators if there are two different collections.  If there exists false positive in $\mathcal{B}_1 \cap \mathcal{B}_2$, then Algorithm~\ref{algorithm:compare} is not able to delete these false positives.

Next we proceed to  construct  two algorithms  that can be used for model comparison.  The first of these, Algorithm~\ref{algorithm:compare}, is a generic procedure  that can be used  for merging  two collections of estimated   change points $\mathcal{B}_1 $  and $\mathcal{B}_2$. Algorithm \ref{algorithm:compare}  deletes  from $\mathcal{B}_1 \cup \mathcal{B}_2$ potential  false  positives  by  testing their validity one by one using test statistics  based on the CUSUM KS.  However,  Algorithm \ref{algorithm:compare}  does not scan for potential false positives in the set $\mathcal{B}_1  \cap \mathcal{B}_2$.

\begin{algorithm}[h!]
	\begin{algorithmic}
		\INPUT $\{Y_{t, i}, \, t = 1, \ldots, n, \, i = 1, \ldots, n_t \}\subset \mathbb{R}$, $\mathcal{B}_1 \neq \mathcal{B}_2$, $\lambda$
		\State $\mathcal{C} \leftarrow (\mathcal{B}_2 \setminus \mathcal{B}_1) \cup (\mathcal{B}_1 \setminus \mathcal{B}_2)$, $n_c \leftarrow |\mathcal{C}|$.	
		\State $\widehat{\mathcal{B}} \leftarrow \mathcal{B}_1 \cap \mathcal{B}_2$.
		\For{$i \in \{1, \ldots, n_c\}$}
		\State $\eta \leftarrow \eta_i \in \mathcal{C}$
		\If{$\eta \in \mathcal{B}_2 \setminus \mathcal{B}_1$}
		\State $k \leftarrow $ the integer satisfying $\eta \in (\hat{\eta}_k, \hat{\eta}_{k+1})$, where $\{\hat{\eta}_k, \hat{\eta}_{k+1}\}\subset \mathcal{B}_1$
		\Else
		\State $k \leftarrow $ the integer satisfying $\eta \in (\hat{\eta}_k, \hat{\eta}_{k+1})$, where $\{\hat{\eta}_k, \hat{\eta}_{k+1}\}\subset \mathcal{B}_2$			
		\EndIf		
		\State $\hat{z} \leftarrow \min \argmax_{z \in \{Y_{t, i}\}} \left|D^{\eta}_{\hat{\eta}_k, \hat{\eta}_{k+1}}(z)\right|$
		\If{
			\begin{align*}
			& \sum_{t = \hat{\eta}_k + 1}^{\eta} \sum_{i = 1}^{n_t}\left(\mathbbm{1}_{\{Y_{t, i} \leq \hat{z}\}} - \widehat{F}_{(\hat{\eta}_k + 1): \eta}(\hat{z})\right)^2 + \sum_{t = \eta + 1}^{\hat{\eta}_{k+1}} \sum_{i = 1}^{n_t}\left(\mathbbm{1}_{\{Y_{t, i} \leq \hat{z}\}} - \widehat{F}_{(\eta + 1): \hat{\eta}_{k+1}}(\hat{z})\right)^2 + \lambda \\
			& \hspace{6cm} <   \sum_{t = \hat{\eta}_k + 1}^{\hat{\eta}_{k+1}} \sum_{i = 1}^{n_t}\left(\mathbbm{1}_{\{Y_{t, i} \leq \hat{z}\}} - \widehat{F}_{(\hat{\eta}_k + 1): \hat{\eta}_{k+1}}(\hat{z})\right)^2
			\end{align*}
		}{ $\widehat{\mathcal{B}} \leftarrow \widehat{\mathcal{B}} \cup \{\eta\}$}
		\EndIf
		\EndFor	
		\OUTPUT $\widehat{\mathcal{B}}$	
		\caption{Update}
		\label{algorithm:compare}
	\end{algorithmic}
\end{algorithm}

In   order to have a practical  scheme for selecting $\tau$,  we propose  Algorithm \ref{alg:full}, which is a  full change point detection procedure with tuning parameter selection. To present Algorithm~\ref{alg:full}, we slightly  abuse the notation.  In particular, in order to emphasize that the NWBS procedure is conducted on the sample $\{W_{t, i}\}$, we include $\{W_{t, i}\}$ as a formal input to  NWBS.  In addition, since the CUSUM KS statistics are  based on $\{Y_{t,i}\}$, we now use the notation $D^{\eta}_{\hat{\eta}_k, \hat{\eta}_{k+1}}(z, \{Y_{t, i}\})$ instead of $D^{\eta}_{\hat{\eta}_k, \hat{\eta}_{k+1}}(z)$.  Finally, the superscript $Y$ in $\widehat{F}^Y_{s:e}(z)$ is included to indicate that the empirical distribution functions are constructed using the observations $\{Y_{t, i}\}$.

\begin{algorithm}[h!]
	\begin{algorithmic}
		\INPUT $\{Y_{t, i}, W_{t, i} \, t = 1, \ldots, n, \, i = 1, \ldots, n_t \} \subset \mathbb{R}$, $\tau_1 < \cdots < \tau_M$, $\{(\alpha_s, \beta_s)\}_{s = 1}^S \subset\{1, \ldots, T\}$, $\lambda > 0$.
		\For{$m \in \{1, \ldots, M\}$}
		\State $\mathcal{B}_m \leftarrow \mathrm{NWBS}((0, T), \{(\alpha_s, \beta_s)\}_{s = 1}^S, \tau_m, \{W_{t, i}\})$
		\EndFor
		\State $\mathcal{O} \leftarrow \mathcal{B}_M$
		\For{$m \in \{1, \ldots, M-1\}$}
		\If{$\mathcal{B}_{m+1} \neq \mathcal{O}$}
		\State $\eta \leftarrow \min  \mathcal{O}\setminus \mathcal{B}_{m+1}$
		\State $k \leftarrow $ the integer satisfying $\eta \in (\hat{\eta}_k, \hat{\eta}_{k+1})$, where $\{\hat{\eta}_k, \hat{\eta}_{k+1}\}\subset \mathcal{B}_{m+1}$
		\State $\hat{z} \leftarrow \min \argmax_{z \in \{Y_{t, i}\}} \left|D^{\eta}_{\hat{\eta}_k, \hat{\eta}_{k+1}}(z, \{Y_{t, i}\})\right|$
		\If{
			\begin{equation}
			\label{eqn:full_condition}
			%	\begin{align*}
			\begin{array}{ll}
			& \displaystyle  \sum_{t = \hat{\eta}_k + 1}^{\eta} \sum_{i = 1}^{n_t}\left(\mathbbm{1}_{\{Y_{t, i} \leq \hat{z}\}} - \widehat{F}^Y_{(\hat{\eta}_k + 1): \eta}(\hat{z})\right)^2 + \sum_{t = \eta + 1}^{\hat{\eta}_{k+1}} \sum_{i = 1}^{n_t}\left(\mathbbm{1}_{\{Y_{t, i} \leq \hat{z}\}} - \widehat{F}^Y_{(\eta + 1): \hat{\eta}_{k+1}}(\hat{z})\right)^2 + \lambda \\
			& \hspace{6cm} >   \displaystyle  \sum_{t = \hat{\eta}_k + 1}^{\hat{\eta}_{k+1}} \sum_{i = 1}^{n_t}\left(\mathbbm{1}_{\{Y_{t, i} \leq \hat{z}\}} - \widehat{F}^Y_{(\hat{\eta}_k + 1): \hat{\eta}_{k+1}}(\hat{z})\right)^2
			%\end{align*}
			\end{array}
			\end{equation}}
		\State    $\mathcal{O} \leftarrow \mathcal{B}_{m+1}$
		\Else {\quad Terminate the algorithm}				
		\EndIf
		\Else {\quad $\mathcal{O} \leftarrow \mathcal{B}_{m+1}$}	
		\EndIf
		\EndFor
		\OUTPUT $\mathcal{O}$	
		\caption{NWBS with tuning parameter selection}
		\label{alg:full}
	\end{algorithmic}
\end{algorithm}

Algorithm~\ref{alg:full} requires two subsamples: $\{Y_{t, i}\}$ and $\{W_{t, i}\}$.  In practice, this is not a problem, as one can  perform  sample splitting, or if $n_t \geq 2$, for all $t$, then one can partition the data at every time point.  In fact, there is no need to keep both subsamples having exactly the same number of sample points $n_t$ for  all $t$. Our theoretical guarantees in  
 Theorem~\ref{thm-full}   still hold    as  long as the the number of observations  have the same scaling  at each time point  in the two samples  $\{Y_{t, i}\}$ and $\{W_{t, i}\}$.
 
 As for the implementation of Algorithm~\ref{alg:full},   we arrange all the candidate sets in increasing order of their corresponding $\tau$ values.  This ensures a decreasing nesting of all these collections.  We begin with the set corresponding to the smallest $\tau$, and, in sequence, compare consecutive sets. However, unlike in Algorithm~\ref{algorithm:compare},  we pick a single element from the difference sets, and decide to move on or to terminate the procedure. Theorem~\ref{thm-full} provides suitable conditions that guarantee that this procedure results in a consistent estimator of the change points.
 
 % Algorithm~\ref{alg:full}, instead of going through all the elements in the difference sets, we pick any element in the difference sets, and decide to move on or terminate the procedure.  The consistency can be ensured under suitable conditions, which are stated in Theorem~\ref{thm-full}. 

\begin{theorem}\label{thm-full}
Suppose  that the following holds:
	\begin{itemize}
	\item the sequences ${\{Y_{t, i}\}_{i = 1}^{n_t}}_{t = 1}^T, {\{W_{t, i}\}_{i = 1}^{n_t}}_{t = 1}^T$ are independent and satisfy Assumptions~\ref{as1} and \ref{assumption-3};
	\item the collection of intervals $\{(\alpha_s, \beta_s)\}_{s = 1}^S \subset \{1, \ldots, T\}$, whose endpoints are drawn independently and uniformly from $\{1, \ldots, T\}$, satisfy $\max_{s = 1, \ldots, S}(\beta_s - \alpha_s) \leq C_S\delta$, almost surely, for an absolute constant $C_S > 1$; and 
	\item the tuning parameters $\{\tau_j\}_{j = 1}^J$ satisfy
		\begin{equation}\label{eq-tau-candidates}
			\tau_J > \ldots > c_{\tau, 2}  \kappa \delta^{1/2} \frac{n^{3/2}_{\min}}{n_{\max}} > \ldots > \tau_{j^*} > \ldots > c_{\tau, 1} \sqrt{\log(n_{1:T})} > \ldots > \tau_1,
		\end{equation}
		where $c_{\tau, 1}, c_{\tau, 2} > 0$ are constants.
	\end{itemize}

Let $\widehat{\mathcal{B}} = \{\hat{\eta}_1, \ldots, \hat{\eta}_{\hat{K}}\}$ be the output of Algorithm~\ref{alg:full} with inputs satisfying the conditions above.  If $\lambda = C\log(n_{1:T})$, with a large enough constant $C > 0$, then
	\begin{align}
		& \mathbb{P}\left\{\widehat{K} = K \quad \mbox{and} \quad \epsilon_k \leq C_{\epsilon}\kappa^{-2}_k\log (n_{1:T}) n^{9}_{\max}n^{-10}_{\min}, \, \forall k = 1, \ldots, K\right\} \nonumber\\
		\geq & 1 - \frac{48 \log(n_{1:T})}{T^3 n_{1:T}} - \frac{96 T}{n_{1:T} \log(n_{1:T}) \delta} - \exp\left\{\log\left(\frac{T}{\delta}\right) - \frac{S\delta^2}{16T^2}\right\}. \nonumber
	\end{align}
\end{theorem}

The proof  of Theorem \ref{thm-full}  can be found in Appendix  \ref{sec-full}. It implicitly  assumes that  the nested sets  $\{ \mathcal{B}_m \}$  in Algorithm \ref{alg:full}    satisfy 
$\vert \mathcal{B}_m   \setminus \mathcal{B}_{m+1} \vert \,=\,  1$, for  $m=  1,\ldots, M$. If this conditions is not met, then the conclusion of Theorem \ref{thm-full}  still holds provided  that we replace  (\ref{eqn:full_condition}) in Algorithm \ref{alg:full}, with the inequality
\[
         \lambda   \,>\,        \underset{s = 1,\ldots, S}{\max}  \,\, \underset{z \in \mathbb{R}  }{\sup}\,\left\vert D_{a_s,b_s}^{\eta}(z,  \{ Y_{t,i}\} )   \right\vert^2,  
\]
where  $(a_s,b_s)  \,=\,  (\hat{\eta}_k,\hat{\eta}_{k+1})\cap (\alpha_s, \beta_s)  $  for  $s= 1,\ldots, S$.

It is worth pointing out that similar results on tuning parameter selection can be found in Theorem~3.3 in \cite{fryzlewicz2014wild}.  Unfortunately,  the proof of Theorem~3.3 is wrong.  It can be seen from the proof  of \Cref{thm-full} presented in Appendix~\ref{sec-full} that there are a number of  technical issues that require a careful analysis.

\section{Numerical experiments}\label{sec-num}	

%\textcolor{red}{Would it be possible to male the code available?}

In this section we present the results of various simulation experiments  aimed  at assessing the performance of our methods under a wide range of scenarios and in relation to other competing methods. %The description of our experiments is given next.
%We proceed   evaluate the empirical performance of different change point  detection algorithms. 

We measure the performance of an estimator  $\hat{K}$ of the true number of change points by the absolute error $\vert  K - \hat{K}\vert$. In all our examples,  we report the average   absolute  error  over  100 Monte Carlo  simulations. 
Furthermore, denoting by $\mathcal{C}  = \{ \eta_1, \ldots, \eta_K  \}$  the set of true  change points, the performance of an estimator $\hat{\mathcal{C}}$ of $\mathcal{C}$ is measured by the one-sided Hausdorff distance
%we evaluate the localization performance of an algorithm  that  estimates $S$  as $\hat{S}$
%by the measure 
\[
d(\hat{\mathcal{C}}| \mathcal{C}) \,:=  \,  \underset{\eta \in \mathcal{C} }{\max }\,\underset{ x\in \hat{\mathcal{C}}   }{\min}\,\vert x - \eta\vert.
\]
%\textcolor{red}{I would also have reported $\max_x \min_\eta | \eta - x| $, which gives the largest distance from a true change point. This metric allows to overestimate the  number of change points, but, as long as each of them is close to a true one, we are OK. On the other hand, it can be insensitive to underestimation of $K$. The metric above shows the opposite patterns: it can be insensitive to overestimation  of $K$ and penalizes underestimation of $K$}
%By convention we set its value to be infinity when $\hat{\mathcal{C}}$ is the empty set. Note that  if $\hat{\mathcal{C}} \,=\, \{ 1,\ldots,T \}$ then  $d(\hat{\mathcal{C}}| \mathcal{C})  \,=\,0$. Hence, one has to be cautious when using the Hausdorff distance. In particular, it makes sense  for comparing $\hat{\mathcal{C}}_1$ and $\hat{\mathcal{C}}_2$ when these are sets of the same size.   In such case, we would prefer the  estimate  which is closer to $\mathcal{C}$ in  Hausdorff distance. In all of our simulations, for a method that produces an estimate $\hat{\mathcal{C}}$, we report  the median of $d(\hat{\mathcal{C}}| \mathcal{C}) $ over 50 Monte Carlo  simulations. 

By convention we set its value to be infinity when $\hat{\mathcal{C}}$ is the empty set. Note that  if $\hat{\mathcal{C}} \,=\, \{ 1,\ldots,T \}$ then  $d(\hat{\mathcal{C}}| \mathcal{C})  \,=\,0$. Thus,  $d(\hat{\mathcal{C}}| \mathcal{C}) $  can be insensitive to overestimation. To over come  this, we also calculate $d( \mathcal{C}|\hat{\mathcal{C}}) $.   In all of our simulations, for a method that produces an estimate $\hat{\mathcal{C}}$, we report  the median of both  $d(\hat{\mathcal{C}}| \mathcal{C}) $ and $d( \mathcal{C}|\hat{\mathcal{C}})$  over 100 Monte Carlo  simulations.% }

\begin{figure}[t!]
	\begin{center}
		\includegraphics[width=3.2in,height= 2.7in]{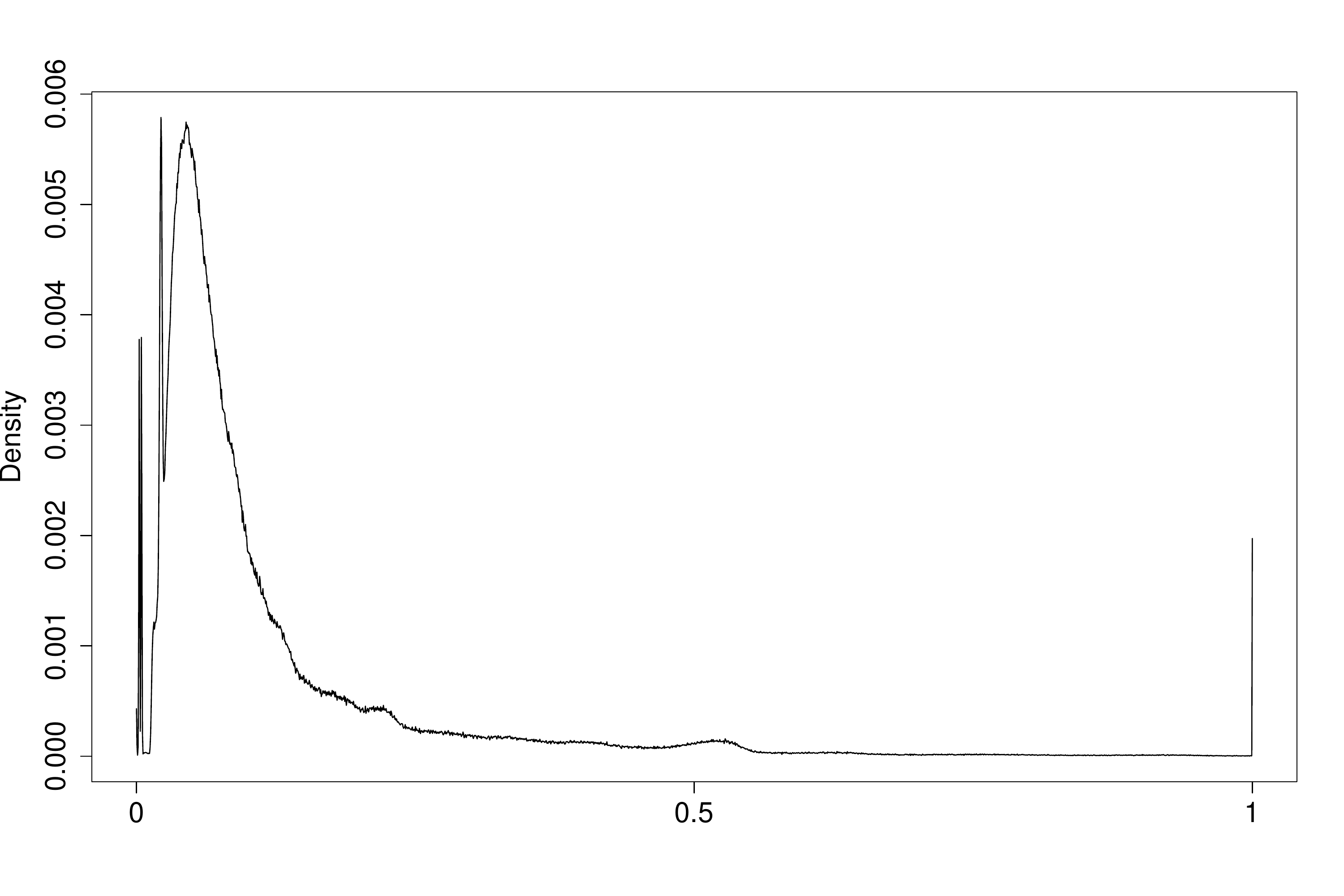} 
		\includegraphics[width=3.2in,height= 2.7in]{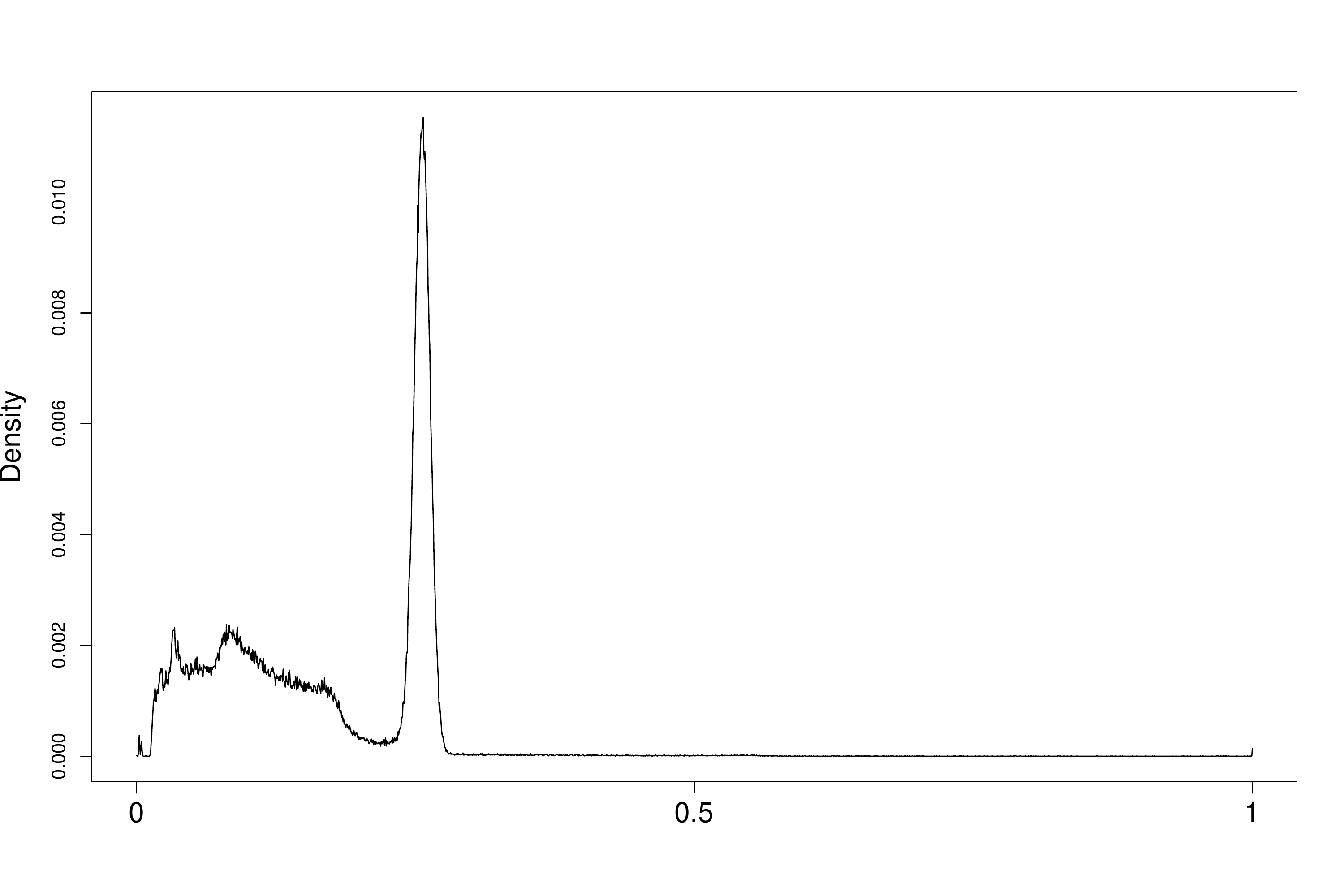}
		\caption{\label{fig1}  Densities  taken from \cite{padilla2018sequential}.  These are   used  in our generative models for the distributions between change points.  }
	\end{center}
\end{figure}

\begin{figure}[t!]
	\begin{center}
		\includegraphics[width=3.2in,height= 2.7in]{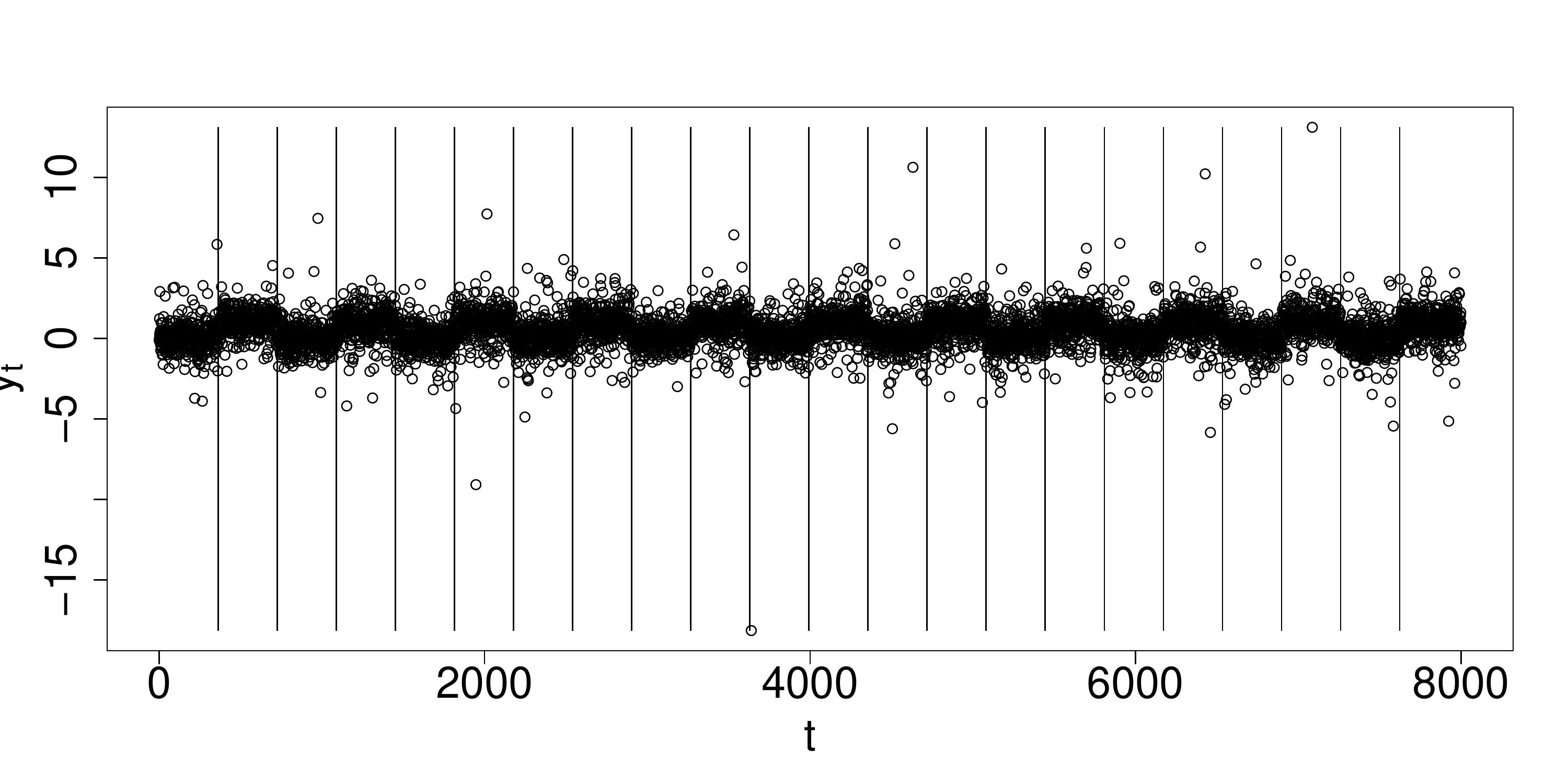} 
		\includegraphics[width=3.2in,height= 2.7in]{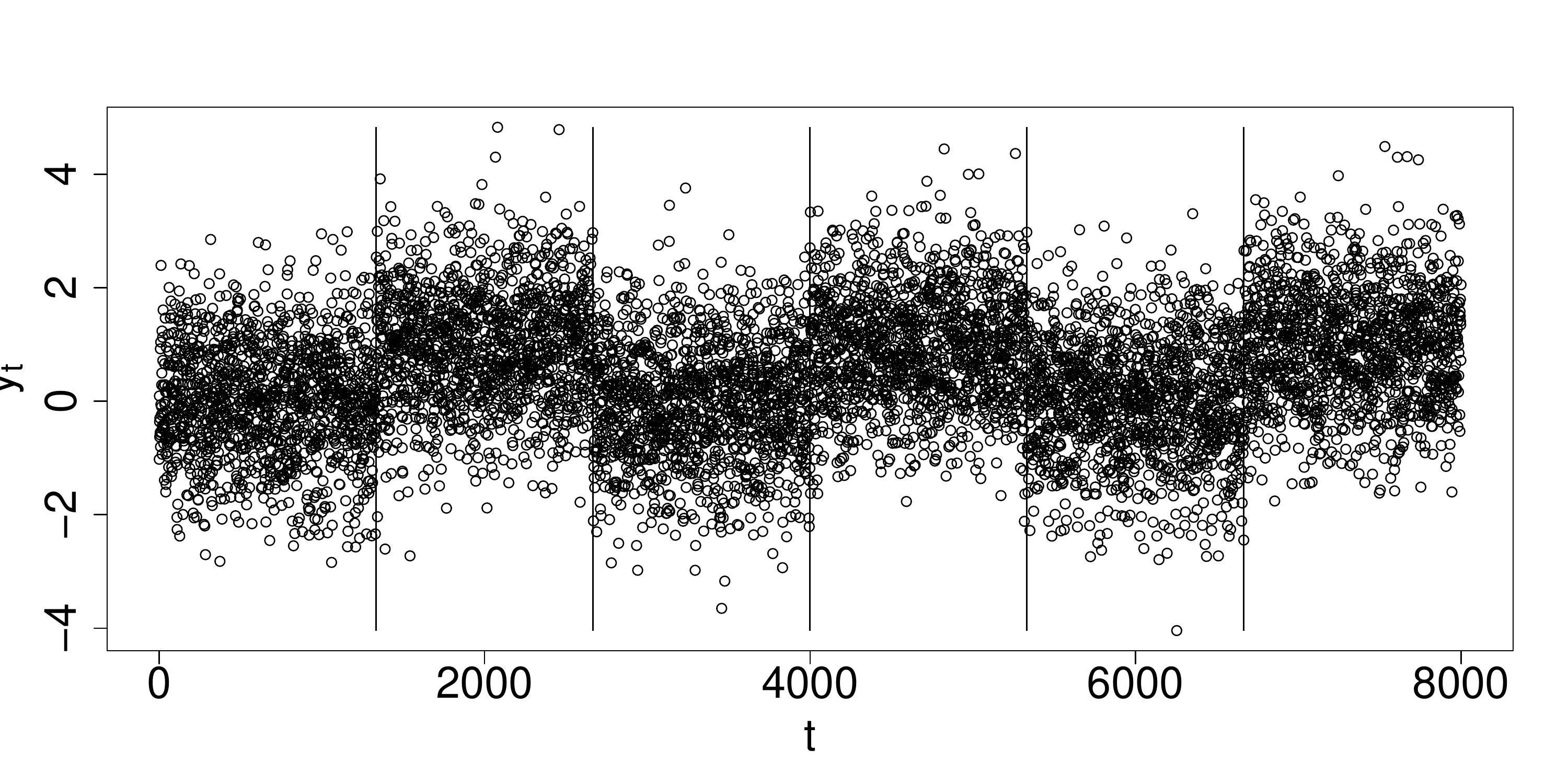}
		\includegraphics[width=3.2in,height= 2.7in]{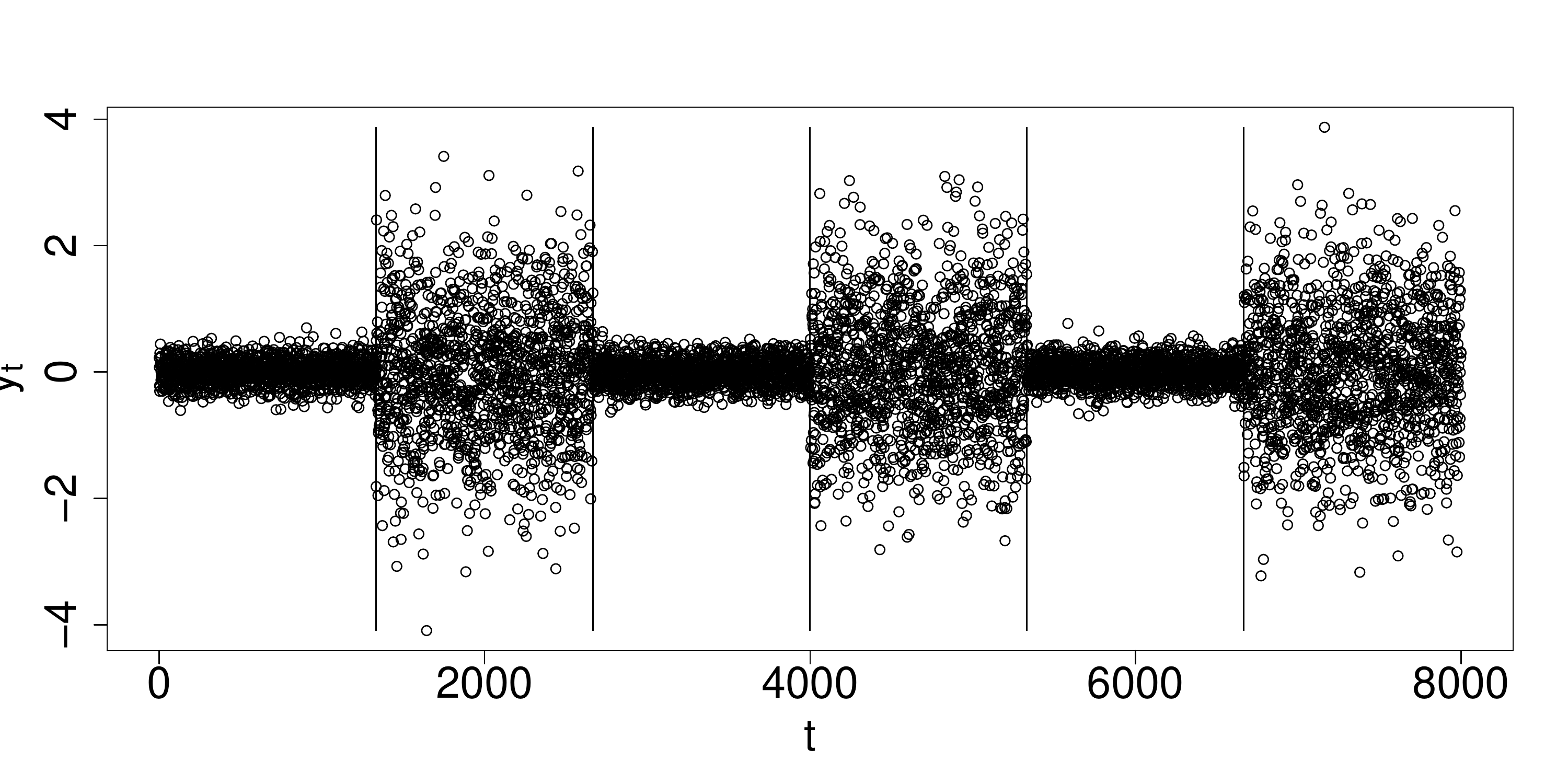}
				\includegraphics[width=3.2in,height= 2.7in]{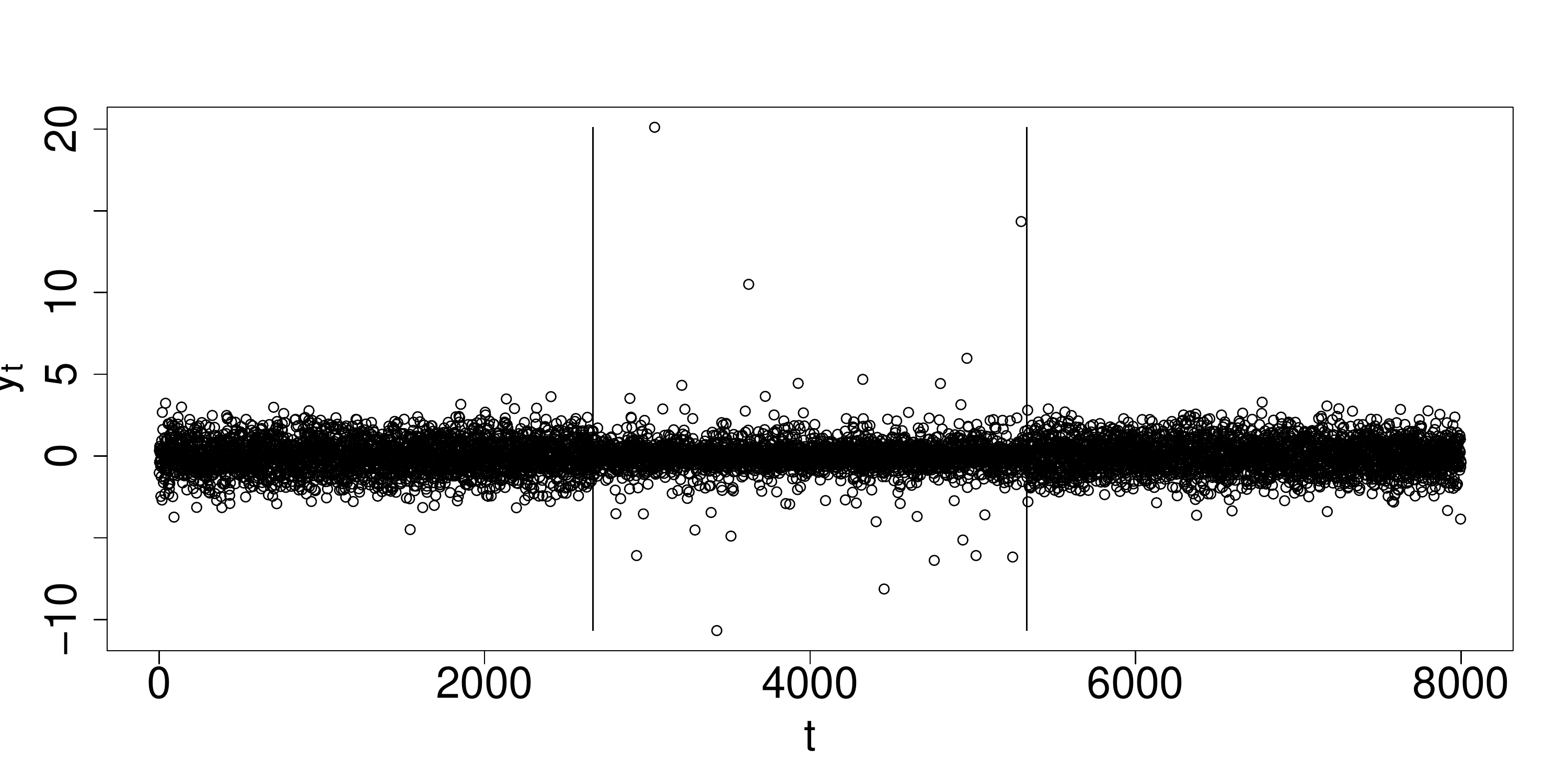}
		\caption{\label{fig2} Examples of  data generated from different scenarios with $T = 8000$. From left to right and from top to bottom  the plots correspond to data  generated from Scenarios 2, 3, 4, and 5, respectively.}
			%top to bottom and from left to right the plots correspond to Scenarios 2, 3, 4, and 5.}
			% Densities  taken from \cite{padilla2018sequential}.  These are   used  in our generative models as the distributions between change points.  }
	\end{center}
\end{figure}

%\begin{figure}[t!]
%	\begin{center}
%		\includegraphics[width=5.8in,height= 2.2in]{error_comp3p2.png} 
%		%\vspace{-0.5in}
%		\includegraphics[width=5.8in,height= 2.2in]{hausdofrr_comp3p2.png}
%		%	\vspace{-0.5in} 
		
%	%	\includegraphics[width=3.2in,height= 2.2in]{error_comp3p2.png}hausdofrr_comp3p2
%		\caption{\label{fig:cobalt} The top  panel shows the  average  error (over 50 Monte %Carlo simulations) for estimating $K$ with different methods under Scenario 3.  The second %panel from top to bottom shows the corresponding Hausdorff distance for the different %methods also under  Scenario 1. The third  panel  gives the  results for estimating $K$ %under Scenario 4. }
%	\end{center}
%\end{figure}

%Fo estimating  the number of change points we use the absolute error: $\vert K - \hat{K}\vert$  for evaluating the performance  of an estimate $\hat{K}$ of $K$.

\subsection{Case $n_t =1$}
\label{sec-nt=1}

We start by focusing on the case in which $n_t = 1$ for all $t = 1,\dots, T$. This allows us to compare our proposed  methods with state of the art algorithms  for change point detection. The methods  that  we benchmark against are:
\begin{itemize}
	\item WBS: The original wild binary segmentation algorithm from \cite{fryzlewicz2014wild}, implemented  in the R package  ``wbs".
	\item B$\&$P: The method from \cite{bai2003computation}, using the R package ``strucchange".
	\item S3IB, as introduced in \cite{rigaill2010pruned} and  with the R package ``Segmentor3IsBack".
	\item  PELT, the estimator from \cite{killick2012optimal}. We obtain this  via the R package ``changepoint".
	%\item cumSegm,  the algorithm described in \cite{muggeo2010efficient} and implemented in the R package ``cumSeg".
	\item SMUCE, the method  proposed by \cite{FrickEtal2014}. For this procedure we obtain the estimates using the R package ``stepR".
	\item  The H-SMUCE procedure from \cite{pein2017heterogeneous}. Estimates are calculated with the function ``stepFit" from  the R package ``stepR".
	
	\item NMCD: this is the nonparametric  maximum likelihood  approached  developed by \cite{zou2014nonparametric}. We  compute this estimator  using the  R package  ``changepoint.np". 
\end{itemize}

For all the competing methods,  we select the respective tuning parameters  by  using the  default choices in their respective R packages,  following similar guidelines to \cite{fryzlewicz2014wild}.

As for our approaches, we consider NBS (Algorithm \ref{alg:BS}) and  NWBS (Algorithm \ref{algorithm:WBS}) choosing the tuning parameter $\tau$ as described in Section \ref{sec:tuning_parameter}. Specifically,
we  use  Algorithm \ref{alg:full} with  $\tau    \,=\,   2(\log n_{1:T}) /3$, a choice  that  is guided by Theorem  \ref{thm-full} and that we find  reasonable in practice. Moreover, we construct  the samples $\{Y_{t,i}\}$  and $\{W_{t,i}\}$  by splitting the data into two time series of  roughly the same  size,  according to  odd and even values of $t$. As for $S$, the number of random intervals, we set its value to $120$.

\begin{table}[t!]
	\centering
	\caption{\label{tab1} Scenario 1.}
	\medskip
	\setlength{\tabcolsep}{7pt}
	\begin{small}
		\begin{tabular}{ rrrrrrrrrrr }
			\hline
			$T$ & Metric                   & NWBS           & NBS             &WBS    &PELT    &S3IB       & NMCD     & SMUCE    & B$\&$P & H-SMUCE   \\  
			\hline	
			%250 & $\vert K - \hat{K}\vert$ &\textbf{1.6}   &3.1               &3.2    &3.7     &5.2        &3.7        &2.2     &7.0 \\
			%250 & $d(\hat{S}|S)$           &\textbf{14.0}  &\textbf{14.0}     &49.0   &67.5    &140.0      &100.0 	   & 36.5    & $\infty$\\		
			1000 &$\vert K - \hat{K}\vert$  &\textbf{1.5}   &4.3              &11.0    &4.35     &6.3       &1.8        &4.35       &6.85       &1.6\\			
   1000 & $d(\hat{\mathcal{C}}|\mathcal{C})$&32.0            &53.0             &63.0   &64.5     &$\infty$  &\textbf{23.0}&43.5       &$\infty$    &75.5 \\
  1000 & $d(\mathcal{C} |\hat{\mathcal{C}})$&52.0           &28.5              &82.0   &90.0     &$-\infty$ &\textbf{36.0}&93.5        &$-\infty$  &56.5\\
			4000& $\vert K - \hat{K}\vert$  &\textbf{0.2 }  &0.5               &26.3   &21.4      &6.9      &2.2        &52.0       &4.9      &8.3\\                      
   4000&  $d(\hat{\mathcal{C}}|\mathcal{C})$&31.0           &55                &43.5   &46.5     &$\infty$ &\textbf{13.5} &448.0      &1386.0 &70.0\\
 4000 & $d(\mathcal{C} |\hat{\mathcal{C}})$  &\textbf{31.0 }&62.6              & 364.5 & 448     &$-\infty$&69.5          &32.7    &101.5       &319.0\\
  8000& $\vert K - \hat{K}\vert$            &\textbf{0.0 }  &0.9               &28.8   &41.2     &5.35   &3.25       &62.7        &3.7       &17.6\\                      
8000&  $d(\hat{\mathcal{C}}|\mathcal{C})$    &38.0          &47.0              &29.5   &55.0     &$\infty$ &\textbf{10.0}&70.0     &1831       &84.5\\
  8000 & $d(\mathcal{C} |\hat{\mathcal{C}})$&\textbf{38.0}  &89.0             &818.5  &955.0    &$-\infty$ &227.0     &958.0      &206.5    &910.0\\

			\hline  
		\end{tabular}
	\end{small}
\end{table}

\begin{table}[t!]
	\centering
	\caption{\label{tab2} Scenario 2.}
	\medskip
	\setlength{\tabcolsep}{7pt}
	\begin{small}
		\begin{tabular}{ rrrrrrrrrrr }
			\hline
			$T$ & Metric                              & NWBS           & NBS             &WBS    &PELT    &S3IB       & NMCD     & SMUCE    & B$\&$P  & H-SMUCE \\  
			\hline	
  			1000 &$\vert K - \hat{K}\vert$             &1.3             &1.7              &9.9    &2.7     &1.85          &1.7      &4.6            &7.05       &\textbf{0.45}\\			
  1000 & $d(\hat{\mathcal{C}}|\mathcal{C})$            &        11.0    &14.5             &8.5    &9.0    &15.5         &\textbf{7.0}&10.5            & 738.5    &23.0 \\
  1000 & $d(\mathcal{C} |\hat{\mathcal{C}})$           &\textbf{13.0}            &20.5             &54.5   &45.5   &29.5         &28.0       &53.5          &16.0     &22.0\\
  4000& $\vert K - \hat{K}\vert$                        &\textbf{0.0}  &1.0              &15.4   &10.05   &14.7     & 3.35      &14.8       &11.95  &\textbf{0.0} \\                      
  4000&  $d(\hat{\mathcal{C}}|\mathcal{C})$             &16.0          &22.0              &9.5    &12.0     &$\infty$   &\textbf{9.0}    &35.0     &1007   &16.0 \\
  4000 & $d(\mathcal{C} |\hat{\mathcal{C}})$            &\textbf{16.0} &34.0              &176    &163.0    &$-\infty$  &107.5    &164     &20.0     &\textbf{16.0} \\
  8000& $\vert K - \hat{K}\vert$                        &1.3   &8.4    &17.4              &18.45  &20.8     &4.75      & 28.5            &18.4       &\textbf{0.1}\\                      
  8000&  $d(\hat{\mathcal{C}}|\mathcal{C})$             &363.0         &1470.5             &15.5  &\textbf{9.5}&$\infty$  &13.5       &40.0  &2179       &20.0\\	
  8000 & $d(\mathcal{C} |\hat{\mathcal{C}})$            &\textbf{18.0} &28.0             &254.5   & 2179  &$-\infty$  &129.5      &257       &19.5       &20.0\\ 
 \hline		
		\end{tabular}
	\end{small}
\end{table}

Next we  explain the different generative models that are deployed in our simulations. For all scenarios, we consider three values for $T$: $1000$, $4000$ or $8000$. Moreover, we  consider the partition $\mathcal{P}$ of  $[1,T]$  induced  by the change points $\eta_1,\ldots,\eta_K$   which we always take to be  evenly spaced in $[1,T]$. Specifically, the elements of $\mathcal{P}$ are $A_1,\ldots,A_{K+1}$, with $A_j \,=\, [\eta_{j-1}+1,\eta_j]$. Using this partition,  we now describe the explicit  scenarios  considered in our simulations.
%\[
 %   A_j \,=\,   \begin{cases}
  %  [1,\eta_1] &   \text{if}   \,\,\,\,\,   j = 1\\
   % [\eta_{j-1}+1,\eta_j] &   \text{if}  \,\,\,\,\,      j \in \{2,\ldots,K \}\\
    %[\eta_{K}+1 ,T]    &  \text{if}   \,\,\,\,\,   j = T.
    %\end{cases}
%\]
%Using this  partition, we now describe the explicit  scenarios  considered in our simulations.

%we let  $\mathcal{A}  \,=\,  \{  A_1,\ldots,A_{K+1} \} $ 

\paragraph*{Scenario 1.}

We construct  examples where $K = 7$ for each instance of $T$. Then we define $F_t$ to have   probability density function   as in the left panel of Figure \ref{fig1} for $t \in A_j$ with $j$  odd, and as in the right panel of Figure \ref{fig1} for $t \in A_j$ with $j$  even.
%has as pdf  the density in the left panel of Figure \ref{fig1} for $t \in A_j$ with $j$  odd. On the other hand, $F_t$  has as 
%pdf  the density in the right panel of Figure \ref{fig1} for $t \in A_j$ with $j$  even.
%As for the distributions between change points,  we assume that   these have pdfs  as depicted in Figure \ref{fig1}. Specifically, $f_t$ is  the left panel of Figure \ref{fig1} for  $t \in [1,\eta_1] \cup [\eta_2+1, T]$, whereas  the pdf  is the one in the right panel of Figure \ref{fig1}   for $t \in [\eta_1 +1 , \eta_2]$.

\paragraph*{Scenario 2.}

We construct a piecewise constant signal for each value of $T$, and then proceed to add   noise. Thus, letting $K  \,=\, \floor{  \sqrt{T/( 2\log T)}  }$, we define  $\theta \in \mathbb{R}^T$ as
\begin{equation}
	\label{eqn:piecewise_signal}
	\theta_t \,=\,  \begin{cases}
	1     & \text{if } \,\,\,  t  \in A_j, \,\,\,\,\text{with}\,\, j\,\, \text{odd},\\
	0  &  \text{otherwise.} 
	\end{cases}
\end{equation}
Then the data are generated as
\begin{equation}
\label{eqn:m1}
y_t \,= \, \theta_t  + \frac{\epsilon_t}{\sqrt{3}},\,\,\,\,\,\,  t = 1,\ldots, T,
\end{equation} 
where the $\epsilon_t$'s  are i.i.d., and follow a  $t$-distribution with $3$ degrees of freedom. The scaling of the errors  by  $\sqrt{3}$, in (\ref{eqn:m1}), is incorporated to make the variance  of the noise equals to 1.

\begin{table}[t!]
	\centering
	\caption{\label{tab3} Scenario 3.}
	\medskip
	\setlength{\tabcolsep}{7pt}
	\begin{small}
		\begin{tabular}{ rrrrrrrrrrr }
			\hline
			$T$ & Metric                                 & NWBS           & NBS           &WBS           &PELT           &S3IB               & NMCD     & SMUCE    & B$\&$P & H-SMUCE  \\  
			\hline	
						1000 &$\vert K - \hat{K}\vert$  &0.8              &  1.8            &\textbf{0.0} &\textbf{0.0}  &0.1             &0.8          &\textbf{0.0}   &\textbf{0.0}       &0.2\\			
			1000 & $d(\hat{\mathcal{C}}|\mathcal{C})$   &  16.0            & 24.0            &9.0        & \textbf{8.5}  &\textbf{8.5}    & 9.5         &\textbf{8.5}  &\textbf{8.5} & 9.5\\
			1000 & $d(\mathcal{C} |\hat{\mathcal{C}})$   &19.0             &25.0             & 9.5       &  \textbf{8.5} &0.9              &10.5         &\textbf{8.5}  &\textbf{8.5} &9.5\\
			4000& $\vert K - \hat{K}\vert$               &0.1            &0.3                &\textbf{0.0}&\textbf{0.0}& \textbf{0.0} &1.8         &\textbf{0.0}   & \textbf{0.0} &\textbf{0.0}\\                      
			4000&  $d(\hat{\mathcal{C}}|\mathcal{C})$    &22.0           &28.0               &\textbf{8.0} &       9.0 &\textbf{8.0}  &11.5       & \textbf{8.0} &\textbf{8.0}&9.5\\
			4000 & $d(\mathcal{C} |\hat{\mathcal{C}})$  & 20.0           &47.0                &\textbf{8.0} &      9.0 & \textbf{8.0} & 136.5      &\textbf{8.0} &\textbf{8.0}&\textbf{8.0}\\
			8000& $\vert K - \hat{K}\vert$               &0.2            &0.2               &\textbf{0.0}&\textbf{0.0} &0.1             &1.7     &\textbf{0.0} &\textbf{0.0}&0.2\\                      
			8000&  $d(\hat{\mathcal{C}}|\mathcal{C})$    &11.5           &20.5              &\textbf{6.0} &\textbf{6.0} &\textbf{6.0}   &8.5       & \textbf{6.0} &\textbf{6.0}  &\textbf{6.0}\\
			8000 & $d(\mathcal{C} |\hat{\mathcal{C}})$  &11.5            & 23.5             &\textbf{6.0}&\textbf{6.0}&\textbf{6.0}     &353      &\textbf{6.0}    &\textbf{6.0}  &6.5\\			
			\hline  
		\end{tabular}
	\end{small}
\end{table}

\begin{table}[t!]
	\centering
	\caption{\label{tab4} Scenario 4.}
	\medskip
	\setlength{\tabcolsep}{7pt}
	\begin{small}
		\begin{tabular}{ rrrrrrrrrrrr }
			\hline
			$T$ & Metric                    & NWBS           & NBS             &WBS    &PELT    &S3IB     & NMCD     & SMUCE    & B$\&$P  & H-SMUCE\\  
			\hline	
			1000 &$\vert K - \hat{K}\vert$  & \textbf{0.9}  &3.9               &4.0  &   9.8     &4.9          &2.45         &  27.75   &5.0   &4.7\\			
1000 & $d(\hat{\mathcal{C}}|\mathcal{C})$   & 36.0          &$\infty$          &$\infty$&40.0        &$\infty$   & \textbf{4.0} & 24.5       &$\infty$& $\infty$ \\
1000 & $d(\mathcal{C} |\hat{\mathcal{C}})$  &  \textbf{32.0 }   &$-\infty$         &$-\infty$&153.5        &$-\infty$ &67.0          &157  &$-\infty$&$-\infty$\\
4000& $\vert K - \hat{K}\vert$              &\textbf{0.0}     &0.1               &3.8   &36.3    &5.0      &2.7       &  71.1           &5.0       & 4.5\\                      
4000&  $d(\hat{\mathcal{C}}|\mathcal{C})$   &19.0            & 30.0            &$\infty$ &106.5  &$\infty$  &\textbf{ 4.5} & 46.0        & $\infty$     &$\infty$ \\
4000 & $d(\mathcal{C} |\hat{\mathcal{C}})$  &  \textbf{19.0}     &30.0             & $-\infty$&644.5    & $-\infty$ &66.0    &  651.5      &  $-\infty$     &$-\infty$ \\
8000& $\vert K - \hat{K}\vert$              &\textbf{0.1}   &\textbf{0.1}       &3.5      & 60.3   & 5.0      &4.0              &109.3       &5.0       &4.5\\                      
8000&  $d(\hat{\mathcal{C}}|\mathcal{C})$    &23.0            &29.5             &6301.5   & 115.0   &$\infty$ &\textbf{2.5}    &47.5        & $\infty$      &$\infty$\\
8000 & $d(\mathcal{C} |\hat{\mathcal{C}})$  &\textbf{28.0}    &29.5             &238      &  1300.5  & $-\infty$& 566.5         &1316      & $-\infty$&$-\infty$\\
			\hline  
		\end{tabular}
	\end{small}
\end{table}
\begin{table}[t!]
	\centering
	\caption{\label{tab5} Scenario 5.}
	\medskip
	\setlength{\tabcolsep}{7pt}
	\begin{small}
		\begin{tabular}{ rrrrrrrrrrr }
			\hline
			$T$ & Metric                   & NWBS           & NBS             &WBS         &PELT    &S3IB       & NMCD     & SMUCE    & B$\&$P  & H-SMUCE\\  
			\hline	
			1000 &$\vert K - \hat{K}\vert$  &\textbf{0.4}   &  3.75           &3.6          & 7.2  & 5.0       &1.5       &  26.95    & 5.0      &4.8\\			
1000 & $d(\hat{\mathcal{C}}|\mathcal{C})$   &   27.0        & 665              &$\infty$     &70.0    & $\infty$ &\textbf{4.5}& 21.0        &$\infty$ &$\infty$ \\
1000 & $d(\mathcal{C} |\hat{\mathcal{C}})$ &  29.0          & \textbf{1.0}     &$-\infty$   & 147.0  &$-\infty$&32        &  159.0     &$-\infty$& $-\infty$\\
4000& $\vert K - \hat{K}\vert$             &\textbf{0.1}   & 0.32            &3.5      &38.2   & 5.0    & 3.35    &  72.05      &5.0  &4.5\\                      
4000&  $d(\hat{\mathcal{C}}|\mathcal{C})$  &  24.0             & 24.0          &$\infty$&82.0      &$\infty$ & \textbf{3.0}& 39.5   &$\infty$ &$\infty$\\
4000 & $d(\mathcal{C} |\hat{\mathcal{C}})$ &\textbf{ 25.0}             &  42.0          &$-\infty$&629.5   &$-\infty$ &       275& 640.5    &$-\infty$&$-\infty$\\
8000& $\vert K - \hat{K}\vert$              &\textbf{0.0}    & 0.3             &4.2    &63.9     &5.0   &   4.4     &114.0        &5.0 &4.6\\                      
8000&  $d(\hat{\mathcal{C}}|\mathcal{C})$    &    37.0       & 38.9           &$\infty$&107.5     &$\infty$& \textbf{2.5} &55.0        & $\infty$&$\infty$\\
8000 & $d(\mathcal{C} |\hat{\mathcal{C}})$  &   \textbf{37.0}        &45.0              &$-\infty$& 1309   &$-\infty$& 552   & 1310.5     &$-\infty$&$-\infty$\\
			\hline  
		\end{tabular}
	\end{small}
\end{table}

\paragraph*{Scenario 3.}

%Here the only difference with the previous  scenario is in the  distribution of the data in the intervals $[1,\eta_{1}]$,  $ [\eta_1 +1 , \eta_2]$, and $[\eta_2+1, T]$.  Now we have
Here,  we also construct a piecewise constant parameter. Letting  $K = 5$,  we define  $\theta \in \mathbb{R}^T$ as in (\ref{eqn:piecewise_signal}).
Then we consider data generated as $ y_t \,= \, \theta_t  + \epsilon_t$, with $\epsilon_t   \overset{\text{ind}}{\sim} N(0,1)$, for   $t = 1,\ldots, T$.
%\begin{equation*}
%	\label{eqn:m1}
%	 y_t \,= \, \theta_t  + \epsilon_t,
%\end{equation*} 
%where $\epsilon_t \sim N(0,1)$.
 
\paragraph*{Scenario 4.}

%This scenario  is constructed as Scenario 3,  with the difference  that we now choose $K  =  \floor{  T^{   \frac{1}{3}  }  } $.
In this class of examples the  distributional changes involve changes in the variance only. Specifically, $K = 5$ and  there is a parameter   $\theta \in \mathbb{R}^T$   satisfying 
\[
\theta_t \,=\,  \begin{cases}
0.2     & \text{if } \,\,\,  t  \in A_j, \,\,\,\,\text{with}  \, \,j \,\,\text{odd},\\
1  &  \text{otherwise.} 
\end{cases}
\]
Then
\[
	 y_t \,= \, \theta_t  \epsilon_t,\,\,\,\,\,\,  t = 1,\ldots, T,
\]
where $\epsilon_t  \overset{\text{ind}}{\sim} N(0,1)$.

\begin{table}[t!]
	\centering
	\caption{\label{tab6} }
	\medskip
	\setlength{\tabcolsep}{10.5pt}
	\begin{small}
		\begin{tabular}{rrrrrrrr}%{ rrrrrrrrrr }
			%	\hline
			&	                 &             &               &     Scenario 1             &                                &                                & \\
			\hline
			Method    & Metric                 &  $n_t = 5$ &  $n_t =  15$ &  $n_t =  30$    & $n_t \sim \text{Pois}(5)$   & $n_t \sim \text{Pois}(15)$  & $n_t \sim\text{Pois}(30)$  \\                      
			\hline	
			WNBS     &$\vert K - \hat{K}\vert$&\textbf{0.1} &\textbf{0.2}  &\textbf{0.0}   &0.6                           &\textbf{0.3}                &\textbf{0.0}\\ 
			NBS     &$\vert K - \hat{K}\vert$&1.45         &0.3           &0.2            &\textbf{0.4}                  &\textbf{0.3}                & 0.3\\ 
			WNBS     &$d(\hat{\mathcal{C}}|\mathcal{C})$&\textbf{6.5} &\textbf{2.5}  &\textbf{2.0}   &\textbf{6.0}                  &\textbf{2.0}                &\textbf{1.5} \\ 
			NBS     &$d(\hat{\mathcal{C}}|\mathcal{C})$&9.0          &3.0           &\textbf{2.0}   &7.5                           &  4.0                       & 3.5\\ 
			
			%	2000& $\vert K - \hat{K}\vert$ &\textbf{1.1}   &2.3	              &1.3	  &1.4     &1.4        &3.2        &     &2.0\\                      
			%	2000&  $d(\hat{S}|S)$          &35.5           &38.0               &235.5  &284.0   &310.       &\textbf{18.5}&     &$\infty$\\
			\hline  
			&	               &             &             &               &                                &                                & \\
			&       &             &                &               &                                &                                & \\
			&	        &                    &               &     Scenario 2             &                                &                                & \\
			\hline
			Method    & Metric                 &  $n_t = 5$   &    $n_t =  15$   &  $n_t =  30$  & $n_t \sim \text{Pois}(5)$   & $n_t \sim \text{Pois}(15)$  & $n_t \sim\text{Pois}(30)$  \\                      
			\hline	
			WNBS     &$\vert K - \hat{K}\vert$&\textbf{0.1} &\textbf{0.0} &\textbf{0.0}    &\textbf{0.4}                  &\textbf{0.0}               &\textbf{0.0}\\ 
			NBS     &$\vert K - \hat{K}\vert$&0.9          &0.2          &\textbf{0.0}    &\textbf{0.4}                  &\textbf{0.0}                & 0.1\\ 
			WNBS     &$d(\hat{\mathcal{C}}|\mathcal{C})$&\textbf{3.0} &\textbf{1.0} &\textbf{0.0}    &\textbf{3.0}                  &\textbf{1.0}                 & \textbf{0.0}\\ 
			NBS     &$d(\hat{\mathcal{C}}|\mathcal{C})$&6.0          &2.0          &1.0             &5.5                            &2.0                       & 1.0\\ 
			\hline  
			&	        &       &             &             &               &                                &                                 \\
			&	        &                    &                &               &                                &                                & \\
			&	        &                    &               &     Scenario 3             &                                &                                & \\
			\hline
			Method    & Metric       &  $n_t = 5$   &    $n_t =  15$   &  $n_t =  30$  & $n_t \sim \text{Pois}(5)$   & $n_t \sim \text{Pois}(15)$  & $n_t \sim\text{Pois}(30)$  \\                      
			\hline
			WNBS     &$\vert K - \hat{K}\vert$&\textbf{0.3} &0.3          &\textbf{0.0}   & \textbf{0.4}                  & \textbf{0.0}              &\textbf{0.0}\\ 
			NBS     &$\vert K - \hat{K}\vert$&0.9          &\textbf{0.2} &\textbf{0.0}    &   0.8                         &    0.1                        &\textbf{0.0} \\ 
 WNBS     &$d(\hat{\mathcal{C}}|\mathcal{C})$&\textbf{6.5} & \textbf{1.0}&\textbf{0.5}    &  \textbf{5.0}                 &   \textbf{2.0}            & \textbf{1.0}\\ 
  NBS     &$d(\hat{\mathcal{C}}|\mathcal{C})$&7.0          &2.0          &1.0             &     8.0                         &  \textbf{2.0}                            & \textbf{1.0}\\ 
		
			\hline 
			&	        &       &             &             &               &                                &                                 \\
			&	        &                    &                &               &                                &                                & \\
			&	        &                    &               &     Scenario 4             &                                &                                & \\
			\hline
			Method    & Metric                &  $n_t = 5$   &    $n_t =  15$   &  $n_t =  30$  & $n_t \sim \text{Pois}(5)$   & $n_t \sim \text{Pois}(15)$  & $n_t \sim\text{Pois}(30)$  \\                      
			\hline	
			WNBS     &$\vert K - \hat{K}\vert$&\textbf{0.2}  &\textbf{0.0} & \textbf{0.0}    &     \textbf{0.0}             &\textbf{0.0}                &\textbf{0.0}  \\ 
			NBS     &$\vert K - \hat{K}\vert$& 0.4           &     0.1     &  \textbf{0.0}    &         0.3                  & 0.2                        & 0.1\\ 
 WNBS     &$d(\hat{\mathcal{C}}|\mathcal{C})$&\textbf{6.0}  &\textbf{2.0}   &\textbf{0.0}      &          5.0                & \textbf{1.0}                  &\textbf{1.0}   \\  
  NBS     &$d(\hat{\mathcal{C}}|\mathcal{C})$& 6.5          & 3.0           & \textbf{1.0}     &   \textbf{4.5}               &    2.5                     &\textbf{1.0}   \\ 
			\hline 
			&	        &       &             &             &               &                                &                                 \\
			
			&	        &       &             &             &               &                                &                                 \\
			%&	        &                    &                &               &                                                                \\
			&	        &                    &               &     Scenario 5             &                                &                                & \\
			\hline
			Method    & Metric       &  $n_t = 5$   &    $n_t =  15$   &  $n_t =  30$  & $n_t \sim \text{Pois}(5)$   & $n_t \sim \text{Pois}(15)$  & $n_t \sim\text{Pois}(30)$  \\                      
			\hline	
			WNBS     &$\vert K - \hat{K}\vert$&\textbf{0.1} &\textbf{0.0}  &\textbf{0.0}   &\textbf{0.0}                  &\textbf{0.0}                    &\textbf{0.0} \\ 
			NBS     &$\vert K - \hat{K}\vert$&1.4          &0.2           &0.2            &0.3                           &0.3                             &0.3 \\ 
			WNBS     &$d(\hat{\mathcal{C}}|\mathcal{C})$&9.5          &\textbf{3.0}  &\textbf{2.0}   &\textbf{6.0}                  &\textbf{5.5}                    & \textbf{6.0}\\  
			NBS     &$d(\hat{\mathcal{C}}|\mathcal{C})$&$\infty$     &6.0           &4.0            &11.5                          &9.5                             &7.0 \\ 
			\hline 
			
		\end{tabular}
	\end{small}
\end{table}

\paragraph*{Scenario 5.}

We consider situations  where the  mean and variance remain constant but the distribution of the data changes. Setting $K = 2$,  we assume that  the data are generated as
\[
  y_t   \,=\,    \begin{cases}
  \epsilon_t   & \text{if } \,\,\,  t  \in A_j, \,\,\,\,\text{with} \,\, j \,\, \text{ odd,}\\
  \frac{\xi_t}{\sqrt{  \frac{2.5}{2.5 - 2}    }}  &  \text{otherwise.} 
  \end{cases}
\]
where   $\epsilon_t  \overset{ \text{ind}  }{\sim} N(0,1)$, and the $\xi_t$'s are  i.i.d. and follow a  $t$ distribution with $2.5$  degrees of freedom.

A visualization of the different scenarios is given in  Figures \ref{fig1} and \ref{fig2}. There, we can see that  indeed  the five scenarios  capture a broad range of models  that can  allow us to  asses the quality of different methods. 

Based on the results in Tables  \ref{tab1}--\ref{tab5},  we can see that, generally,  the  best  performance is attained by  the NWBS method. This is  seen in Scenario 1, where,  as the sample size  grows,  both NWBS and NBS  provide the best estimates  of $K$. This is not surprising  as  Scenario 1 presents a situation where the  distributions are not members of usual parametric families. 

In Scenario 2, H-SMUCE  attains the best performance  with NWBS  as a close competitor.   This scenario poses a challenge for most methods  due  to the heavy tails nature of the $t$-distribution. For instance, WBS is known to enjoy consistency properties for detecting mean changes but  when the errors are sub-Gaussian.
%With the same spirit, the superior performance of NWBS in Scenario  2 is perhaps due to the fact the t-distribution has heavy tails. This poses a challenge for most methods, for instance WBS is known to enjoy consistency properties for detecting mean changes  when the errors are sub-Gaussian. 

 Interestingly, in Scenario 3, we can see that   NBS and NWBS are  not the best methods. Their performance  is overshadowed  by methods like  WBS, S3IB, SMUCE, and   B$\&$P. However,  NWBS is still competitive in estimating $K$, although the aforementioned  methods  give better  change point estimates. This should not come as a surprise, since all these methods are designed to work well in this particular change point model.
 %&WBS         &PELT       &S3IB          & NMCD     & SMUCE    & B$\&$P 
 %The former  of these  has   statistical guarantees for models  like Scenario  3  where the errors are Gaussian.

 In Scenario 4  we  see  a very  clear  advantage for using one of the nonparametric  approaches NWBS, NBS, and NMCD. As it  can be seen in Table \ref{tab4}, methods  like WBS, PELT,  B$\&$P, SMUCE, and H-SMUCE  suffer greatly in this scenario. This is particularly  more interesting for H-SMUCE, as  such method  is known to detect  changes in variance but  when there are  also changes in mean. However,  Scenario  4 only  includes  changes in variance with the mean remaining constant. %Furthermore, NWBS, NBS, and NMCD generally provide better  results  under the performance metrics considered.
 % NBS suffers greatly due to the high number of change points, whereas  NWBS performs the best closely followed by NMCD. A similar behavior is observed in Scenario 4.

Finally, Scenario 5 seems to be the most challenging  for all methods. In fact, NWBS seems to be the only method capable of estimating  correctly the number of change points.

\subsection{ Case  $n_t > 1$}

We conclude the experiments  section  with simulations in the case where the number of  data points collected at any time instance can be more than one. To construct our examples,  once again we consider the five  scenarios from  Section \ref{sec-nt=1}. The only difference now is that we fix $T =  1000$,  and allow $n_t$  to vary.

 For values of $n_t$, we consider two  generic situations. The first consists of having  $n_t =  c$,  for all $t \in [1,T]$, and for  a constant  $c$. The different values of $c$ that we consider are $5$, $15$, and $30$. The second situation that we construct is  $n_t \sim \text{Pos}(c)$, where  $c$ is fixed to one of the values   $5$, $15$, or $30$.

With regards to benchmarks, we  only compare the NBS and NWBS methods. The reason for this is  that for all other methods, handling the case $n_t>1$ would require to make  significant adjustments to their algorithms/code. This  could lead to unfair comparisons. 

As for performance  metrics,  we use the same criteria from Section \ref{sec-nt=1}, omitting the  measure $d(\mathcal{C}|\hat{\mathcal{C}})$ as it does not provide additional information in the comparisons  between  NWBS and NBS.

The choice of tuning parameter for both WNBS and NBS  is  done in  the same way as in Section \ref{sec-nt=1}. Thus, we use Algorithm \ref{alg:full}  with $\tau  =  2(\log n_{1:T})/3$.	
	
Based on the results in Table \ref{tab6},  we see that, generally, NWBS  outperforms NBS. This tends to be more  evident as the number of samples grows, and  it goes in agreement with our results in  Section \ref{sec-nt=1}, and with our theory in Section \ref{sec-theory}.

%\textcolor{red}{I killed the discussion}

%\section{Discussion} \label{sec-discussion}
%
%In this paper, we studied nonparametric change point detection problem, proposed two computationally-efficient and theoretically-nearly-optimal methods, %and conducted extensive simulation results to support different aspects of our theoretical claims.  
%
%The key feature of our results is the generality.  Compared not only to parametric cousins, but also to nonparametric ones, we assume minimal assumptions %on underlying models.  
%
%We would also like to point out that the problem we studied in this paper can also be dealt with by Potts functional \citep[e.g.][]{FrickEtal2014, %wang2018univariate}.  To be specific, the change point estimators can be defined as
%	\begin{equation}\label{eq-dis}
%	\{\hat{\eta}_k\}_{k=1}^{\widehat{K}} = \argmin_{\{\hat{\eta}_{k}\}_{k=1}^K \subset \{1, \ldots, T-1\}} \left\{\sum_{k = 1}^{\hat{K}} \sup_{z \in \mathbb%{R}}\sum_{t = \hat{\eta}_k+1}^{\hat{\eta}_{k+1}} \sum_{i = 1}^{n_t} \left(\mathbbm{1}_{\{Y_{t, i} \leq z\}} - \widehat{F}_{(\hat{\eta}_k + 1): \hat{\eta%}_{k+1}}(z) \right)^2 + \lambda \widehat{K}\right\},
%	\end{equation}
%	where $\lambda > 0$ is a suitable tuning parameter. The optimization problem \eqref{eq-dis} can be solved by dynamic programming.  

\appendix

%	\item B$\&$P: The method from \cite{bai2003computation}, using the R package ``strucchange".
%\item S3IB, as introduced in \cite{rigaill2010pruned} and  with the R package ``Segmentor3IsBack".
%\item  PELT, the estimator from \cite{killick2012optimal}. We obtain this  via the R package ``changepoint".
%\item cumSegm,  the algorithm described in \cite{muggeo2010efficient} and implemented in the R package ``cumSeg".
%\item SMUCE, t toasses

\section{Main proofs of Theorems \ref{thm-bs} and  \ref{thm-wbs} }\label{sec-appa}

\subsection{Notation}
%\ref{def:cusum}
\begin{definition}
	\label{def:pop_cusum}
	Denote the population version of the CUSUM Kolmogorov--Smirnov statistic as
	\[
	\Delta_{s,e}^t = \underset{ z \in \mathbb{R}  }{\sup }\bigl|  \Delta_{s,e}^t(z)  \bigr|,  
	\]
	where
	\begin{equation}\label{eq-pop-KS-z}
	\Delta_{s,e}^t(z) = \sqrt{\frac{  n_{s:t} \, n_{(t+1):e} }{n_{s:e}} } \left\{ F_{s:t}(z) -   F_{ (t+1):e }(z)  \right\},  
	\end{equation}
	and
	\[
	F_{s:e}(z) \,=\, \frac{1}{n_{s:e}}\sum_{t=s}^e   \, n_t  F_t(z).
	\]
\end{definition}

\subsection{Proof of Theorem \ref{thm-bs}}

\begin{proof}%[Proof of Theorem~\ref{thm-bs}]
	We  start by noticing that Assumption \ref{as1}--\ref{as2} imply that	
	\begin{equation}
		\label{eqn:as2_lower-bound}
		  \frac{ n_{\max}^{9/2}}{ n_{\min}^{5}  } \sqrt{\log n_{1:T}}\,\leq \, c_{\alpha}T^{5\Theta - 9/2}, 
	\end{equation}
	for a sufficiently small $c_{\alpha}$. To see this, notice that for some positive constants $c_1$ and $c_2$,  we have  $c_1 n \leq  n_{\min} \leq n_{\max}  \leq c_2 n$,  and so 
	\[
	 \frac{ n_{\max}^{9/2}}{ n_{\min}^{5}  } \sqrt{\log n_{1:T}}   \,\leq \,  \frac{ c_2^{9/2}  }{  c_1^5  \sqrt{n}  } \sqrt{ \log(c_2 T)  + \log n  } \,\leq \, \frac{c_2^{9/2}  }{ c_1^5 } \sqrt{  \log(c_2 T)  + 1  }.
	\]
	
\noindent Next,	let  $[s,e] \,\subset\,[1,T]$  an interval  satisfying
	\[
	\eta_{i_0} \leq s < \eta_{i_0+1} <\ldots  < e \leq \eta_{i_0+l}, \\
	\]
	for some $l \,\geq \, 0$. Suppose also that 
	\begin{equation}
	\label{eqn:eq7}
	\max\{ \min\{ \eta_{i_0+1} - s, s-  \eta_{i_0}  \},  \min\{ \eta_{i_0+l}- e, e- \eta_{i_0+l-1}\} \} \,= \, \epsilon. 
	\end{equation}
	Then  we notice  that by Assumption \ref{as2} and (\ref{eqn:as2_lower-bound}),
	\begin{equation}
	\label{eqn:e8}
	\begin{array}{lll}
	\frac{\epsilon  }{ \delta/4  }   & \leq &\frac{ C_2 \, \log( n_{1:T} )  \,\delta^{-3} \kappa^{-1} (e-s)^{7/2}\frac{n_{\max}^{5/2} }{ n_{\min}^3 }}{\delta/4}\\
	& \leq& \frac{ 4C_2  \kappa^{-1} T^{7/2} \left[c_{\alpha}  T^{5\Theta - 9/2}  \right] }{  C_{\alpha}^4 \kappa^{-4}  T^{4\Theta}    }\\
	& \leq&  \left[ 4C_2 c_{\alpha}   C_{\alpha}^{-4} \right] \kappa^3 T^{\Theta -1} \\
	& \leq& 1.
	\end{array}
	\end{equation}
	Hence, for any change point $\eta_{i}$,   it holds that 
	\[
	\vert \eta_{i} -s \vert   \,\leq  \,   \epsilon, \,\,\,\,\,\,\,\,\,\,\,\,\text{or}\,\,\,\,\,\,\,\,\,\,\,\, \vert \eta_{i} -s \vert \,\geq \,  \delta  -  \epsilon \,\geq \,  \frac{3\delta}{4}.
	\]
	Moreover, the same  holds for $e$. 
	
	To continue  the proof, we proceed as in the proof of Theorem  1 \cite{wang2017optimal}. To that end, we assume that the event  
	\begin{equation}
	\label{eqn:concentration2}
	\underset{ 1\leq  s\leq b <e \leq   T}{\max }\, \underset{z \in \mathbb{R}}{\sup}  \,\,   \vert  D_{s,e}^b(z) - \Delta_{s,e}^b(z)\vert\,\leq \,   C\sqrt{\log n_{1:T}} ,
	\end{equation}
	holds.
	
	\paragraph{Step 1.} Suppose that there exists $\eta_i \in (s,e)$  such that $\min\{ e - \eta_i, \eta_i - s  \} \,\geq \, 3\delta/4$.  Then set 
	\[
	\gamma \,:=\, \frac{ n_{\max}^5 }{ n_{\min}^{6}  } \log n_{1:T}\,\leq \, c_{\alpha}^2 T^{10\Theta - 9}\,\leq \,  c_{\alpha}^2 T^{\Theta}\,\leq \,  c_{\alpha }^2  C_{\alpha}^{-1} \kappa\delta \,\leq \, \frac{3 \delta}{32}.
	\]
	Therefore,   by Lemma \ref{lem3}
	\begin{equation}
	\label{eqn:aux1}
	\underset{  k =   \ceil{s + \gamma},\ldots, \floor{e - \gamma}  }{\max} \,\Delta_{s+1,e}^{k} \,\geq \,  \frac{3 \kappa  \delta n_{\min}  }{8   \sqrt{(e-s)n_{\max}} }.
	\end{equation}
	The latter inequality combined with (\ref{eqn:concentration2}) and the proof of Lemma \ref{lem:concetration} imply 
	\begin{equation}
	\label{eqn:aux2}
	\begin{array}{lll}
	\underset{  k =   \ceil{s + \gamma},\ldots, \floor{e - \gamma}  }{\max} \,D_{s+1,e}^{k} &\geq &  	\underset{  k =   \ceil{s + \gamma},\ldots, \floor{e - \gamma}  }{\max} \, \Delta_{s+1,e}^{k}   -  C\sqrt{ \log n_{1:T}}\\
	& \geq &  \frac{3 \kappa  \delta n_{\min}  }{8   \sqrt{(e-s)n_{\max}} } - C\sqrt{ \log n_{1:T}}\\
	& \geq & \frac{ \kappa  \delta n_{\min}  }{8   \sqrt{(e-s)n_{\max}} },
	\end{array}
	\end{equation}
	where the last inequality follows by noticing that from Assumption \ref{as2}
	\begin{equation}
	\label{eqn:e7}C\sqrt{\log n_{1:T}}\,\leq \,   C\, c_{\alpha} T^{5\Theta - 9/2}  \frac{n_{\min}^{3}  }{ n_{\max}^{5/2}  } \,\leq \, Cc_{\alpha} T^{\Theta} T^{-1/2}\frac{ n_{\min}  }{ \sqrt{n_{\max}}  } \,\leq \, \frac{Cc_{\alpha}}{C_{\alpha}} \kappa\delta T^{-1/2}\frac{ n_{\min}  }{ \sqrt{n_{\max}}  } \,\leq \, \frac{ \kappa  \delta n_{\min}  }{64   \sqrt{(e-s)n_{\max}} },
	\end{equation}
	and also
	\begin{equation}
	\label{eqn:e7.2}
	C \sqrt{\log n_{1:T}} \,\leq \,   C\, c_{\alpha} T^{5\Theta - 9/2}  \frac{n_{\min}^{5}  }{ n_{\max}^{9/2}  } \,\leq \,  C\, c_{\alpha} T^{4\Theta - 7/2}  \frac{n_{\min}^{5}  }{ n_{\max}^{9/2}  }  \,\leq \,   \frac{ C\, c_{\alpha} \kappa \delta^{4}}{C_{\alpha}^4 (e-s)^{7/2}   }\frac{n_{\min}^{5}  }{ n_{\max}^{9/2}}.
	%\frac{C\, c_{\alpha}\kappa \delta^{3}}{   (s-s)^{5/2}   }\frac{n_{\min}^{5}  }{ n_{\max}^{9/2} .
	%{c_{\alpha}}{C_{\alpha}} T^{\Theta} T^{-1/2}\frac{ n_{\min}  }{ \sqrt{n_{\max}}  } \,\leq \, \frac{c_{\alpha}}{C_{\alpha}} \kappa\delta T^{-1/2}\frac{ n_{\min}  }{ \sqrt{n_{\max}}  } \,\leq \, \frac{ \kappa  \delta n_{\min}  }{64   \sqrt{(s-s)n_{\max}} }.
	\end{equation}
	
	Therefore from (\ref{eqn:e7}), a change point will be detected provided that (\ref{eqn:tuning}) holds.
	
	If there is no undetected  change point in $(s,e) $ then one of the following cases must hold.
	
	\paragraph{(Case  1). }  There is no change point in $(s,e)$. Then 
	\[
	\underset{  k =   \ceil{s + \gamma},\ldots, \floor{e - \gamma}  }{\max} \, D_{s+1,e}^{k} \,\leq \,  	\underset{  k =   \ceil{s + \gamma},\ldots, \floor{e - \gamma}  }{\max} \, \Delta_{s+1,e}^{k} \,+\,  C \sqrt{\log n_{1:T}}\,=\,C\sqrt{\log n_{1:T}}\,<\, \tau,
	\]
	so no change points are detected.
	
	\paragraph{(Case  2). }There exists two change points  $\eta_{i}, \eta_{i+1}$  within $(s,e) $  and $\max\{ \eta_i - s, e -\eta_{i+1}  \}\,\leq \,\epsilon$. Then, by Lemmas \ref{lem2} and \ref{lem4} 
	\[
	\begin{array}{lll}
	\underset{  k =   \ceil{s + \gamma},\ldots, \floor{s - \gamma}  }{\max} \, D_{s+1,e}^{k} &\leq &  	\underset{  k =   \ceil{s + \gamma},\ldots, \floor{e - \gamma}  }{\max} \, \Delta_{s+1,e}^{k} \,+\,  C\sqrt{\log n_{1:T}}\\
	& \leq& \max\left\{ D_{s+1,e}^{\eta_{i} }, D_{s+1,e}^{\eta_{i+1} }\right\}\\
	&\leq &2\sqrt{n_{\max}\epsilon}     \,+\, C \sqrt{\log n_{1:T}} \\
	& <& \tau,
	\end{array}
	\]
	so no change points are detected.
	
	\paragraph{(Case  3). }  If there is only one  change point  $\eta_i$ in $(s,e)$ such that $\min\{  \eta_i- s, e - \eta_i  \} \,\leq \, \epsilon$, then 
	as in the previous  case, we can show that no change points are detected.
	
	\paragraph{Let us now  assume that there exists a change point  $\eta_i \in (s,e)$  such that   $\min\{    \eta_i - s,  e - \eta_i \}   \,\geq \, \epsilon$.  }
	Then let
	\[
	b \in          	\underset{  k =   \ceil{s + \gamma},\ldots, \floor{e - \gamma}  }{\arg \max} \, D_{s+1,e}^{k}
	\]
	and 
	\[
	z_0 \in  \underset{ z\in \mathbb{R} }{\arg \max}\, \vert    D_{s+1,e}^{b}(z) \vert.
	\]
	Next notice that  condition (\ref{eqn:eq2}) in  Lemma \ref{lem:rate}  holds  combining (\ref{eqn:e7}) with (\ref{eqn:e7.2}). Now observe that (\ref{eqn:aux2}) implies
	\[
	\underset{  k =   \ceil{s + \gamma},\ldots, \floor{e - \gamma}  }{\arg \max} \,\vert D_{s+1,e}^{k}(z_0)\vert \,=\,   	\underset{  k =   \ceil{s + \gamma},\ldots, \floor{e - \gamma}  }{\arg \max} \, D_{s+1,e}^{k}  \,\geq \,\frac{ \kappa  \delta n_{\min}  }{8   \sqrt{(e-s)n_{\max}} }.   
	%	\underset{  k =   \ceil{s + \gamma},\ldots, \floor{s - \gamma}  }{\arg \max} \,\vert D_{s+1,s}^{k}(z_0)\vert  \,-\, C\sqrt{\log n_{1:T} } \,\geq \, 
	\]
	Moreover,  $ k \in  (s,  e+ \gamma) \cup  (s - \gamma , e)  $and $z\in \mathbb{R}$ implies
	\[
	\begin{array}{lll}
	\vert D_{s+1,e}^{k}(z)\vert  &\leq &     \vert \Delta_{s+1,e}^{k}(z_0)\vert    \,+\,C \sqrt{\log n_{1:T}}\\
	& \leq&  2\sqrt{ n_{\max} \gamma  }  \,+\,  \frac{ \kappa  \delta n_{\min}  }{64   \sqrt{(e-s)n_{\max}} }\\
	& \leq & 2\sqrt{ \frac{n_{\max}^{6}}{n_{\min}^{6} } \log n_{1:T }}\,+\, \frac{ \kappa  \delta n_{\min}  }{64   \sqrt{(e-s)n_{\max}} }\\
	& =  & 2\frac{ n_{\min} }{\sqrt{n_{\max}}  }  \frac{n_{\max}^{7/2}}{n_{\min}^{4} } \sqrt{\log n_{1:T}}\,+\, \frac{ \kappa  \delta n_{\min}  }{64   \sqrt{(e-s)n_{\max}} }\\
	& \leq &  2\frac{ n_{\min} }{\sqrt{n_{\max}}  }  c_{\alpha} T^{5\Theta  -   9/2}\,+\, \frac{ \kappa  \delta n_{\min}  }{64   \sqrt{(e-s)n_{\max}} }\\
	& \leq &  2\frac{ n_{\min} }{\sqrt{n_{\max}}  }  c_{\alpha} C_{\alpha}^{-1} \kappa \delta (e- s)^{-1/2}\,+\, \frac{ \kappa  \delta n_{\min}  }{64   \sqrt{(e-s)n_{\max}} }\\
	& \leq &  \frac{ \kappa  \delta n_{\min}  }{32   \sqrt{(e-s)n_{\max}}},
	\end{array}
	\]
	where the second inequality follows  from (\ref{eqn:e7}) and Lemma \ref{lem4},  the third from the definition of $\gamma$,  and the last   three  by Assumption \ref{as2}.
	
	Therefore,
	\[
	b \in          	\underset{ s < k < e }{\arg \max} \,\vert D_{s+1,e}^{k}\vert.
	\]
	The proof concludes by Lemma \ref{lem:rate}.
	
\end{proof}

\subsection{ Proof of Theorem~\ref{thm-wbs} }

\begin{proof}
	%[Proof of Theorem~\ref{thm-wbs}]

Let $\epsilon = C_{\epsilon}\kappa^{-2}\log(n_{1:T}) n^{7}_{\max}n^{-8}_{\min}$.  Since $\epsilon$ is the upper bound of the localisation error, by induction, it suffices to consider any interval $(s, e) \subset (1, T)$	 that satisfies
	\[
		\eta_{k-1} \leq s \leq \eta_k \leq \ldots \leq \eta_{k+q} \leq e \leq \eta_{k+q+1}, \quad q \geq -1,
	\]
	and
	\[
		\max\bigl\{\min\{\eta_k - s, \, s - \eta_{k-1}\}, \, \min \{\eta_{k+q+1} - e, \, e - \eta_{k+q}\} \bigr\} \leq \epsilon,
	\]
	where $q = -1$ indicates that there is no change point contained in $(s, e)$.
	
By Assumption~\ref{assumption-3}, it holds that 
	\[
		\epsilon \leq c\frac{n^7_{\max}}{n^{7}_{\min}} \delta \leq \delta/4.
	\]	
	It has to be the case that for any change point $\eta_k \in (0, T)$, either $|\eta_k - s| \leq \epsilon$ or $|\eta_k - s| \geq \delta - \epsilon \geq 3\delta/4$.  This means that $\min\{|\eta_k - s|, \, |\eta_k - e|\}\leq \epsilon$ indicates that $\eta_k$ is a detected change point in the previous induction step, even if $\eta_k \in (s, e)$.  We refer to $\eta_k \in (s, e)$ an undetected change point if $\min\{|\eta_k - s|, \, |\eta_k - e|\} \geq 3\delta/4$.
	
In order to complete the induction step, it suffices to show that we (i) will not detect any new change point in $(s, e)$ if all the change points in that interval have been previous detected, and (ii) will find a point $b \in (s, e)$ such that $|\eta_k - b| \leq \epsilon$ if there exists at least one undetected change point in $(s, e)$.

For $j = 1, 2$, define the events
	\[
		\mathcal{A}_j(\gamma) = \left\{\max_{1\leq s < b < e \leq T}\sup_{z \in \mathbb{R}}\left|\sum_{k = s}^e w^{(j)}_k\sum_{i = 1}^{n_k} \left\{\mathbbm{1}_{\{Y_{k, i} \leq z\}} - \mathbb{E}\left(\mathbbm{1}_{\{Y_{k, i} \leq z\}}\right) \right\}\right| \leq \gamma\right\},
	\]
	where
	\[
		w_k^{(1)} = \begin{cases}
 				\sqrt{\frac{n_{(b+1):e}}{n_{s:b}n_{s:e}}}, & k = s, \ldots, b, \\
 				-\sqrt{\frac{n_{s:b}}{n_{(b+1):e} n_{s:e}}}, & k = b+1, \ldots, e,
			\end{cases}\, \mbox{and }
		w_k^{(2)} = \frac{1}{\sqrt{n_{s:e}}}	.
	\]
	Define
	\[
		\mathcal{S} = \bigcap_{k = 1}^K\left\{\alpha_s \in [\eta_k - 3\delta/4, \eta_k - \delta/2], \, \beta_s \in [\eta_k + \delta/2, \eta_k + 3\delta/4], \mbox{ for some } s = 1, \ldots, S\right\}.
	\] 
	Set $\gamma$ to be $C_{\gamma} \sqrt{\log(n_{1:T})}$, with a sufficiently large constant $C_{\gamma} > 0$.  The rest of the proof assumes the the event $\mathcal{A}_1(\gamma) \cap \mathcal{A}_2(\gamma) \cap \mathcal{S}$, the probability of which can be lower bounded using Lemma~\ref{lem:concetration} and also Lemma~13 in \cite{wang2018univariate}. \\

\noindent \textbf{Step 1.}  In this step, we will show that we will consistently detect or reject the existence of undetected change points within $(s, e)$.  Let $a_m$, $b_m$ and $m^*$ be defined as in Algorithm~\ref{algorithm:WBS}.   Suppose there exists a change point $\eta_k \in (s, e)$ such that $\min \{\eta_k - s, \, e - \eta_k\} \geq 3\delta/4$.  In the event $\mathcal{S}$, there exists an interval $(\alpha_m, \beta_m)$ selected such that $\alpha_m \in [\eta_k - 3\delta/4, \eta_k - \delta/2]$ and $\beta_m \in [\eta_k + \delta/2, \eta_k + 3\delta/4]$.  

Following Algorithm~\ref{algorithm:WBS}, $(s_m, e_m) = (\alpha_m, \beta_m) \cap (s, e)$.  We have that $\min \{\eta_k - s_m, e_m - \eta_k\} \ge (1/4)\delta$ and $(s_m, e_m)$ contains at most one true change point. 

It follows from Lemma~\ref{lem3}, with $c_1$ there chosen to be $1/4$, that
	\[
		\max_{s_m < t < e_m} \bigl|\Delta^{t}_{s_m, e_m}\bigr|  \geq \frac{\kappa  \delta n_{\min}^{3/2}  }{8    \sqrt{(e_m - s_m)}n_{\max}},
	\]
	Therefore
	\begin{align*}
		 a_m  = \max_{s_m < t < e_m} D^{t}_{s_m, e_m} \geq \max_{s_m < t < e_m} \Delta^{t}_{s_m, e_m} - \gamma  \geq \frac{1}{8\sqrt{C_S}} \kappa \delta^{1/2} \frac{n^{3/2}_{\min}}{n_{\max}} - \gamma.
	\end{align*}
	Thus for any undetected change point $\eta_k \in (s, e)$, it holds that
		\begin{equation}\label{eq:wbsrp size of population}
			a_{m^*} = \sup_{1\le m\le  S} a_m \geq \frac{1}{8\sqrt{C_S}} \kappa \delta^{1/2} \frac{n^{3/2}_{\min}}{n_{\max}} - \gamma \geq c_{\tau, 2}  \kappa \delta^{1/2} \frac{n^{3/2}_{\min}}{n_{\max}},   
		\end{equation}
		where the last inequality is from the choice of $\gamma$ and $c_{\tau, 2} > 0$ is achievable with a sufficiently large $C_{\mathrm{SNR}}$ in Assumption~\ref{assumption-3}.  This means we accept the existence of undetected change points.

Suppose that there are no undetected change points within $(s, e)$, then for any $(s_m, e_m)$, one of the following situations must hold.
	\begin{itemize}
		\item [(a)]	There is no change point within $(s_m, e_m)$;
		\item [(b)] there exists only one change point $\eta_k \in (s_m, e_m)$ and $\min\{\eta_k - s_m, e_m - \eta_k\} \le \epsilon_k$; or
		\item [(c)] there exist two change points $\eta_k, \eta_{k+1} \in (s_m, e_m)$ and $\eta_k - s_m \leq \epsilon_k$, $e_m - \eta_{k+1} \leq \epsilon_{k+1}$.
	\end{itemize}

Observe that if (a) holds, then we have
	\[
		\max_{s_m < t < e_m} D^{t}_{s_m, e_m}  \leq \max_{s_m < t < e_m}  \Delta^{t}_{s_m, e_m}+ \gamma = \gamma  <   \tau,
	\]
	so no change points  are detected.
	
	Cases (b) and (c) are similar, and case (b) is simpler than (c), so we will only focus on case (c).  It follows from Lemma~\ref{lem4} that
	\[
		\max_{s_m < t < e_m} \Delta^t_{s_m, e_m} \leq \sqrt{n_{\max}}\sqrt{e_m - \eta_{k+1}}\kappa_{k+1} + \sqrt{n_{\max}} \sqrt{\eta_k - s_m} \kappa_k \leq 2C_{\epsilon} \sqrt{\log (n_{1:T})},
	\]
	therefore
	\[
		\max_{s_m < t < e_m} D^t_{s_m, e_m} \leq \max_{s_m < t < e_m} \Delta^t_{s_m, e_m} + \gamma \leq 2C_{\epsilon} \sqrt{\log (n_{1:T})} + C_{\gamma} \sqrt{\log (n_{1:T})} < \tau.
	\]
	Under \eqref{eq-thm4-tau}, we will always correctly reject the existence of undetected change points. \\

\noindent \textbf{Step 2.}  Assume that there exists a change point $\eta_k \in (s, e)$ such that $\min\{\eta_k - s, \eta_k - e\} \ge 3\delta/4$.  Let $s_m$, $e_m$ and $m^*$ be defined as in Algorithm~\ref{algorithm:WBS}.  To complete the proof it suffices to show that, there exists a change point $\eta_k \in (s_{m*}, e_{m*})$ such that $\min\{\eta_k - s_{m*}, \eta_k - e_{m*}\} \geq \delta/4$ and $|b_{m*} - \eta_k| \leq \epsilon$.

To this end, we are to ensure that the assumptions of Lemma~\ref{lem-21} are verified.  Note that \eqref{eq-lem13-2} follows from \eqref{eq:wbsrp size of population}, \eqref{eq-lem13-1} and \eqref{eq-lem13-4} follow from the definitions of events $\mathcal{A}_1(\gamma)$ and $\mathcal{A}_2(\gamma)$, and \eqref{eq-lem13-3} follows from Assumption~\ref{assumption-3}.

Thus, all the conditions in Lemma~\ref{lem-21} are met. Therefore, we conclude that there exists a change point $\eta_{k}$, satisfying
	\begin{equation}
		\min \{e_{m^*}-\eta_k,\eta_k-s_{m^*}\} > \delta /4 \label{eq:coro wbsrp 1d re1}
	\end{equation}
	and
	\[
		| b_{m*}-\eta_{k}| \leq C\frac{n^9_{\max}}{n^{10}_{\min}}\kappa^{-2}\gamma^{2} \leq \epsilon,
	\]
	where the last inequality holds from the choice of $\gamma$ and Assumption~\ref{assumption-3}.

The proof is completed by noticing  that \eqref{eq:coro wbsrp 1d re1} and  $(s_{m^*}, e_{m^*}) \subset (s, e)$ imply that
	\[
		\min \{e-\eta_k,\eta_k-s\}  >  \delta /4 > \epsilon.
	\]
	As discussed in the argument before {\bf Step 1}, this implies that $\eta_k $ must be an undetected change point.	
\end{proof}

\section{Auxiliary lemmas and proofs}

\begin{lemma}	
\label{lem:rate}
	Let  $[s,e] \,\subset\,[1,T]$  an interval  satisfying
	\[
	\eta_{i_0} \leq s < \eta_{i_0+1} <\ldots  < e \leq \eta_{i_0+l}, \\
	\]
	for some $l \,\geq \, 1$. Let   
	$$b \in  \underset{  s <  k < e }{\arg  \max}  \,   D_{ s+1,e}^{k} $$
	and
	\[
	z_0 \in  \underset{z \in \mathbb{R}}{\arg \min}\,\vert D_{ s+1,e}^{b}(z)\vert .
	\]
	Suppose that for   some  $c_1 > 0$  and $\epsilon >0$ we have  that
	\begin{equation}
	\label{eqn:eq1}
	\max\{ \min\{ \eta_{i_0+1} - s, s -  \eta_{i_0}  \},  \min\{ \eta_{i_0+l}- e, e- \eta_{i_0+l-1}\} \} \,= \, \epsilon
	\end{equation}
	with 
	\begin{equation}
	\label{eqn:eq2}
	\epsilon\,<\, \min\left\{  \frac{1}{8}\left( \frac{3 c_1}{4}  \right)^2 \frac{\kappa^2 \delta^2  }{(e- s)}\frac{ n_{\min} }{n_{\max}  } , \frac{\delta}{4}  \right\},
	\end{equation}
	and
	\begin{equation}
	\label{eqn:e3}
	D_{ s+1,e}^{b}(z_0)\,\geq \,  \frac{c_1 \kappa  \delta n_{\min}  }{ \sqrt{(s-e)n_{\max}} }.
	\end{equation}
	In addition, assume that there  exists  a constant  $c_3 > 0$ such that 
	\begin{equation}
	\label{eqn:e4}
	\underset{ s < k < e }{\max }\,\vert   D_{s+1,e}^k(z_0) -  \Delta_{s+1,e}^k(z_0) \vert \,\leq \, C\sqrt{\log n_{1:T} }\,<\, \min\left\{ c_3 \frac{\kappa \delta^4 n_{\min}^5}{(e-s)^{7/2} n_{\max}^{9/2}} ,\frac{c_1 \kappa  \delta n_{\min}  }{ 4\sqrt{(s-e)n_{\max}} } \right\},
	\end{equation}
	%\frac{\kappa \delta^4 n_{\min}^5}{(e-s)^{7/2} n_{\max}^{9/2}}
	%\frac{  \kappa \delta^3 }{(s-e)^{5/2}} \frac{n_{\min}^5 }{ n_{\max}^{9/2} }
	where $C>0$ is an appropriate constant.
	Then  there exists   a change point $\eta_k \in [s,e] $  such that
	\begin{equation}
	\label{eqn:conclusion}
	\begin{array}{l}
	\min\{  \eta_{k} -s, e -\eta_{k}  \} \,>\,  \frac{1}{8} \left( \frac{3 c_1}{4}  \right)^2 \frac{\kappa^2 \delta^2  }{(e - s)}  \frac{ n_{\min } }{n_{\max}  },\\
	\vert \eta_k -b  \vert \,\leq \,
	C_1 \sqrt{\log n_{1:T} } \,\delta^{-3} \kappa^{-1} (e-s)^{7/2}\frac{n_{\max}^{5/2} }{ n_{\min}^3 } ,  \\
	\vert \Delta_{s+1,e}^{ \eta_k }(z_0) \vert \,\geq \,  \vert	D_{ s+1,e}^{b}(z_0) \vert   -  C\sqrt{\log n_{1:T} } \,\geq \,
	\underset{s < k < s}{\max} \, \vert	\Delta_{ s+1,e}^{k}(z_0) \vert  -  2 C\sqrt{\log n_{1:T} } ,
	\end{array}
	\end{equation}
	for some  positive constant  $C_1$.
\end{lemma}

%\textcolor{Orchid}{Oscar, can you double check the ?? in the proof?}

\begin{proof}
	
	Let $\lambda_1 = C \sqrt{\log n_{1:T}}$, where is $C>0$ is a constant   such that $C \sqrt{\log n_{1:T}}$  is an upper bound on the left  hand side of the equation from  Remark \ref{rem:1}. 	Then note that
	\begin{equation}
	\label{eqn:first}
	\underset{s< k < e}{\max} \, \vert	\Delta_{ s+1,e}^{k}(z_0) \vert \,\leq \,\underset{s < k < e}{\max} \, \vert	D_{ s+1,e}^{k}(z_0) \vert\,+\, \lambda_1 \,=\,  
	\vert	D_{ s+1,e}^{b}(z_0) \vert\,+\, \lambda_1 \,\leq \, \vert	\Delta_{ s+1,e}^{b}(z_0) \vert\,+\,2 \lambda_1. 
	\end{equation}
	Let  us assume that  $\eta_{i_0 + i} \leq  b \leq  \eta_{i_0 + i+1}$  for some $i \in \{1,\ldots,l\}$.
	Then from (\ref{eqn:e3}) and (\ref{eqn:e4}) we have that 
	\begin{equation}
	\label{eqn:e5}
	\vert	\Delta_{ s+1,e}^{b}(z_0) \vert \,\geq \,  \vert	D_{ s+1,e}^{b}(z_0) \vert - \lambda_1 \,>\, \frac{(3c_1/4) \kappa  \delta n_{\min}  }{ \sqrt{(s-s)n_{\max}} } \,>\,0.
	\end{equation}
	Next,  without loss of generality let us assume that $\Delta_{ s+1,e}^{b}(z_0)  >0$. Then by the construction in the proof of Lemma \ref{lem2} and by Lemma  2.2  from \cite{venkatraman1992consistency} it follows that the function $k \rightarrow  \Delta_{ s+1,e}^{k}(z_0) $ is either monotonic  or  decreasing and then increasing on $[\eta_{i_0 + i},\eta_{i_0 + i+1} ]$. Therefore, 
	\[
	\max\left\{ 	\Delta_{s+1,e}^{  i_0+i}(z_0) ,	\Delta_{ s+1,e}^{  i_0+i+1  }(z_0)   \right\} \,\geq \, 	\Delta_{ s+1,e}^{b}(z_0).
	\]
	Let us assume that  $\Delta_{ s+1,e}^{k}(z_0)$ is locally decreasing at  $b$. Then
	\begin{equation}
	\label{eqn:e6}
	\Delta_{ s+1,e}^{  i_0+i}(z_0)  \,\geq \, 	\Delta_{ s+1,e}^{ b}(z_0) \,>\,\frac{(3c_1/4) \kappa  \delta n_{\min}  }{ \sqrt{(e-s)n_{\max}} },
	\end{equation}
	by (\ref{eqn:e5}). 
	
	Let us now suppose that 
	\begin{equation}
		\label{eqn:interval}
			\min\{ e- \eta_{i_0+i} , \eta_{i_0+i}-s \} \,\leq \,    \frac{1}{8}  \left( \frac{3 c_1}{4}  \right)^2 \frac{\kappa^2 \delta^2  }{(e -s)} \frac{n_{\min}}{ n_{\max}  }. 
	\end{equation}
	Then by  Lemma \ref{lem4}, and the fact that $  n_{\max } \,\leq \,  2 n_{\min}$,  we  arrive to a contradiction to 	(\ref{eqn:e6}). This proves the first part of the lemma.
	
	\paragraph{	We now claim that} 
	\begin{equation}
	\label{eqn:lower}
		\min\{ \eta_{i_0 + i} -s, e - \eta_{i_0 + i}  \} \,\geq \, (3/4)\delta.
	\end{equation}
	Let us proceed by contradiction.  If $\eta_{i_0 + i} - s \,\leq \, (3/4)\delta  $ then   $s -  \eta_{i_0 + i-1} \geq  \delta/4$  as  $\eta_{i_0 + i} - \eta_{i_0 + i-1} \geq \delta$.  Hence,  $i = 1$. Moreover, by hypothesis $\min\{\eta_{i_0 + 1}-s, s -\eta_{i_0}   \} \,\leq\,\epsilon \,<\,  \delta/4$. And so, it must be the case  that 
	$$ \eta_{i_0 + 1} - s \,\leq \,  \epsilon  \,\leq \,   \frac{1}{8}  \left( \frac{3 c_1}{4}  \right)^2 \frac{\kappa^2 \delta^2  }{(e -s)}\frac{ n_{\min} }{n_{\max}  },  $$
	which implies (\ref{eqn:interval}), and  as we saw before that leads to a contradiction.
	
	\paragraph{	To conclude the proof,}  we use Lemma  \ref{lem:local}  combined with  (\ref{eqn:e4})  to obtain that there exists a constant $C_1 >0$ and a $d$ such that
	\[
	d \in  \left[ \eta_{i_0 + i},  \eta_{i_0 + i} +   C_1 \lambda_1 \,\delta^{-3} \kappa^{-1} (e-s)^{7/2}\frac{n_{\max}^{5/2} }{ n_{\min}^3 }    \right]%\frac{  \kappa \delta^3 }{(s-s)^{5/2}} \frac{n_{\min}^5 }{ n_{\max}^{9/2} }    \right]\unde
	\]
	and 
	\[
	\begin{array}{lll}
	\Delta_{ s+1,e}^{v_{i_0+i} }(z_0) -  \Delta_{s+1,e}^{  d}(z_0)&>&c \left[\frac{  n_{\min}^2 }{ n_{\max}^2 }\right]\delta\,\Delta_{ s+1,e}^{\eta_{i_0+i} }(z_0)   \,\vert d-   \eta_{i_0+1} \vert (e-s)^{-2}\\
	&\geq  & c\left[\frac{(3c_1/4) \kappa  \delta n_{\min}  }{ \sqrt{(e-s)n_{\max}} }\right]\left[\frac{  n_{\min}^2 }{ n_{\max}^2 }\right] \,\delta\,\vert d-   \eta_{i_0+1} \vert (e-s)^{-2}\\
	& =& c \left[\frac{(3c_1/4) \kappa  \delta n_{\min}  }{ \sqrt{(e-s)n_{\max}} }\right] \left[\frac{  n_{\min}^2 }{ n_{\max}^2 }\right] \,\delta\, (e-s)^{-2}\left[  C_1 \lambda_1 \,\delta^{-3} \kappa^{-1} (e-s)^{7/2}\frac{n_{\max}^{5/2} }{ n_{\min}^3 }  \right]\\
	 &  \geq  &  c C_1    \frac{e-s}{\delta }  \lambda_1   \\
	& \geq& 2\lambda_1,
	\end{array}
	\]
	where the second inequality follows from (\ref{eqn:e6}), and the last  one  by noticing that $e-s \geq \delta$ because  of (\ref{eqn:lower}). Hence, if $b \,\geq \, d$ then we would have that $\Delta_{ s+1,e}^k(z_0)$ is locally decreasing on $[\eta_{i_0 + i},b]$ and $[\eta_{i_0 + i},d]$. Thus,
	\[
	\Delta_{ s+1,e}^b(z_0) \,\leq \,  \Delta_{ s+1,e}^d(z_0)\,\leq \,\Delta_{ s+1,e}^{v_{i_0+i} }(z_0) -2 \lambda_1 \,\leq\, \underset{s<k <e}{\max} \vert \Delta_{ s+1,e}^{k }(z_0)  \vert  \,-\,  2\lambda_1, 
	\]
	which contradicts (\ref{eqn:first}).
\end{proof}

%Lemma~\ref{lem2} plays the role of Lemma~2.2 in \cite{venkatraman1992consistency} and Lemma~15 in \cite{wang2018univariate}.  Lemma~\ref{lem:concetration} controls the deviance between sample and population Kolmogorov--Smirnov statistics.  Lemma~\ref{lem3} is the density version of Lemma~2.4 in \cite{venkatraman1992consistency}.  Lemma~\ref{lem:local} plays the role of Lemma~2.6 of \cite{venkatraman1992consistency}.  Lemma~\ref{Lem-17} is essentially Lemma~17 in \cite{wang2018univariate}.  Lemma~\ref{lem-19} is Lemma~19 in \cite{wang2018univariate}.
 
\begin{lemma}
\label{lem2}
Under Assumption~\ref{as1}, for any pair $(s, e) \subset (1, T)$ satisfying
	\[
	\eta_{k-1} \leq s \leq \eta_k \leq \ldots \leq \eta_{k+q} \leq e \leq \eta_{k+q+1}, \quad q \geq 0,
	\]	
	let
	\[
	   b_1 \in \underset{b = s +1, \ldots, e-1}{\arg \max}\,\, \Delta_{s,e}^b.
	\]
	Then $ b_1 \in \{\eta_1,\ldots,\eta_K \} $.
	
Let $z \in \argmax_{x \in \mathbb{R}} |\Delta^b_{s, e}(x)|$.  If $\Delta^{t}_{s, e}(z) > 0$ for some $t \in (s, e)$, then $\Delta_{s, e}^t(z)$ is either monotonic or decreases and then increases within each of the interval $(s, \eta_k), (\eta_{k}, \eta_{k+1}), \ldots, (\eta_{k+q}, e)$.
\end{lemma}

\begin{proof}
	Let us proceed by contradiction assuming that $b_1 \notin \{\eta_1, \ldots, \eta_{k}\}$.  Now let  $z_0$  be such that 
	\[	
		z_0 \in \argmax_{z \in \mathbb{R}} \vert\Delta_{s,e}^{b_1}(z)\vert.
	\]
	Note that due to the fact for any CDF function $F: \mathbb{R} \to [0, 1]$, it holds that $F(-\infty) = 1 - F(\infty) = 0$, we have that $z_0 \in \mathbb{R}$ exists.
	
	Therefore, 
	\[
		b_1 \in \underset{b = s+1, \ldots, e-1}{\arg \max} \vert \Delta_{s,e}^b(z_0)\vert. 
	\]
	Next consider the time series $\{r_l(z_0)\}_{l= 1}^{n_{s:e}}$  defined as
	\[
		r_l(z_0) = \begin{cases}
			F_{s}(z_0) & l \in \{1, \ldots, n_s\} \\
			F_{s+1}(z_0) & l \in \{n_s + 1, \ldots, n_{s:(s+1)}\} \\
			\ldots \\
			F_e(z_0) & l  \in \{ n_{s:(e-1)} + 1, \ldots, n_{s:e} \},
		\end{cases}
	\]
	and for $1\leq  l< n_{s:e}$ define  
	\[
		\tilde{r}_{1, n_{s:e}}^l(z_0) = \sqrt{\frac{n_{s:e} - l}{n_{s:e} l}} \sum_{t=1}^l r_t(z_0) - \sqrt{\frac{l}{n_{s:e}(n_{s:e}-l)}} \sum_{t = l+1}^{n_{s:e}} r_t(z_0).
	\]
	We also notice that the set of change points of the time series $\{r_l(z_0)\}_{l=1}^{n_{s:e}}$ is 
	\[
		\left\{n_{s:\eta_k}, \ldots, n_{s: \eta_{k +q}} \right\}.
	\]

We then notice that by Lemma~2.2 from \cite{venkatraman1992consistency} applied to  $\{ r_l(z_0)\}_{l=1}^{n_{s:e}}$, we have that
	\[
		\Delta_{s,e}^{b_1} = \Delta_{s,e}^{b_1}(z_0) = \tilde{r}_{1,n_{s:e}}^{n_{s:b_1} }(z_0) < \underset{   j \in \{k,\ldots, k +q \} }{\max}\tilde{r}_{1,n_{s:e}}^{    n_{s:\eta_{j} }     }(z_0) = \underset{   j \in \{k, \ldots,k +q \} }{\max} \Delta_{s,e}^{ \eta_{j} }(z_0) \leq \underset{   j \in \{k, \ldots,k +q \} }{\max}  \Delta_{s,e}^{ \eta_{j} },
	\]
	which is a contradiction.

\end{proof}

\begin{lemma}\label{lem4}
Under Assumption~\ref{as1}, let $t \in (s, e)$.  It holds that 
	\begin{equation}\label{eq-lem4-1}
		\Delta_{s,e}^t \leq 2\sqrt{n_{\max}} \min\{\sqrt{s-t + 1},\, \sqrt{e-t}\}.
	\end{equation}
	
	If $\eta_k$ is the only change point in $(s, e)$, then
	\begin{equation}\label{eq-lem4-2}
		\Delta_{s, e}^{\eta_k} \leq \kappa_k \sqrt{n_{\max}} \min\{\sqrt{s-\eta_k+1}, \, \sqrt{e-\eta_k}\}.
	\end{equation}
	
	If $(s, e) \subset (1, T)$ contain two and only two change points $\eta_k$ and $\eta_{k+1}$, then we have
	\begin{equation}\label{eq-2cpt}
		\max_{t = s+1, \ldots, e-1} \Delta^t_{s, e} \leq \sqrt{n_{\max}}\sqrt{e - \eta_{k+1}} \kappa_{k+1} + \sqrt{n_{\max}} \sqrt{\eta_k - s} \kappa_k.
	\end{equation}
\end{lemma}

\begin{proof}
	As for \eqref{eq-lem4-1}, it follows from that
		\[
			\Delta_{s,e}^b   \,\leq \,  2 \sqrt{  \frac{n_{s:b}   n_{(b+1):e} }{n_{s:e}}  }   \,\leq \,  2  \min\{  \sqrt{ n_{s:b} }  , \sqrt{ n_{(b+1):e} }    \} \,\leq \,  2\sqrt{n_{\max}} \,\min\{ \sqrt{s-b+1},  \sqrt{e-b}  \}.
		\]
		As for \eqref{eq-lem4-2}, it is due to that
		\begin{align*}
			\Delta_{s, e}^{\eta_k}  = \sqrt{\frac{n_{s:\eta_k}n_{(\eta_k + 1):e}}{n_{s:e}}} \sup_{z \in \mathbb{R}} \bigl|F_{s:\eta_k}(z) - F_{(\eta_k+1):e}\bigr|  \leq \kappa_k \sqrt{n_{\max}} \min\{\sqrt{s-\eta_k+1}, \, \sqrt{e-\eta_k}\}.
		\end{align*}
		Eq.~\eqref{eq-2cpt} follows similarly.
		
\end{proof}

\begin{lemma}
\label{lem:concetration}
Under Assumption~\ref{as1}, for any $1 \leq s < b < e \leq T$ and $z \in \mathbb{R}$, define 
		\[
			\Lambda_{s, e}^b(z) = D^b_{s, e}(z) - \Delta^b_{s, e}(z),
		\]
		where $D^b_{s, e}(z)$ and $\Delta^b_{s, e}(z)$ are the sample and population versions of the Kolmogorov--Smirnov statistic defined in \eqref{eq:ks_z} and \eqref{eq-pop-KS-z}, respectively.  It holds that
		\begin{align*}
		& \mathbb{P}\left\{\max_{1\leq s < b < e \leq T}\sup_{z \in \mathbb{R}}\left|\Lambda_{s, e}^b(z)\right| > \sqrt{\log\left(\frac{T^4}{12 \delta}\right) + \log(n_{1:T})} + 6\sqrt{\log(n_{1:T})} + \frac{48\log(n_{1:T})}{\sqrt{n_{1:T}}}\right\} \\
		\leq & \frac{12 \log(n_{1:T})}{T^3 n_{1:T}} + \frac{24 T}{n_{1:T} \log(n_{1:T}) \delta}.
		\end{align*}
		Moreover
		\begin{align}
			& \mathbb{P}\bigg\{\max_{1\leq s < e \leq T}\sup_{z \in \mathbb{R}}\left|\frac{1}{\sqrt{n_{s:e}}}\sum_{t = s}^e\sum_{i = 1}^{n_t} \left\{\mathbbm{1}_{Y_{t, i} \leq z} -  \mathbb{E}(\mathbbm{1}_{Y_{t, i} \leq z} ) \right\}\right| \nonumber \\
			& \hspace{3cm} > \sqrt{\log\left(\frac{T^4}{12 \delta}\right) + \log(n_{1:T})} + 6\sqrt{\log(n_{1:T})} + \frac{48\log(n_{1:T})}{\sqrt{n_{1:T}}} \bigg\} \nonumber \\
			\leq & \frac{12 \log(n_{1:T})}{T^3 n_{1:T}} + \frac{24 T}{n_{1:T} \log(n_{1:T}) \delta}. \label{eq-lem7-other}
		\end{align}
\end{lemma}

\begin{remark}
	\label{rem:1}
Lemma~\ref{lem:concetration} shows that as $T$ diverges unbounded, it holds that
	\[
		\max_{1\leq s < b < e \leq T}\sup_{z \in \mathbb{R}}\left|\Lambda_{s, e}^b(z)\right| = O_p\left(\sqrt{\log(n_{1:T})}\right).
	\]	
\end{remark}

\begin{proof}[Proof of Lemma~\ref{lem:concetration}]
	For any $1 \leq s < b < e \leq T$ and $z \in \mathbb{R}$, let
	\begin{align*}
   \sqrt{\frac{  n_{s:b}\,   n_{(b+1):e} }{n_{s:e}} }  \left[	\widehat{F}_{s:b}(z) - \widehat{F}_{(b+1):e}(z)  \right]& = \sum_{k=s}^b\sum_{i = 1}^{n_k} \sqrt{\frac{n_{(b+1):e}}{n_{s:b}n_{s:e}}} \mathbbm{1}_{\{Y_{k, i} \leq z\}} - \sum_{k=b+1}^e\sum_{i = 1}^{n_k} \sqrt{\frac{n_{s:b}}{n_{(b+1):e}n_{s:e}}} \mathbbm{1}_{\{Y_{k, i} \leq z\}} \\
	& = \sum_{k = s}^e w_{k} \sum_{i = 1}^{n_k} \mathbbm{1}_{\{Y_{k, i} \leq z\}},
	\end{align*}
	where 
	\begin{equation}\label{eq-lem7-pf-wk}
	w_{k} = \begin{cases}
	\sqrt{\frac{n_{(b+1):e}}{n_{s:b}n_{s:e}}}, & k = s, \ldots, b; \\
	- \sqrt{\frac{n_{s:b}}{n_{(b+1):e}n_{s:e}}}, & k = b+1, \ldots, e.
	\end{cases}
	\end{equation}	
	Therefore, we have
	\begin{align*}
	&  \sqrt{\frac{  n_{s:b}\,   n_{(b+1):e} }{n_{s:e}} }\left[F_{s:b}(z) - F_{(b+1):e}(z)\right] = \sum_{k = s}^e w_{k} \sum_{i = 1}^{n_k}  \mathbb{E}\bigl(\mathbbm{1}_{\{Y_{k, i} \leq z\}}\bigr), \\
	& D^b_{s, e} = \sup_{z \in \mathbb{R}}\left|\sum_{k = s}^e w_{k}\sum_{i = 1}^{n_k}  \mathbbm{1}_{\{Y_{k, i} \leq z\}}\right|, \quad \mbox{and}\quad \Delta^b_{s, e} = \sup_{z \in \mathbb{R}}\left|\sum_{k = s}^e w_{k} \sum_{i = 1}^{n_k}  \mathbb{E}\bigl(\mathbbm{1}_{\{Y_{k, i} \leq z\}}\bigr)\right|.
	\end{align*}
	
	Since
	\begin{align*}
	D^b_{s, e} & = \sup_{z \in \mathbb{R}}\left|\sum_{k = s}^e  w_{k}\sum_{i = 1}^{n_k}\bigl\{ \mathbb{E}\bigl(\mathbbm{1}_{\{Y_{k, i} \leq z\}}\bigr) + \mathbbm{1}_{\{Y_{k, i} \leq z\}} - \mathbb{E}\bigl(\mathbbm{1}_{\{Y_{k, i} \leq z\}}\bigr)\bigr\}\right| \\
	& \leq \sup_{z \in \mathbb{R}}\left|\sum_{k = s}^e w_{k} \sum_{i = 1}^{n_k}\mathbb{E}\bigl(\mathbbm{1}_{\{Y_{k, i} \leq z\}}\bigr)\right| + \sup_{z\in \mathbb{R}}\left|\sum_{k = s}^e w_k \sum_{i = 1}^{n_k} \bigl\{\mathbbm{1}_{\{Y_{k, i} \leq z\}} - \mathbb{E}\bigl(\mathbbm{1}_{\{Y_{k, i} \leq z\}}\bigr)\bigr\}\right| \\
	& = \Delta^b_{s, e}  + \sup_{z\in \mathbb{R}}\left|\sum_{k = s}^e w_k \sum_{i = 1}^{n_k} \bigl\{\mathbbm{1}_{\{Y_{k, i} \leq z\}} - \mathbb{E}\bigl(\mathbbm{1}_{\{Y_{k, i} \leq z\}}\bigr)\bigr\}\right|,
	\end{align*}
	we have
	\begin{equation}\label{eq-lemma1-1}
	\bigl|D^b_{s, e} - \Delta^b_{s, e}\bigr| \leq \sup_{z\in \mathbb{R}}\left|\sum_{k = s}^e w_{k} \sum_{i = 1}^{n_k}\bigl\{\mathbbm{1}_{\{Y_{k, i} \leq z\}} - \mathbb{E}\bigl(\mathbbm{1}_{\{Y_{k, i} \leq z\}}\bigr)\bigr\}\right|.
	\end{equation}
	
	Next  for $z \in \mathbb{R}$ define
	\[
	\Lambda_{s,e}^b(z) = \sum_{k = s}^e w_k \sum_{i = 1}^{n_k} \bigl\{\mathbbm{1}_{\{Y_{k, i} \leq z\}} - \mathbb{E}\bigl(\mathbbm{1}_{\{Y_{k, i} \leq z\}}\bigr)\bigr\},
	\]
	and let $\{s^k_1, \ldots, s^k_{m-1}\} \subset \mathbb{R}$ satisfy  
	\[
	s_j^k = F^{-1}_k(j/m),
	\]
	where $m$ is a positive integer to be specified.  Let $I^k_1 = (-\infty, s^k_1]$, $I^k_j = (s^k_{j-1}, s^k_j]$, $j = 2, \ldots, m-1$, and $I^k_m = (s_{m-1}^k, \infty)$.  With this notation, for any $k \in \{1, \ldots, n\}$, we get a partition of $\mathbb{R}$, namely $\mathcal{I}_k = \{I^k_1, \ldots, I^k_m\}$.  Let $\mathcal{I}  = \cap_{k = 1}^T \mathcal{I}_k = \{ I_1, \ldots, I_M   \}$.  Note that there are at most $T/\delta$ distinct $\mathcal{I}_k$'s, and therefore $M \leq Tm/\delta$.
	
	Let also  $z_j $  be an interior point of  $I_j$  for all $j \in \{1,\ldots,M\}$.  Then
	\begin{equation}\label{eq-Lambda-decomp}
	\underset{z \in \mathbb{R}}{\sup} \vert \Lambda_{s,e}^b(z)  \vert \leq \underset{ j = 1,\ldots, M}{\max } \left[   \vert  	\Lambda_{s,e}^b(z_j)  \vert   \,+\,  \underset{z \in I_j}{\sup} \vert 	\Lambda_{s,e}^b(z_j) -  	\Lambda_{s,e}^b(z)    \vert  \right].
	\end{equation}
	By the Hoeffding's inequality and a union bound argument, we have for any $\varepsilon > 0$
	\begin{equation}\label{eq-Lambda-1st}
		\mathbb{P}\left\{\underset{1 \leq s < b < e \leq T  }{\max } \,\, \underset{ j = 1,\ldots, M }{\max } \,\,  \vert  	\Lambda_{s,e}^b(z_j)  \vert > \varepsilon \right\} \leq \frac{2T^4 m}{\delta}\exp\bigl(-2\varepsilon^2\bigr),
	\end{equation}
	since 
	\[
	\displaystyle \sum_{k=s}^{e} \sum_{i=1}^{n_k}   w_k^2 \,=\, 1.
	\]
	
	On the other hand,  for $j \in \{ 1,\ldots, M\}$ let  $z \in I_j$  and without loss of generality let us assume that $z_j < z  $.   Let 
	\[
	u_j \,=\,  \left\vert  \{ (i,k) \,:\, k \in \{1,\ldots,T\},i\in \{1,\ldots,n_k\},  \text{and}\,   y_{k,i} \in I_j  \}  \right\vert.
	\]
	Let us also  write   $ r(t) $ if  $\eta_{r(t)-1} +1 \leq  t \leq \eta_{ r(t)}$  for  $t \in \{1,\ldots,T\}$,  and for $I_j  \in \mathcal{I}$  let $q(j,k)$ be such that $I_j \subset I^{r(k)}_{ q(j,k) } $. With this notation  set 
	\[
	   v_j \,=\, \left\vert  \{ (i,k) \,:\, k \in \{1,\ldots,T\},i\in \{1,\ldots,n_k\},  \text{and}\,   y_{k,i} \in I^{r(k)}_{ v(j,k) }  \}  \right\vert.
	\]
	Clearly, $u_j \leq  v_j$  and  $\mathbb{E}(v_j) \,=\,  n_{1:T}/m$.

	We have,
	\begin{equation}
	\label{eqn:approx}
	\begin{array}{lll}
	\vert 	\Lambda_{s,e}^b(z_j) -  	\Lambda_{s,e}^b(z)    \vert    & \leq& \displaystyle  \left\vert \sum_{k=s}^e \sum_{i=1}^{n_k}     w_{k}\bigl\{\mathbbm{1}_{\{y_{k, i} \leq z_j\}} - \mathbbm{1}_{\{y_{k, i} \leq z\}} \bigr\}   \right\vert \,+\, \left\vert \sum_{k=s}^e \sum_{i=1}^{n_k}     w_{k}\bigl\{F_k(z_j) - F_k(z)  \}   \right\vert\\
	& \leq &  \displaystyle  \left\vert \sum_{k=s}^e \sum_{i=1}^{n_k}     \mathbbm{1}_{  \{z_j <  y_{k, i} \leq z\}}  \right\vert^{  \frac{1}{2} } \,\, \,+\,\,  \left[  \sum_{k=s}^e \sum_{i=1}^{n_k}   \vert  w_{k}\vert \right]\, \underset{k=s,\ldots,e}{\max} \vert F_k(z) - F_k(z_j)\vert\\
	& \leq&  \underset{1\leq  j\leq M }{\max}  \,\,u_j^{ \frac{1}{2} }  \,+\,   \frac{2}{m}\sqrt{\frac{n_{(b+1):e} n_{s:b}}{n_{s:e}}}\\
	& \leq&  \underset{1\leq  j\leq M}{\max}  \,\,u_j^{ \frac{1}{2} }  \,+\,  \frac{2\sqrt{n_{1:T}}}{m}.
	\end{array}
	\end{equation}
	
	However, from  the multiplicative Chernoff bound 
	\begin{equation}
	\label{eqn:binom_concentration}
	\begin{array}{lll}
	\mathbb{P}\left(  \underset{1\leq  j\leq M}{\max}  \,\,u_j \geq \frac{3n_{1:T}}{2m} \right)&\leq &  \displaystyle  M \mathbb{P}\left( \,\,u_j \geq  \frac{3n_{1:T}}{2m} \right)\\
	 & \leq  &   \displaystyle  M \mathbb{P}\left( \,\,v_j \geq  \frac{3n_{1:T}}{2m} \right)\\
	& < &\displaystyle \frac{Tm}{\delta}  \exp\left( -\frac{n_{1:T}}{12 m}  \right)\\
	&  \leq  & \exp\left( -\frac{n_{1:T}}{12 m} + \log(T) + \log(m) - \log(\delta)  \right).
	\end{array}
	\end{equation}

Combining \eqref{eq-Lambda-decomp}, \eqref{eq-Lambda-1st}, \eqref{eqn:approx} and \eqref{eqn:binom_concentration}, we have
\begin{align}
& \mathbb{P}\left\{\max_{1\leq s < b < e \leq T}\sup_{z \in \mathbb{R}}\left|\Lambda_{s, e}^b(z)\right| > \epsilon + \sqrt{\frac{3n_{1:T}}{2m}} + \frac{2\sqrt{n_{1:T}}}{m}\right\} \nonumber\\
\leq & \frac{2T^4m}{\delta}\exp(-2\varepsilon^2)  + \exp\left( -\frac{n_{1:T}}{12 m} + \log(T) + \log(m) - \log(\delta)  \right).\label{eq-combine}
\end{align}

Choosing
\[
	\varepsilon = \sqrt{\log\left(\frac{2T^4 m}{\delta}\right)}, \quad \mbox{and }\quad m = \frac{n_{1:T}}{24 \log(n_{1:T})},
\]
\eqref{eq-combine} results in
\begin{align*}
	& \mathbb{P}\left\{\max_{1\leq s < b < e \leq T}\sup_{z \in \mathbb{R}}\left|\Lambda_{s, e}^b(z)\right| > \sqrt{\log\left(\frac{T^4}{12 \delta}\right) + \log(n_{1:T})} + 6\sqrt{\log(n_{1:T})} + \frac{48\log(n_{1:T})}{\sqrt{n_{1:T}}}\right\} \\
	\leq & \frac{12 \log(n_{1:T})}{T^3 n_{1:T}} + \frac{24 T}{n_{1:T} \log(n_{1:T}) \delta}.
\end{align*}

As for the result \eqref{eq-lem7-other}, we only need to change \eqref{eq-lem7-pf-wk} to $w_k = (n_{s:e})^{-1/2}$.
\end{proof}

\begin{lemma}
\label{lem3}
Under Assumption~\ref{as1}, let $1 \leq s < \eta_k < e \leq T$ be any interval satisfying 
	\[
	   \min\{\eta_k - s, \, e - \eta_k\} \geq c_1 \delta,
	\]
	with $c_1 > 0$.  Then we have that 
	\[
	\underset{t = s + 1, \ldots, e - 1}{\max} \, \Delta_{s, e}^{t} \geq  \frac{c_1 \kappa  \delta n_{\min} }{2   \sqrt{(e - s)n_{\max}}} \geq \frac{c_1 \kappa  \delta n_{\min}^{3/2}  }{2   \sqrt{(e - s)}n_{\max}}.
	\]
\end{lemma}

\begin{proof}
Let
	\[
		z_0 \in \underset{z\in \mathbb{R}}{\argmax} \,\,\vert  F_{\eta_{k}}(z)  -  F_{\eta_{k+1}}(z)   \vert.
	\]
	Without loss of generality, assume that $F_{\eta_k}(z_0) > F_{\eta_{k+1}}(z_0)$.  For $s < t < e$, note that 
	\begin{align*}
	\Delta_{s, e}^{t}(z_0)  = \sqrt{\frac{n_{s:e}n_{s:t}}{n_{(t+1):e}}} \left\{\frac{1}{n_{s:t}} \sum_{l = s}^t n_l F_l(z_0) - \frac{1}{n_{s:e}} \sum_{l=s}^e n_l F_l(z_0)\right\}  = \sqrt{\frac{n_{s:e}}{n_{s:t}n_{(t+1):e}}} \sum_{l = s}^t n_l \widetilde{F}_l(z_0),
	\end{align*}
	where $\widetilde{F}_l(z_0) = F_l(z_0) - (n_{s:e})^{-1}\sum_{l = s}^e n_l F_l(z_0)$.
	
Due to Assumption~\ref{as1}, it holds that $\widetilde{F}_{\eta_k}(z_0) > \kappa/2$.  Therefore
	\[
		\sum_{l = s}^{\eta_k} n_l \widetilde{F}_l(z_0) \geq (c_1/2) \kappa n_{\min} \delta, \quad \mbox{and}\quad \sqrt{\frac{n_{s:e}}{n_{s:t}n_{(t+1):e}}} \geq   \frac{1}{\sqrt{(e-s) n_{\max}}}  \geq     \frac{\sqrt{n_{\min}}}{\sqrt{(e-s)}n_{\max}}.
	\]
	Then
	\[
	\underset{t = s + 1, \ldots, e - 1}{\max} \,\Delta_{s, e}^{t}\geq   \frac{1}{2 \sqrt{(e-s) n_{\max}}}    \geq   \frac{c_1 \kappa  \delta n_{\min}^{3/2}  }{2   \sqrt{(e - s)}n_{\max}}.
	\]
\end{proof}

Note that in the following lemma, the condition \eqref{eqn:lower_bound} follows from Lemma~\ref{lem3}, and \eqref{eqn:upper_bound} follows from Lemma~\ref{lem:concetration}.  

\begin{lemma}
\label{lem:local}
Let $z_0 \in \mathbb{R}$, $(s, e) \subset (1, T)$.  Suppose that there exits a true change point $\eta_k \in (s, e)$ such that 
	\begin{equation} \label{eqn:interval_lb}
		\min\{\eta_k - s, \, e - \eta_k\} \geq c_1 \delta,
	\end{equation}
	and 
	\begin{equation} \label{eqn:lower_bound}
		\Delta_{s, e}^{\eta_{k}} (z_0) \geq (c_1/2)  \frac{n^{3/2}_{\min}}{ n_{\max}} \frac{\kappa \delta }{ \sqrt{e - s}},
	\end{equation}
	where $c_1 > 0$ is a sufficiently small constant.  In addition, assume that
	\begin{equation} \label{eqn:upper_bound} 
		\max_{s < t < e}  \vert \Delta_{s, e}^{t} (z_0)\vert -  \Delta_{s, e}^{\eta_k} (z_0) \leq 3\log\left(\frac{T^4}{\delta}\right) + 3\log\bigl(n_{1:T}\bigr) \leq \frac{\kappa \delta^4 n_{\min}^5}{(e-s)^{7/2} n_{\max}^{9/2}}.
	\end{equation}

Then there exists $d \in (s, e)$ satisfying
	\begin{equation} \label{eqn:distance}
		\vert d -  \eta_k\vert \leq  \frac{c_1\delta   n_{\min}^2}{32 n_{\max}^2},
	\end{equation} and
	\[
		\Delta_{s, e}^{\eta_k} (z_0) -  \Delta_{s, e}^{d}(z_0) > c  |d - \eta_k| \delta \frac{n_{\min}^2}{n_{\max}^2} \Delta^{\eta_k}_{s, e}(z_0)(e-s)^{-2},
	\]
	where $c > 0$ is a sufficiently small constant.
\end{lemma}

\begin{proof}
Let us assume without loss of generality that $d \geq  \eta_k$.  Following the argument of Lemma~2.6 in \cite{venkatraman1992consistency}, it suffices to consider two cases: (i) $\eta_{k + 1} > e$ and (ii) $\eta_{k + 1} \leq e$.\\
	
\begin{figure}
\begin{center}
\begin{tikzpicture}
\draw[Fuchsia!100, thick]  (0, 0) -- (10, 0);

\filldraw [Fuchsia!50] (1,0) circle (2pt);	
\draw (1, -0.1) node[below] {$s$};
\filldraw [Fuchsia!50] (3,0) circle (2pt);	
\draw (3, -0.1) node[below] {$\eta_k$};
\filldraw [Fuchsia!50] (5,0) circle (2pt);	
\draw (5, -0.1) node[below] {$d$};
\filldraw [Fuchsia!50] (7,0) circle (2pt);	
\draw (7, -0.1) node[below] {$e$};
\filldraw [Fuchsia!50] (9,0) circle (2pt);	
\draw (9, -0.1) node[below] {$\eta_{k+1}$};

\draw [decorate,decoration={brace,amplitude=10pt},xshift=0pt,yshift=0pt]
(3.2,0.2) -- (5, 0.2) node [black,midway,xshift=0cm, yshift = 0.6cm] {$N_3$};

\draw [decorate,decoration={brace,amplitude=10pt, mirror},xshift=0pt,yshift=0pt]
(1, -0.6) -- (3, -0.6) node [black,midway,xshift=0cm, yshift = -0.6cm] {$N_1$};
\draw [decorate,decoration={brace,amplitude=10pt, mirror},xshift=0pt,yshift=0pt]
(3.2, -0.6) -- (7, -0.6) node [black,midway,xshift=0cm, yshift = -0.6cm] {$N_2$};

\end{tikzpicture}
\end{center}
\caption{Illustrations of Case (i) in the proof of Lemma~\ref{lem:local}. \label{fig-case1}}
\end{figure}
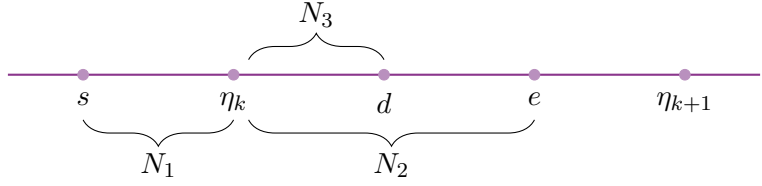

\noindent \textbf{Case (i) $\eta_{k+1} > e$.}  Notice that 
	\[
		\Delta_{s, e}^{\eta_k} (z_0) = \sqrt{\frac{N_1 N_2}{N_1 + N_2}}\bigl\{F_{\eta_k}(z_0) - F_{\eta_{k+1}}(z_0)\bigr\}
	\]
	and
	\[
		\Delta_{s, e}^{d} (z_0) = N_1 \sqrt{\frac{N_2 - N_3}{(N_1 + N_3)(N_1 + N_2)}}\bigl\{F_{\eta_k}(z_0) - F_{\eta_{k+1}}(z_0)\bigr\},
	\]
	where 
	$N_1 = n_{s: \eta_k}$, $N_2 = n_{(\eta_k+1):e}$ and $N_3 = n_{(\eta_k+1): d}$.  Therefore, due to \eqref{eqn:interval_lb}, we have	
	\begin{align}
		E_l & = \Delta_{s, e}^{\eta_k} (z_0) - \Delta_{s, e}^{d} (z_0) = \left(1 - \sqrt{\frac{N_1(N_2 - N_3)}{N_2(N_1 + N_3)}}\right) \Delta^{\eta_k}_{s, e}(z_0) \nonumber \\
		& = \frac{N_1 + N_2}{\sqrt{N_2(N_1 + N_3)}{\bigl(\sqrt{N_2(N_1+N_3)} + \sqrt{N_1(N_2 - N_3)}\bigr)}} N_3 \Delta^{\eta_k}_{s, e}(z_0) \nonumber \\
		& \geq c_1 \frac{n_{\min}^2}{n_{\max}^2} |d - \eta_k| \delta \Delta^{\eta_k}_{s, e}(z_0) (e-s)^{-2}. \label{eq-EL1}
	\end{align}

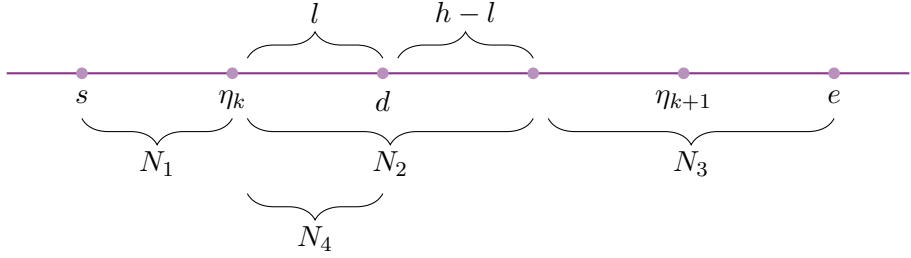
\begin{figure}
\begin{center}
\begin{tikzpicture}
\draw[Fuchsia!100, thick]  (0, 0) -- (12, 0);

\filldraw [Fuchsia!50] (1,0) circle (2pt);	
\draw (1, -0.1) node[below] {$s$};
\filldraw [Fuchsia!50] (3,0) circle (2pt);	
\draw (3, -0.1) node[below] {$\eta_k$};
\filldraw [Fuchsia!50] (5,0) circle (2pt);	
\draw (5, -0.1) node[below] {$d$};
\filldraw [Fuchsia!50] (7,0) circle (2pt);	
%\draw (7, -0.1) node[below] {$d$};
\filldraw [Fuchsia!50] (9,0) circle (2pt);	
\draw (9, -0.1) node[below] {$\eta_{k+1}$};
\filldraw [Fuchsia!50] (11,0) circle (2pt);	
\draw (11, -0.1) node[below] {$e$};

\draw [decorate,decoration={brace,amplitude=10pt},xshift=0pt,yshift=0pt]
(3.2,0.2) -- (5, 0.2) node [black,midway,xshift=0cm, yshift = 0.6cm] {$l$};

\draw [decorate,decoration={brace,amplitude=10pt},xshift=0pt,yshift=0pt]
(5.2,0.2) -- (7, 0.2) node [black,midway,xshift=0cm, yshift = 0.6cm] {$h - l$};

\draw [decorate,decoration={brace,amplitude=10pt, mirror},xshift=0pt,yshift=0pt]
(1, -0.6) -- (3, -0.6) node [black,midway,xshift=0cm, yshift = -0.6cm] {$N_1$};
\draw [decorate,decoration={brace,amplitude=10pt, mirror},xshift=0pt,yshift=0pt]
(3.2, -0.6) -- (7, -0.6) node [black,midway,xshift=0cm, yshift = -0.6cm] {$N_2$};
\draw [decorate,decoration={brace,amplitude=10pt, mirror},xshift=0pt,yshift=0pt]
(7.2, -0.6) -- (11, -0.6) node [black,midway,xshift=0cm, yshift = -0.6cm] {$N_3$};
\draw [decorate,decoration={brace,amplitude=10pt, mirror},xshift=0pt,yshift=0pt]
(3.2, -1.6) -- (5, -1.6) node [black,midway,xshift=0cm, yshift = -0.6cm] {$N_4$};

\end{tikzpicture}
\end{center}
\caption{Illustrations of Case (ii) in the proof of Lemma~\ref{lem:local}. \label{fig-case2}}
\end{figure}

\noindent \textbf{Case (ii) $\eta_{k+1} \leq e$.}  Let  $N_1 =  n_{s: \eta_k}$, $N_2 = n_{(\eta_k+1): (\eta_k+h)}$ and $N_3 = n_{(\eta_k+h+1): e}$,  where  $h =  c_1 \delta/8$.  Then,
    \[
		\Delta_{s, e}^{\eta_k} (z_0) = a \sqrt{\frac{N_1 + N_2 + N_3}{N_1(N_2 + N_3)}}, \, \mbox{and }\Delta_{s, e}^{\eta_k + h} (z_0) = (a+ N_2 \theta) \sqrt{\frac{N_1 + N_2 + N_3}{N_3(N_1 + N_2)}}.
    \]
    where 
    \[
    	a = \sum_{l = s}^{\eta_k} n_l F_l(z_0) - c_0, \quad c_0 = \frac{1}{n_{s:e}} \sum_{l = s}^e n_l F_l(z_0)
    \]
    and
    \begin{align*}
    \theta = \frac{a\sqrt{(N_1 + N_2)N_3}}{N_2} \left\{ \frac{1}{\sqrt{N_1(N_2 + N_3)}} - \frac{1}{(N_1 + N_2)N_3} + \frac{b}{a\sqrt{N_1 + N_2 + N_3}}\right\},
    \end{align*}
    with  $b= \Delta_{s, e}^{\eta_k +h} (z_0) - \Delta_{s, e}^{\eta_k} (z_0)$.
    
Next, we set $l = d - \eta_k \leq h/2$ and $N_4 = n_{(\eta_k + 1): d}$.  Therefore, as in the proof of Lemma 2.6 in \cite{venkatraman1992consistency}, we have that 
    \begin{equation}\label{eq-El1}
	    E_l = \Delta_{s, e}^{\eta_k} (z_0) -\Delta_{s, e}^{\eta_k + l} (z_0) = E_{1l}(  1+  E_{2l} )+ E_{3l},
    \end{equation}
    where     
    \[
    E_{1l} = \frac{a N_4(N_2 - N_4)  \sqrt{N_1 + N_2 + N_3}} {\sqrt{N_1(N_2 + N_3)} \sqrt{(N_1 + N_4)(N_2 + N_3 - N_4)} \left(\sqrt{(N_1 + N_4) (N_2 + N_3 - N_4)} +\sqrt{N_1(N_2 + N_3)}\right)},
    \]
    \[
    E_{2l} = \frac{(N_3 - N_1)(N_3 - N_1 - N_4)}{\left(\sqrt{(N_1 + N_4)(N_2 + N_3 - N_4)} + \sqrt{(N_1 + N_2)N_3}\right)\left(\sqrt{N_1(N_2 + N_3)} + \sqrt{(N_1 + N_2)N_3}\right)},
    \]
    and
    \[
    E_{3l} = -\frac{b N_4}{N_2} \sqrt{\frac{(N_1 + N_2)N_3}{(N_1 + N_4)(N_2 + N_3 - N_4)}}.
    \]  
    Next,  we notice that $N_2 - N_4 \geq n_{\min}c_1 \delta/16$.  It holds that
    \begin{equation}\label{eq-El2}
	E_{1l} \geq c_{1l}|d - \eta_k| \delta \frac{n_{\min}^2}{n_{\max}^2} \Delta^{\eta_k}_{s, e}(z_0)(e-s)^{-2}, 		
	\end{equation}
	where $c_{1l} > 0$ is a sufficiently small constant depending on $c_1$.  As for $E_{2l}$, due to \eqref{eqn:distance}, we have
	\begin{align}\label{eq-El3}
		E_{2l} \geq -1/2.
	\end{align}
	As for $E_{3l}$, we have
	\begin{align}
		E_{3l} & \geq -\left\{ 3\log\left(\frac{T^4}{\delta}\right) + 3\log\bigl(n_{1:T}\bigr)\right\}|d - \eta_k| \frac{n_{\min}^2}{n_{\max}^2}	\frac{e-s}{c_1^2 \delta^2} \nonumber \\
		& \geq - c_{3l} \left\{ 3\log\left(\frac{T^4}{\delta}\right) + 3\log\bigl(n_{1:T}\bigr)\right\}|d - \eta_k| \Delta^{\eta_k}_{s, e}(z_0) \delta (e-s)^{-2}\frac{n_{\min}^2}{n_{\max}^2} \nonumber \\
		& \hspace{1cm} \times \frac{n_{\max}^{9/2}}{n_{\min}^5} (e-s)^{7/2}\frac{\log(n_{1:T})}{\kappa \delta^4} \nonumber \\
		& \geq -c_{1l}/2|d - \eta_k| \delta \frac{n_{\min}^2}{n_{\max}^2} \Delta^{\eta_k}_{s, e}(z_0)(e-s)^{-2}, \label{eq-El4}
	\end{align}
	where the first inequality follows from \eqref{eqn:upper_bound}, the identity follows from \eqref{eqn:lower_bound}, and the second inequality follows from \eqref{eqn:upper_bound}.

Combining \eqref{eq-El1}, \eqref{eq-El2}, \eqref{eq-El3} and \eqref{eq-El4}, we have
    \begin{equation}\label{eq-EL2}
		\Delta_{s, e}^{\eta_k} (z_0) -\Delta_{s, e}^{d} (z_0) \geq c  |d - \eta_k| \delta \frac{n_{\min}^2}{n_{\max}^2} \Delta^{\eta_k}_{s, e}(z_0)(e-s)^{-2},   
    \end{equation}
    where $c > 0$ is a sufficiently small constant.  
   
In view of \eqref{eq-EL1} and \eqref{eq-EL2}, we conclude the proof.    
	
\end{proof}

\begin{lemma}
\label{Lem-17}	
Suppose $(s, e) \subset (1, T)$ such that $e - s \leq C_S \delta$ and that
	\[
		\eta_{k-1} \leq s \leq \eta_k \leq \ldots \leq \eta_{k+q} \leq e \leq \eta_{k+q+1}, \quad q \geq 0. 
	\]
	Denote 
	\[
		\kappa_{\max}^{s, e} = \max \bigl\{\kappa_p: \, p = k, \ldots, k+q \bigr\}.
	\]
	Then for any $p \in \{k-1, \ldots, k+q\}$, it holds that
	\[
		\sup_{z \in \mathbb{R}}\left|\frac{1}{n_{s:e}}\sum_{t = s}^e n_t F_t(z) - F_{\eta_p}(z)\right|	\leq (C_S + 1)\kappa^{s, e}_{\max}.
	\]	
\end{lemma}

\begin{proof}

Since $e - s \leq C_S \delta$, the interval $(s, e)$ contains at most $C_S + 1$ true change points.  Note that
	\begin{align*}
	& \sup_{z \in \mathbb{R}}\left|\frac{1}{n_{s:e}}\sum_{t = s}^e n_t F_t(z) - F_{\eta_p}(z)\right|	  \\
	= & \sup_{z \in \mathbb{R}} \Bigg\{\frac{1}{n_{s:e}} \Bigg|\sum_{t = s}^{\eta_k} n_t\left(F_{\eta_{k-1}}(z) - F_{\eta_p}(z)\right) + \sum_{t = \eta_k + 1}^{\eta_{k+1}} n_t \left(F_{\eta_{k}}(z) - F_{\eta_p}(z)\right) + \ldots \\
	& \hspace{6cm} + \sum_{t = \eta_{k+q} + 1}^{e} n_t \left(F_{\eta_{k+q}}(z) - F_{\eta_p}(z)\right)\Bigg|\Bigg\} \\
	\leq & \frac{|p-k|\sum_{t = s}^{\eta_k} n_t + |p-k-1| \sum_{t = \eta_k+1}^{\eta_{k+1}} n_t + \ldots + |p-k-q-1| \sum_{t = \eta_{k+q}+1}^e n_t}{n_{s:e}} \kappa^{s, e}_{\max} \\
	\leq & (C_S + 1)\kappa^{s, e}_{\max}.
	\end{align*}	
\end{proof}

For any $x = (x_i) \in \mathbb{R}^{n_{s:e}}$, define
	\[
		\mathcal{P}^d_{s, e}(x) = \frac{1}{n_{s:e}} \sum_{i = 1}^{n_{s:e}} x_i + \langle x, \psi^d_{s, e}\rangle \psi^d_{s, e},
	\]
	where $\langle \cdot, \cdot \rangle$ is the inner product in Euclidean space, and $\psi_{s,e}^d \in \mathbb{R}^{n_{s:e}}$ with
	\[
		(\psi_{s,e}^d)_i=  \begin{cases}
			\sqrt\frac{n_{(d+1):e}}{n_{s:e}n_{s:d}}, &  i = 1, \ldots, n_{s:d}, \\
			-\sqrt\frac{n_{s:d}}{n_{s:e}n_{(d+1):e}}, &  i = n_{s:d} + 1, \ldots, n_{s:e},
		\end{cases}
	\]
	i.e. the $i$-th entry of $\mathcal{P}_{s,e}^d (x)$ satisfies 
	\[
		\mathcal{P}_{s,e}^d (x)_i= \begin{cases}
			\frac{1}{n_{s:d}}\sum_{j=1}^{n_{s:d}} x_j, &  i = 1, \ldots, n_{s:d}, \\
			\frac{1}{n_{(d+1):e}}\sum_{j=n_{s:d}+1}^{n_{s:e}} x_j, & i = n_{s:d}+1, \ldots, n_{s:e}.
		\end{cases}  
	\]

\begin{lemma}
\label{lem-anova}
Suppose Assumption~\ref{as1} holds, and  consider any interval $(s, e) \subset (1, T)$ satisfying that there exists a true change point $\eta_k \in (s, e)$.  Let 
	\[
	b \in \argmax_{s < t < e} D^t_{s, e} \quad \mbox{and} \quad z_0 \in \argmax_{z \in \mathbb{R}} |D^b_{s, e}(z)|.
	\]
	Let
	\[
	\mu_{s, e} = \bigl(\underbrace{F_{s}(z_0), \ldots, F_s(z_0)}_{n_s}, \ldots, \underbrace{F_{e}(z_0), \ldots, F_e(z_0)}_{n_e}\bigr)^{\top} \in \mathbb{R}^{n_{s:e}}
	\]
	and
	\[
	Y_{s, e} = \left(\underbrace{\mathbbm{1}_{\{Y_{s, 1} \leq z_0\}}, \ldots, \mathbbm{1}_{\{Y_{s, n_s} \leq z_0\}}}_{n_s}, \ldots, \underbrace{\mathbbm{1}_{\{Y_{e, 1} \leq z_0\}}, \ldots, \mathbbm{1}_{\{Y_{e, n_e} \leq z_0\}}}_{n_e}\right)^{\top} \in \mathbb{R}^{n_{s:e}}.
	\]
	
We have
	\begin{equation}\label{eq-anova-1}
		\bigl\|Y_{s, e} - \mathcal{P}^b_{s, e}\bigl(Y_{s, e}\bigr)\bigr\|^2 \leq \bigl\|Y_{s, e} - \mathcal{P}^{\eta_k}_{s, e}\bigl(Y_{s, e}\bigr)\bigr\|^2 \leq \bigl\|Y_{s, e} - \mathcal{P}^{\eta_k}_{s, e}\bigl(\mu_{s, e}\bigr)\bigr\|^2.
	\end{equation}
\end{lemma}

\begin{proof}
Note that for any $d \in (s, e)$, we have
	\begin{align*}
	\bigl\|Y_{s, e} - \mathcal{P}^d_{s, e}\bigl(Y_{s, e}\bigr)\bigr\|^2 & = n_{s:d}(Y_1 - Y_1^2) + n_{(d+1):e} (Y_2 - Y_2^2) \\
	& = -\bigl\{D^t_{s, e}(z_0)\bigr\}^2 + \frac{\left\{\sum_{t = s}^e \sum_{i=1}^{n_t} \mathbbm{1}_{\{Y_{t, i} \leq z_0\}}\right\}^2}{n_{s:e}} - \sum_{t = s}^e \sum_{i=1}^{n_t} \mathbbm{1}_{\{Y_{t, i} \leq z_0\}},
	\end{align*}
	where 
	\[
		Y_1 = \frac{1}{n_{s:d}}\sum_{t = s}^d \sum_{i = 1}^{n_t} \mathbbm{1}_{\{Y_{t, i} \leq z_0\}}, \quad \mbox{and} \quad Y_2 = \frac{1}{n_{(d+1):e}}\sum_{t = d+1}^e \sum_{i = 1}^{n_t} \mathbbm{1}_{\{Y_{t, i} \leq z_0\}}.
	\]
	It follow from the definition of $b$, we have that	
	\[
		\bigl\|Y_{s, e} - \mathcal{P}^b_{s, e}\bigl(Y_{s, e}\bigr)\bigr\|^2 \leq \bigl\|Y_{s, e} - \mathcal{P}^{\eta_k}_{s, e}\bigl(Y_{s, e}\bigr)\bigr\|^2.
	\]
	The second inequality in \eqref{eq-anova-1} follows from the observation that the sum of the squares of errors maximized by the sample mean.

\end{proof}

\begin{lemma}
\label{lem-19}	

Let $(s, e) \subset (1, T)$ contains two or more change points such that
	\[
		\eta_{k-1} \leq s \leq \eta_k \leq \ldots \leq \eta_{k+q} \leq e \leq \eta_{k+q+1}, \quad q \geq 1.
	\]
	If $\eta_k - s \leq c_1 \delta$, for $c_1 > 0$, then
	\[
		\Delta^{\eta_k}_{s, e} \leq \sqrt{\frac{c_1n_{\max}}{n_{\min}}} \Delta^{\eta_{k+1}}_{s, e} + 2\sqrt{n_{s:\eta_k}} \kappa_k.
	\]	
\end{lemma}

\begin{proof}
Consider the distribution sequence $\{G_t\}_{t = s}^e$ be such that
	\[
		G_t = \begin{cases}
			F_{\eta_k + 1}, & t = s+1, \ldots, \eta_k,\\
			F_t, & t = \eta_k + 1, \ldots, e.
	 	\end{cases}	
	\]	
	For any $s < t < e$, define
	\[
		\mathcal{G}_{s,e}^t = \underset{ z \in \mathbb{R}  }{\sup }\left\vert  \mathcal{G}_{s,e}^t(z)  \right\vert,  
	\]
	where
	\[	
		\mathcal{G}_{s,e}^t(z) = \sqrt{\frac{  n_{s:t} \, n_{(t+1):e} }{n_{s:e}} } \left\{\frac{1}{n_{s:t}}\sum_{l = s}^t n_l G_l(z) -   \frac{1}{n_{(t+1):e}}\sum_{l = t+1}^e n_l G_l(z)\right\}.
	\]

For any $t \geq \eta_k$ and $z \in \mathbb{R}$, it holds that
	\[
		\left|\Delta^t_{s, e}(z) - \mathcal{G}^t_{s, e}(z)\right| = \sqrt{\frac{n_{(t+1):e}}{n_{s:e}n_{s:t}}} n_{s:\eta_k} \left|F_{\eta_{k+1}}(z) - F_{\eta_{k}}(z)\right| \leq \sqrt{n_{s:\eta_k}} \kappa_k.
	\]
	Thus we have
	\begin{align*}
	\Delta^{\eta_k}_{s, e} & = \sup_{z \in \mathbb{R}} \left|\Delta^{\eta_k}_{s, e}(z) - \mathcal{G}^{\eta_k}_{s, e}(z) + \mathcal{G}^{\eta_k}_{s, e}(z)\right| \leq \sup_{z \in \mathbb{R}} \left| \Delta^{\eta_k}_{s, e}(z) - \mathcal{G}^{\eta_k}_{s, e}(z) \right| + \mathcal{G}^{\eta_k}_{s, e} \\
	& \leq \mathcal{G}^{\eta_k}_{s, e} + \sqrt{n_{s:\eta_k}} \kappa_k \leq \sqrt{\frac{n_{s:\eta_k} n_{(\eta_{k+1} +1):e}}{n_{s:\eta_{k+1}} n_{(\eta_k + 1): e}}} \mathcal{G}^{\eta_{k+1}}_{s, e} + \sqrt{n_{s:\eta_k}} \kappa_k  \\
	& \leq \sqrt{\frac{c_1n_{\max}}{n_{\min}}} \Delta^{\eta_{k+1}}_{s, e} + 2\sqrt{n_{s:\eta_k}} \kappa_k.
	\end{align*}
	
\end{proof}

\begin{lemma}
\label{lem-21}

Under Assumption~\ref{as1}, let $(s_0, e_0)$ be an interval with $e_0 - s_0 \leq C_S \delta$ and contain at least one change point $\eta_k$ such that
	\[
		\eta_{k-1} \leq s_0 \leq \eta_k \leq \ldots \leq \eta_{k+q} \leq e_0 \leq \eta_{k+q+1}, \quad q \geq 0.
	\]
	Suppose that there exists $k'$ such that 
	\[
		\min \bigl\{\eta_{k'} - s_0, \, e_0 - \eta_{k'} \bigr\} \geq \delta/16.
	\]
	Let 
	\[
		\kappa^{\max}_{s, e} = \max\bigl\{\kappa_p:\, \min\{\eta_p - s_0, \, e_0 - \eta_p\} \geq \delta/16\bigr\}.
	\]
	Consider any generic $(s, e) \subset (s_0, e_0)$, satisfying
	\[
		\min\{\eta_k - s_0, e_0 - \eta_k\} \geq \delta/16, \quad \eta_k \in (s, e).
	\]

Let $b \in \argmax_{s < t < e} D^t_{s, e}$.  For some $c_1 > 0$ and $\gamma > 0$, suppose that
	\begin{equation}\label{eq-lem13-2}
		D^b_{s, e} \geq c_1 \kappa^{\max}_{s, e}\sqrt{\delta} \frac{n^{3/2}_{\min}}{n_{\max}},
	\end{equation}
	\begin{equation}\label{eq-lem13-1}
		\max_{s < t < e}\sup_{z \in \mathbb{R}}\left|\Lambda^t_{s, e}(z)\right| \leq \gamma,
	\end{equation}
	and
	\begin{equation}\label{eq-lem13-4}
		\max_{1 \leq s < e \leq T}\sup_{z \in \mathbb{R}}\left| \frac{1}{\sqrt{n_{s:e}}} \sum_{t = s}^e \sum_{i = 1}^{n_t} \left(\mathbbm{1}_{\{Y_{t, i} \leq z\}} - F_t(z)\right)\right| \leq \gamma.
	\end{equation}
	If there exits a sufficiently small $0 < c_2 < c_1/2$ such that
	\begin{equation}\label{eq-lem13-3}
		\gamma \leq c_2 \kappa^{\max}_{s, e}\sqrt{\delta}\frac{n^{3/2}_{\min}}{n_{\max}},
	\end{equation}
	then there exists a change point $\eta_k \in (s, e)$ such that
	\[
		\min\{e - \eta_k, \, \eta_k - s\} \geq \delta/4 \quad \mbox{and} \quad |\eta_k - b| \leq C\frac{n^9_{\max}}{n^{10}_{\min}}\kappa^{-2}\gamma^{2},
	\]		
	where $C > 0$ is a sufficiently large constant.
\end{lemma}

\begin{proof}

Without loss of generality, assume that $\Delta^b_{s, e} > 0$ and that $\Delta^t_{s, e}$ is locally decreasing at $b$.  Observe that there has to be a change point $\eta_k \in (s, b)$, or otherwise $\Delta^b_{s, e} > 0$ implies that $\Delta^t_{s, e}$ is decreasing, as a consequence of Lemma~\ref{lem2}.  Thus, if $s \leq \eta_k \leq b \leq e$, then
	\begin{equation}\label{eq-lem13-pf-1}
		\Delta^{\eta_k}_{s, e} \geq \Delta^b_{s, e} \geq D^b_{s, e} - \gamma \geq (c_1 - c_2)\kappa^{\max}_{s, e}\sqrt{\delta}\frac{n^{3/2}_{\min}}{n_{\max}} \geq (c_1/2)\kappa^{\max}_{s, e}\sqrt{\delta}\frac{n^{3/2}_{\min}}{n_{\max}},
	\end{equation}
	where the second inequality follows from \eqref{eq-lem13-1}, and the second inequality follows from \eqref{eq-lem13-2} and \eqref{eq-lem13-3}.  Observe that $e - s \leq e_0 - s_0 \leq C_S \delta$ and that $(s, e)$ has to contain at least one change point or otherwise $\max_{s < t < e}\Delta^t_{s, e} = 0$, which contradicts \eqref{eq-lem13-pf-1}.  \\
	
\noindent \textbf{Step 1.}  In this step, we are to show that 
	\begin{equation}\label{eq-lem13-pf-2}
		\min\{\eta_k - s, \, e - \eta_k\} \geq \min\{1, c_1^2\}\delta/16.
	\end{equation}	
	Suppose that $\eta_k$ is the only change point in $(s, e)$.  Then \eqref{eq-lem13-pf-2} must hold or otherwise it follows from \eqref{eq-lem4-2} in Lemma~\ref{lem4}, we have
	\[
		\Delta^{\eta_k}_{s, e} \leq \kappa_k \sqrt{n_{\max}}\frac{c_1\sqrt{\delta}}{4},
	\]
	which contradicts \eqref{eq-lem13-pf-1}.
	
Suppose $(s, e)$ contains at least two change points.  Then $\eta_k - s < \min\{1, \, c_1^2\}\delta /16$ implies that $\eta_k$ is the most left change point in $(s, e)$.  Therefore it follows from Lemma~\ref{lem-19} that
	\begin{align*}
		\Delta^{\eta_k}_{s, e} & \leq \frac{c_1}{4}\sqrt{\frac{n_{\max}}{n_{\min}}}	\Delta^{\eta_{k+1}}_{s, e} + 2\sqrt{n_{s:\eta_k}} \kappa_k \leq \frac{c_1}{4} \sqrt{\frac{n_{\max}}{n_{\min}}} \max_{s < t < e} \Delta^t_{s, e} + \frac{\sqrt{\delta}}{4}c_1\sqrt{n_{\max}} \kappa_k \\
		& \leq \frac{c_1}{4} \sqrt{\frac{n_{\max}}{n_{\min}}} \max_{s < t < e} D^t_{s, e} + \frac{c_1}{4} \sqrt{\frac{n_{\max}}{n_{\min}}} \gamma + \frac{\sqrt{\delta}}{4}c_1\sqrt{n_{\max}} \kappa_k \\
		& \leq \max_{s < t < e} D^t_{s, e} - \gamma,
	\end{align*}
	which contradicts with \eqref{eq-lem13-pf-1}.\\

\noindent \textbf{Step 2.}  It follows from Lemma~\ref{lem:local} that there exits $d \in (\eta_k, \eta_k + c_1\delta n^2_{\min} n^{-2}_{\max}/32)$ and that 
	\begin{equation}\label{eq-lem13-pf-3}
		\Delta^{\eta_k}_{s, e} - \Delta^d_{s, e} \geq 2\gamma.
	\end{equation}
	We claim that $b \in (\eta_k, d) \subset (\eta_k, \eta_k + c_1\delta n^2_{\min} n^{-2}_{\max}/16)$.  By contradiction, suppose that $b \geq d$.  Then
	\begin{equation}\label{eq-lem13-pf-4}
		\Delta^b_{s, e} \leq \Delta^d_{s, e} < \Delta^{\eta_k}_{s, e} - 2\gamma \leq \max_{s < t < e}\Delta^t_{s, e} - 2\gamma \leq \max_{s < t < e} D^t_{s, e} - \gamma = D^b_{s, e} - \gamma,
	\end{equation}
	where the first inequality follows from Lemma~\ref{lem2}, the second follows from \eqref{eq-lem13-pf-3}, and the fourth follows from \eqref{eq-lem13-1}.   Note that \eqref{eq-lem13-pf-4} is a contradiction with \eqref{eq-lem13-pf-1}, therefore we have $b \in (\eta_k, \eta_k + c_1\delta n^2_{\min} n^{-2}_{\max}/32)$. \\

\noindent \textbf{Step 3.}  	It follows from \eqref{eq-anova-1} in Lemma~\ref{lem-anova} that
	\[
		\bigl\|Y_{s, e} - \mathcal{P}^b_{s, e}\bigl(Y_{s, e}\bigr)\bigr\|^2 \leq \bigl\|Y_{s, e} - \mathcal{P}^{\eta_k}_{s, e}\bigl(Y_{s, e}\bigr)\bigr\|^2 \leq \bigl\|Y_{s, e} - \mathcal{P}^{\eta_k}_{s, e}\bigl(\mu_{s, e}\bigr)\bigr\|^2,
	\]
	with the notation defined in Lemma~\ref{lem-anova}.  By contradiction, we assume that
	\begin{equation}\label{eq-lem13-pf-contra}
		\eta_k + C\frac{n^9_{\max}}{n^{10}_{\min}}\kappa^{-2}\gamma^{2} < b,
	\end{equation}
	where $C > 0$ is a sufficiently large constant.  We are to show that this leads to the bound that
	\begin{equation}\label{eq-lem13-pf-5}
		\bigl\|Y_{s, e} - \mathcal{P}^b_{s, e}\bigl(Y_{s, e}\bigr)\bigr\|^2 > \bigl\|Y_{s, e} - \mathcal{P}^{\eta_k}_{s, e}\bigl(\mu_{s, e}\bigr)\bigr\|^2,
	\end{equation}	
	which is a contradiction.
	
Note that $\min\{\eta_k - s, \, e - \eta_k\} \geq \min\{1, c_1^2\}\delta/16$ and $|b - \eta_k| \leq c_1\delta n^2_{\min} n^{-2}_{\max}/32)$.  For properly chose $c_1$, we have
	\[
		\min\{e - b, \, b - s\} \geq \min \{1, \, c_1^2\}\delta/32.
	\]	
	
Note that
	\begin{align*}
		& \bigl\|Y_{s, e} - \mathcal{P}^b_{s, e}\bigl(Y_{s, e}\bigr)\bigr\|^2  - \bigl\|Y_{s, e} - \mathcal{P}^{\eta_k}_{s, e}\bigl(\mu_{s, e}\bigr)\bigr\|^2\\
		= & \bigl\|\mu_{s, e} - \mathcal{P}^b_{s, e}\bigl(\mu_{s, e}\bigr)\bigr\|^2 - \bigl\|\mu_{s, e} - \mathcal{P}^{\eta_k}_{s, e}\bigl(\mu_{s, e}\bigr)\bigr\|^2  + 2\langle Y_{s, e}  - \mu_{s, e}, \mathcal{P}^{\eta_k}_{s, e}\bigl(\mu_{s, e}\bigr) - \mathcal{P}^b_{s, e}\bigl(Y_{s, e}\bigr)\rangle.
	\end{align*}
	Therefore if we can show that 
	\begin{equation}\label{eq-lem13-pf-6}
		2\langle Y_{s, e}  - \mu_{s, e}, \mathcal{P}^b_{s, e}\bigl(Y_{s, e}\bigr) - \mathcal{P}^{\eta_k}_{s, e}\bigl(\mu_{s, e}\bigr)\rangle < \bigl\|\mu_{s, e} - \mathcal{P}^b_{s, e}\bigl(\mu_{s, e}\bigr)\bigr\|^2 - \bigl\|\mu_{s, e} - \mathcal{P}^{\eta_k}_{s, e}\bigl(\mu_{s, e}\bigr)\bigr\|^2,
	\end{equation}
	then \eqref{eq-lem13-pf-5} holds. 
	
As for the right-hand side of \eqref{eq-lem13-pf-6}, we have 
	\begin{align}
		& \bigl\|\mu_{s, e} - \mathcal{P}^b_{s, e}\bigl(\mu_{s, e}\bigr)\bigr\|^2 - \bigl\|\mu_{s, e} - \mathcal{P}^{\eta_k}_{s, e}\bigl(\mu_{s, e}\bigr)\bigr\|^2 = \bigl(\Delta_{s,e}^{\eta_k}(z_0)\bigr)^2 - \bigl(\Delta_{s,e}^{b}(z_0)\bigr)^2 \nonumber\\
		\geq & \bigl(\Delta_{s,e}^{\eta_k}(z_0) - \Delta_{s,e}^{b}(z_0)\bigr) \bigl|\Delta_{s,e}^{\eta_k}(z_0) \bigr|. \label{eq-lem13-47-1}
	\end{align}

We are then to utilize the result of Lemma~\ref{lem:local}.  Note that $z_0$ there can be any $z_0 \in\mathbb{R}$ satisfying conditions thereof.  Equation~\eqref{eqn:lower_bound} holds due to the fact that here we have
	\begin{align}
		\bigl|\Delta^{\eta_k}_{s, e}(z_0)\bigr| \geq & \bigl|\Delta^{b}_{s, e}(z_0)\bigr| \geq \bigl|D^{b}_{s, e}(z_0)\bigr| - \gamma \geq c_1 \kappa^{\max}_{s, e}\sqrt{\delta} \frac{n^{3/2}_{\min}}{n_{\max}} - c_2 \kappa^{\max}_{s, e}\sqrt{\delta}\frac{n^{3/2}_{\min}}{n_{\max}} \nonumber \\
		\geq & (c_1)/2\kappa^{\max}_{s, e}\sqrt{\delta} \frac{n^{3/2}_{\min}}{n_{\max}}, \label{eq-lem13-47-2}
	\end{align}
	where the first inequality follows from the fact that $\eta_k$ is a true change point, the second inequality from \eqref{eq-lem13-1}, the third inequality follows from \eqref{eq-lem13-2} and \eqref{eq-lem13-3}, and the final inequality follows from the condition that $0 < c_2 < c_1/2$.   Towards this end, it follows from Lemma~\ref{lem:local} that 
	\begin{equation}\label{eq-lem13-47-3}
		\Delta_{s,e}^{\eta_k}(z_0) - \Delta_{s,e}^{b}(z_0) \geq c |b - \eta_k| \delta \frac{n_{\min}^2}{n_{\max}^2} \Delta^{\eta_k}_{s, e}(z_0)(e-s)^{-2}.
	\end{equation}
	Combining \eqref{eq-lem13-47-1}, \eqref{eq-lem13-47-2} and \eqref{eq-lem13-47-3}, we have
	\begin{equation}\label{eq-lem13-47-rhs}
		\bigl\|\mu_{s, e} - \mathcal{P}^b_{s, e}\bigl(\mu_{s, e}\bigr)\bigr\|^2 - \bigl\|\mu_{s, e} - \mathcal{P}^{\eta_k}_{s, e}\bigl(\mu_{s, e}\bigr)\bigr\|^2 \geq \frac{cc_1^2}{4} \delta^2 \frac{n^5_{\min}}{n^4_{\max}} \kappa^2 (e-s)^{-2} |b - \eta_k|.
	\end{equation}
	
The left-hand side of \eqref{eq-lem13-pf-6} can be decomposed as follows.
	\begin{align}
		& 2\langle Y_{s, e}  - \mu_{s, e}, \mathcal{P}^b_{s, e}\bigl(Y_{s, e}\bigr) - \mathcal{P}^{\eta_k}_{s, e}\bigl(\mu_{s, e}\bigr)\rangle \nonumber \\
		= & 2 \langle Y_{s, e} - \mu_{s, e}, \mathcal{P}^b_{s, e}\bigl(Y_{s, e}\bigr) - \mathcal{P}^b_{s, e}\bigl(\mu_{s, e}\bigr) \rangle + 2 \langle Y_{s, e} - \mu_{s, e}, \mathcal{P}^b_{s, e}\bigl(\mu_{s, e}\bigr) - \mathcal{P}^{\eta_k}_{s, e}\bigl(\mu_{s, e}\bigr) \rangle \nonumber \\
		= & (I) + 2 \left(\sum_{i = 1}^{n_{s:\eta_k}} + \sum_{i = n_{s:\eta_k} + 1}^{n_{s:b}} + \sum_{i = n_{s:b} + 1}^{n_{s:e}} \right) \bigl(Y_{s, e} - \mu_{s, e}\bigr)_i \left(\mathcal{P}^b_{s, e}\bigl(\mu_{s, e}\bigr) - \mathcal{P}^{\eta_k}_{s, e}\bigl(\mu_{s, e}\bigr)\right)_i \nonumber \\
		= & (I) + (II.1) + (II.2) + (II.3).\label{eq-lem13-47-lhs0}
	\end{align}\\

\noindent \textbf{Term (I)}.  It holds that
	\begin{align}\label{eq-lem13-47-lhs1}
		(I) = \frac{2}{n_{s:b}}	\left\{\sum_{j=1}^{n_{s:b}}  \bigl(Y_{s, e} - \mu_{s, e}\bigr)_j\right\}^2 + \frac{2}{n_{(b+1):e}}	\left\{\sum_{j=n_{s:b} + 1}^{n_{s:e}}  \bigl(Y_{s, e} - \mu_{s, e} \bigr)_j\right\}^2 \leq 2 \gamma^2,
	\end{align}
	where the inequality follows from the definition of the CUSUM statistics and \eqref{eq-lem13-1}.  \\

\noindent \textbf{Term (II)}.  It holds that
	\begin{align*}
		(II.1) = 2\sqrt{n_{s:\eta_k}} \left\{\frac{1}{\sqrt{n_{s:\eta_k}}} \sum_{i = 1}^{n_{s:\eta_k}}\left(Y_{s, e} - \mu_{s, e}\right)_i\right\}	\left\{\frac{1}{n_{s:b}}\sum_{i = 1}^{n_{s:b}} (\mu_{s,e})_i - \frac{1}{n_{s:\eta_k}}\sum_{i=1}^{\eta_k} (\mu_{s,e})_i\right\}.
	\end{align*}
	In addition, it holds that
	\begin{align*}
		& \left|\frac{1}{n_{s:b}}\sum_{i = 1}^{n_{s:b}} (\mu_{s,e})_i - \frac{1}{n_{s:\eta_k}}\sum_{i=1}^{\eta_k} (\mu_{s,e})_i\right| = \frac{n_{(\eta_k+1):b}}{n_{s:b}}\left|-\frac{1}{n_{s:\eta_k}} \sum_{i = 1}^{n_{s:\eta_k}} (\mu_{s,e})_i + F_{\eta_k + 1}(z_0)\right| \\
		\leq & \frac{n_{(\eta_k+1):b}}{n_{s:b}} (C_S+1)\kappa_{s, e}^{\max},
	\end{align*}
	where the inequality follows from Lemma~\ref{Lem-17}	.  Combining with \eqref{eq-lem13-4}, it leads to that
	\begin{align}
		(II.1) & \leq 2\sqrt{n_{s:\eta_k}} \gamma \frac{n_{(\eta_k+1):b}}{n_{s:b}} (C_S+1)\kappa_{s, e}^{\max} \nonumber \\ 
		& \leq 2\frac{n_{\max}^{3/2}}{n_{\min}}\frac{4}{\min\{1, \, c_1^2\}}\delta^{-1/2} \gamma |b - \eta_k| (C_S+1) \kappa^{\max}_{s, e}. \label{eq-lem13-47-lhs21}
	\end{align}
	As for the term (II.2), it holds that
	\begin{align}\label{eq-lem13-47-lhs22}
		(II.2) \leq 2 \sqrt{n_{\max}} \sqrt{|b - \eta_k|} \gamma (2C_S + 3) \kappa^{\max}_{s, e}.
	\end{align}
	As for the term (II.3), it holds that 
	\begin{align}\label{eq-lem13-47-lhs23}
		(II.3) \leq 2\frac{n_{\max}^{3/2}}{n_{\min}}\frac{4}{\min\{1, \, c_1^2\}}\delta^{-1/2} \gamma |b - \eta_k| (C_S+1) \kappa^{\max}_{s, e}.
	\end{align}

Therefore, combining \eqref{eq-lem13-47-rhs}, \eqref{eq-lem13-47-lhs0}, \eqref{eq-lem13-47-lhs1}, \eqref{eq-lem13-47-lhs21}, \eqref{eq-lem13-47-lhs22} and \eqref{eq-lem13-47-lhs22}, we have that \eqref{eq-lem13-pf-6} holds if 
	\[
		\delta^2 \frac{n^5_{\min}}{n^4_{\max}} \kappa^2 (e-s)^{-2} |b - \eta_k| \gtrsim \max\left\{\gamma^2, \, \frac{n^{3/2}_{\max}}{n_{\min}} \delta^{-1/2} \gamma |b - \eta_k| \kappa, \, \sqrt{n_{\max}}\sqrt{|b - \eta_k|} \gamma \kappa\right\}.
	\]
	The second inequality holds due to Assumption~\ref{assumption-3}, the third inequality holds due to \eqref{eq-lem13-pf-contra} and the first inequality is a consequence of the third inequality and Assumption~\ref{assumption-3}.

\end{proof}

\section{Proofs of Section~\ref{sec-phase}}

\begin{proof}[Proof of Lemma~\ref{lemma-low-snr}]
Let $P_0$ denote the joint distribution of the independent random variables $\{Y_{t, i}\}_{i=1, t = 1}^{n, T}$, where
	\[
		Y_{1,1}, \ldots, Y_{\delta, n} \stackrel{i.i.d.}{\sim} \delta_0 \quad \text{and}  \quad Y_{\delta + 1, 1}, \ldots, Y_{T, n} \stackrel{i.i.d.}{\sim} \delta_1,
	\]
	where $\delta_c$, $c \in \mathbb{R}$, is the Dirac distribution having point mass at point $c$.  Similarly, let $P_1$ denote the joint distribution of the independent random variables $\{Z_{t, i}\}_{i=1, t = 1}^{n, T}$ such that
	\[
		Z_{1, 1}, \ldots, Z_{T - \delta, n} \stackrel{i.i.d.}{\sim} \delta_1 \quad \text{and}  \quad Z_{T - \delta + 1, 1}, \ldots, Z_{T, n} \stackrel{i.i.d.}{\sim} \delta_0.
	\]
	Observe that $\eta(P_0) = \delta$ and $\eta(P_1) = T - \delta$.  Since $\delta \leq T/3$, it holds that
	\[
		\inf_{\hat \eta} \sup_{P\in \mathcal P_{\zeta}^T} \mathbb{E}_P\bigl(|\hat \eta - \eta|\bigr) \geq (T/3)\bigl\{1- d_{\mathrm{TV}}(P_0, P_1)\bigr\} \geq (T/3) \bigl\{1 - 2\delta n\bigr\} \geq \frac{1 - 2\zeta^2}{3} T,		
	\]
	where $d_{\mathrm{TV}}(\cdot, \cdot)$ is the total variation distance.  In the last display, the first inequality follows from Le Cam's lemma \citep[see, e.g.][]{yu1997assouad}, and the second inequality follows from Eq.(1.2) in \cite{steerneman1983total}.

\end{proof}

\begin{proof}[Proof of Lemma~\ref{lemma-error-opt}]
Let $P_0$ denote the joint distribution of the independent random variables $\{Y_{t, i}\}_{i=1, t = 1}^{n, T}$, where
	\[
		Y_{1,1}, \ldots, Y_{\delta, n}  \stackrel{i.i.d.}{\sim}  F  \quad \text{and} \quad Y_{\delta + 1, 1}, \ldots, Y_{T, n}  \stackrel{i.i.d.}{\sim}  G;
	\]
	and, similarly, let $P_1$ be the joint distribution of the independent random variables $\{Z_{t, i}\}_{i=1, t = 1}^{n, T}$ such that
	\[
		Z_{1, 1}, \ldots, Z_{\delta + \xi, n}  \stackrel{i.i.d.}{\sim}  F, \quad \text{and} \quad Z_{\delta + \xi + 1, 1}, \ldots, Z_{T, n}  \stackrel{i.i.d.}{\sim}  G,
	\]
	where $\xi$ is a positive integer no larger than $n-1 - \delta$, 	
	\[
		F(x) = \begin{cases}
 			0, & x \leq 0,\\
 			x, & 0 < x \leq 1,\\
 			1, & x \geq 1,
 		\end{cases}
 		\quad \mbox{and} \quad
 		G(x) = \begin{cases}
 			0, & x \leq 0,\\
 			(1 - 2\kappa)x, & 0 < x \leq 1/2,\\
 			(1/2 - \kappa) + (1 + 2\kappa)(x-1/2), & 1/2 < x \leq 1,\\
 			1, & x \geq 1.
 		\end{cases}
	\]
	It is easy to check that 
	\[
		\sup_{z \in \mathbb{R}} |F(z) - G(z)| = \kappa.
	\]
	
Observe that $\eta(P_0) = \delta$ and $\eta(P_1) = \delta + \xi$.  By Le Cam's Lemma \citep[e.g.][]{yu1997assouad} and Lemma~2.6 in \cite{Tsybakov2009}, it holds that
	\begin{equation}\label{eq:second.lower}
		\inf_{\hat \eta} \sup_{P\in \mathcal{Q}} \mathbb{E}_P\bigl(|\hat \eta - \eta|\bigr)  \geq \xi \bigl\{1- d_{\mathrm{TV}}(P_0, P_1)\bigr\} \geq \frac{\xi}{2} \exp\left(-\mathrm{KL}(P_0, P_1)\right).
	\end{equation}

Since 
	\begin{align*}
	\mathrm{KL}(P_0, P_1) = \sum_{i \in \{\delta + 1, \ldots, \delta + \xi\}}	 \mathrm{KL}(P_{0i}, P_{1i}) = \frac{n\xi}{2} \log(1 - 4\kappa^2) \leq 2n\xi \kappa^2,
	\end{align*}
	we have 
	\[
		\inf_{\hat \eta} \sup_{P\in \mathcal{Q}} \mathbb{E}_P\bigl(|\hat \eta - \eta|\bigr) \geq \frac{\xi}{2}\exp(-2n\xi\kappa^2).
	\]

Next, set $\xi = \min \{ \lceil \frac{1}{n\kappa^2} \rceil, T - 1 - \delta\}$. 
By the assumption on $\zeta_T$, for all $T$ large enough we must have that $\xi = \lceil \frac{1}{ n\kappa^2} \rceil$.
 Thus, for all $T$ large enough, using \eqref{eq:second.lower},
\[
\inf_{\hat \eta} \sup_{P\in \mathcal Q} \mathbb{E}_P\bigl(|\hat \eta - \eta|\bigr) \geq   \max \left\{ 1, \frac{1}{2} \Big\lceil\frac{1}{n\kappa^2} \Big\rceil e^{-2} \right\}.
\]
	
\end{proof}

\section{ Proof of Theorem \ref{thm-full}  }
\label{sec-full}

\begin{proof}
	It follows from Theorem~\ref{thm-wbs} and the proof thereof that applying Algorithm~\ref{algorithm:WBS} to $\{W_{t, i}\}$ and the $\tau$ sequence defined in \eqref{eq-tau-candidates}, with probability at least
	\[
	1 -\frac{24 \log(n_{1:T})}{T^3 n_{1:T}} - \frac{48 T}{n_{1:T} \log(n_{1:T}) \delta} - \exp\left\{\log\left(\frac{T}{\delta}\right) - \frac{S\delta^2}{16T^2}\right\},
	\]
	the event $\mathcal{A}$, which is defined as follows holds.
	\begin{itemize}
		\item [$\mathcal{A}1$] if $\tau > c_{\tau, 2}  \kappa \delta^{1/2} \frac{n^{3/2}_{\min}}{n_{\max}}$, then the corresponding change point estimators satisfying $\widehat{K} < K$, but for any $\hat{\eta}$ in the estimator set, there exits $k \in \{1, \ldots, K\}$ such that
		\[
		|\hat{\eta} - \eta_k| \leq C_{\epsilon}\kappa^{-2}_k\log (n_{1:T}) n^{9}_{\max}n^{-10}_{\min};
		\]
		\item [$\mathcal{A}2$] if $c_{\tau, 2}  \kappa \delta^{1/2} \frac{n^{3/2}_{\min}}{n_{\max}} \geq \tau \geq c_{\tau, 1} \sqrt{\log(n_{1:T})}$,	then the corresponding change point estimators satisfying $\widehat{K} = K$, and for any $\hat{\eta}$ in the estimator set, there exits $k \in \{1, \ldots, K\}$ such that
		\[
		|\hat{\eta} - \eta_k| \leq C_{\epsilon}\kappa^{-2}_k\log (n_{1:T}) n^{9}_{\max}n^{-10}_{\min};
		\]
		\item [$\mathcal{A}3$] if $\tau < c_{\tau, 1} \sqrt{\log(n_{1:T})}$, then the corresponding change point estimators satisfying $\widehat{K} > K$, and for any true change point $\eta_k$, there exits $\hat{\eta}$ in the estimators such that
		\[
		|\hat{\eta} - \eta_k| \leq C_{\epsilon}\kappa^{-2}_k\log (n_{1:T}) n^{9}_{\max}n^{-10}_{\min}.
		\]
	\end{itemize}
	The rest of the proof is conducted conditionally on the event $\mathcal{A}$.  
	
	Note that different $\tau_j$'s may return the same collections of the change point estimators.  For simplicity, in the rest of the proof, we assume that distinct candidate $\tau_j$'s in \eqref{eq-tau-candidates} return distinct and nested $\mathcal{B}_j$ with $|\mathcal{B}_j| = K_j$. \\
	
	\noindent \textbf{Step 1.}  Let $\hat{\eta}_0 = 0$ and $\hat{\eta}_{K+1} = T$.  In this step, it suffices to show that for any $k \in \{0, \ldots, K-1\}$, it holds that with large probability
	\begin{align}\label{eq-SIC-pf-1}
	\sum_{l = k}^{k+1} \sum_{t = \hat{\eta}_l + 1}^{\hat{\eta}_{l+1}} \sum_{i = 1}^{n_t}\left(\mathbbm{1}_{\{Y_{t, i} \leq \hat{z}\}} - \widehat{F}^Y_{(\hat{\eta}_l + 1): \hat{\eta}_{l+1}}(\hat{z})\right)^2 + \lambda < \sum_{t = \hat{\eta}_k + 1}^{\hat{\eta}_{k+2}} \sum_{i = 1}^{n_t}\left(\mathbbm{1}_{\{Y_{t, i} \leq \hat{z}\}} - \widehat{F}^Y_{(\hat{\eta}_k + 1): \hat{\eta}_{k+2}}(\hat{z})\right)^2.
	\end{align}
	Without loss of generality, we consider the case when $k = 0$.
	
	Note that with probability at least
	\[
	1 -\frac{24 \log(n_{1:T})}{T^3 n_{1:T}} - \frac{48 T}{n_{1:T} \log(n_{1:T}) \delta} - \exp\left\{\log\left(\frac{T}{\delta}\right) - \frac{S\delta^2}{16T^2}\right\},
	\]
	it holds that
	\begin{align}
	& \sum_{t = 1}^{\hat{\eta}_2} \sum_{i = 1}^{n_t} \left(\mathbbm{1}_{\{Y_{t, i} \leq \hat{z}\}} - \widehat{F}_{1: \hat{\eta}_{2}}(\hat{z})\right)^2 - \sum_{l = 0}^{1} \sum_{t = \hat{\eta}_l + 1}^{\hat{\eta}_{l+1}} \sum_{i = 1}^{n_t}\left(\mathbbm{1}_{\{Y_{t, i} \leq \hat{z}\}} - \widehat{F}_{(\hat{\eta}_l + 1): \hat{\eta}_{l+1}}(\hat{z})\right)^2  \nonumber \\
	= &  \left(D^{\hat{\eta}_1}_{1, \hat{\eta}_2}(\{Y_{t, i}\})\right)^2 \geq c_{\tau, 2}^2  \kappa^2 \delta \frac{n^{3}_{\min}}{n_{\max}^2}, \label{eq-SIC-pf-2}
	\end{align}
	where the last inequality follows from the proof of Theorem~\ref{thm-wbs}.
	
	Therefore, for $\lambda = C\log(n_{1:T})$, \eqref{eq-SIC-pf-1} holds due to Assumption~\ref{assumption-3}. \\
	
	\noindent \textbf{Step 2.}  In this step, we are to show with large probability, Algorithm~\ref{alg:full} will not over select.  For simplicity, assume $\mathcal{B}_2 = \{\hat{\eta}_1\}$ and $\mathcal{B}_1 = \{\hat{\eta}, \hat{\eta}_1\}$ with $0 < \hat{\eta} < \hat{\eta}_1$.  Let $\hat{z}$ be the one defined in Algorithm~\ref{alg:full} using the triplet $\{0, \hat{\eta}, \hat{\eta}_1\}$.  
	
	Since
	\begin{align*}
	& \sum_{t = 1}^{\hat{\eta}_{1}} \sum_{i = 1}^{n_t}\left(\mathbbm{1}_{\{Y_{t, i} \leq \hat{z}\}} - \widehat{F}^Y_{1: \hat{\eta}_{1}}(\hat{z})\right)^2 - \sum_{t = 1}^{\hat{\eta}} \sum_{i = 1}^{n_t}\left(\mathbbm{1}_{\{Y_{t, i} \leq \hat{z}\}} - \widehat{F}^Y_{1: \hat{\eta}}(\hat{z})\right)^2 - \sum_{t = \hat{\eta} + 1}^{\hat{\eta}_{1}} \sum_{i = 1}^{n_t}\left(\mathbbm{1}_{\{Y_{t, i} \leq \hat{z}\}} - \widehat{F}^Y_{(\hat{\eta}+1): \hat{\eta}_{1}}(\hat{z})\right)^2 \\
	& = \left(D^{\hat{\eta}}_{0, \hat{\eta}_1}(\{Y_{t, i}\})\right)^2 \leq  c^2_{\tau, 1} \log(n_{1:T})
	\end{align*}
	holds with probability at least
	\[
	1 -\frac{24 \log(n_{1:T})}{T^3 n_{1:T}} - \frac{48 T}{n_{1:T} \log(n_{1:T}) \delta} - \exp\left\{\log\left(\frac{T}{\delta}\right) - \frac{S\delta^2}{16T^2}\right\}.
	\]
	Therefore, for $\lambda = C\log(n_{1:T})$, \eqref{eq-SIC-pf-1} holds. \\
	
	Combining both steps above and the fact that these two steps are conducted in the event $\mathcal{A}$, we have that	
	\begin{align}
	& \mathbb{P}\left\{\widehat{K} = K \quad \mbox{and} \quad \epsilon_k \leq C_{\epsilon}\kappa^{-2}_k\log (n_{1:T}) n^{9}_{\max}n^{-10}_{\min}, \, \forall k = 1, \ldots, K\right\} \nonumber\\
	\geq & 1 - \frac{48 \log(n_{1:T})}{T^3 n_{1:T}} - \frac{96 T}{n_{1:T} \log(n_{1:T}) \delta} - \exp\left\{\log\left(\frac{T}{\delta}\right) - \frac{S\delta^2}{16T^2}\right\}. \nonumber
	\end{align} 
\end{proof}

\bibliographystyle{abbrvnat} 
\bibliography{references}

\end{document}